\documentclass[11pt,oneside,letter]{article}

\usepackage{amsthm, amsmath, amssymb, amsfonts, graphicx, epsfig}
\usepackage{accents}
\usepackage{bbm}
\usepackage{mathrsfs}
\usepackage{epstopdf}
\usepackage{graphicx}
\usepackage[font=sl,labelfont=bf]{caption}
\usepackage{subcaption}
\usepackage{float}
\restylefloat{table}
\usepackage{slashbox}
\usepackage{enumerate}
\usepackage{comment}
\usepackage[colorlinks=true, pdfstartview=FitV, linkcolor=blue,
            citecolor=blue, urlcolor=blue]{hyperref}
\usepackage[usenames]{color}
\usepackage{url}
\makeatletter\def\url@leostyle{%
 \@ifundefined{selectfont}{\def\UrlFont{\sf}}{\def\UrlFont{\scriptsize\ttfamily}}} \makeatother
\newtheorem{theorem}{Theorem}
\newtheorem{proposition}[theorem]{Proposition}
\newtheorem{lemma}[theorem]{Lemma}
\newtheorem{corollary}[theorem]{Corollary}

\theoremstyle{definition}
\newtheorem{definition}[theorem]{Definition}
\newtheorem{example}[theorem]{Example}
\theoremstyle{remark}
\newtheorem{remark}[theorem]{Remark}

\numberwithin{equation}{section}
\numberwithin{theorem}{section}
\definecolor{Red}{rgb}{0.9,0,0.0}
\definecolor{blue}{rgb}{0,0.0,1.0}
\def\cD{\mathcal{D}}
\def\cE{\mathcal{E}}
\def\cF{\mathcal{F}}

\def\cH{\mathcal{H}}

\def\cL{\mathcal{L}}
\def\cM{\mathcal{M}}
\def\cN{\mathcal{N}}

\def\cS{\mathcal{S}}
\def\cT{\mathcal{T}}

\def\cX{\mathcal{X}}

\def\bE{\mathbb{E}}
\def\bF{\mathbb{F}}

\def\bN{\mathbb{N}}

\def\bP{\mathbb{P}}
\def\bQ{\mathbb{Q}}
\def\bR{\mathbb{R}}

\def\bZ{\mathbb{Z}}

\newcommand{\1}{\mathbbm{1}}            % preferable way of writing indicator function
\newcommand{\set}[1]{\{#1\}}            % set: {xyz} to be used for inline formulas
\renewcommand{\mid}{\;|\;}              % mid bar with small spaces before and after: x | y
\newcommand{\bsde}{BS$\Delta$E}
\newcommand{\bsdes}{BS$\Delta$Es}
\newcommand{\ask}{\text{ask}}
\newcommand{\bid}{\text{bid}}

\DeclareMathOperator*{\esssup}{ess\,sup} % ess sup
\DeclareMathOperator*{\essinf}{ess\,inf} % ess inf

\DeclareMathOperator{\var}{\mathrm{V}@\mathrm{R}}           % \V@R Value-at-risk
\usepackage{hyperref}

\title{Dynamic Conic Finance via Backward Stochastic Difference Equations}
\author{Tomasz R. Bielecki\\[-0.3ex]
\url{bielecki@iit.edu} \\[-0.3ex]
\and
Igor Cialenco\\[-0.3ex]
\url{igor@math.iit.edu} \\[-0.3ex]
\and
Tao Chen\\[-0.3ex]
\url{tchen29@hawk.iit.edu}\\[-0.3ex]
\and
\small{Department of Applied Mathematics,}\\[-0.3ex]
\small{Illinois Institute of Technology,}\\[-0.3ex]
\small{Chicago, 60616 IL, USA}\\[-0.3ex]
}
\date{First Circulated: December 18, 2014}

\begin{document}
\maketitle
%\tableofcontents
\begin{abstract}
\noindent
We present an arbitrage free theoretical framework for modeling bid and ask prices of dividend paying securities in a discrete time setup using theory of dynamic acceptability indices.   In the first part of the paper we develop the theory of dynamic subscale invariant performance measures, on a general probability space, and discrete time  setup. We prove a representation theorem of such measures in terms of a family of dynamic convex risk measures, and provide a representation of dynamic risk measures in terms of g-expectations, and solutions of \bsdes \ with convex drivers. We study the existence and uniqueness of the solutions, and derive a comparison theorem for corresponding \bsdes.
 In the second part of the paper we discuss a market model for dividend paying securities by introducing the pricing operators that are defined in terms of dynamic acceptability indices, and find various properties of these operators. Using these pricing operators,  we define the bid and ask prices for the underlying securities and then for derivatives in this market.  We show that the obtained market model is arbitrage free, and we also prove a series of properties of these prices.

\bigskip
{\noindent \small

 }

%\newpage

\end{abstract}
%\tableofcontents

\section{Introduction}
The main goal of this paper is to contribute to the arbitrage free theoretical framework for modeling bid and ask prices of dividend paying securities in a discrete time setup, that was  originated in Bielecki~et~al.~\cite{BCIR2012}. As in \cite{BCIR2012} we follow the \textit{conic finance} methodology, initiated by Cherny and Madan~\cite{Madan2010}. The idea behind (dynamic) conic finance is to use \textit{(dynamic) coherent acceptability indices} to define bid/ask prices in the spirit of no-good-deal method proposed by Cochrane and Saa-Requejo~\cite{Cochrane2000}. We recall that a coherent acceptability index \cite{ChernyMadan2009} is a function $\alpha$ defined on the set of bounded random variables $L^\infty$ with values in $[0,+\infty]$, that is monotone increasing, scale-invariant, quasi-concave and additionally satisfies a continuity (Fatou) property.
The elements of $L^\infty$ can be viewed as terminal values of given portfolios or payoffs, and in this case the coherent acceptability index is a measure of performance of these portfolios, and essentially it is a generalization of the well known measures of performance such as Sharpe Ratio or Gain-to-Loss Ratio. We refer the reader to the original paper \cite{ChernyMadan2009} for  a detailed discussion of economical and financial aspects of these measures, their robust representation and various examples. In \cite{BCZ2010,BiaginiBion-Nadal2012} the authors study the dynamic version of coherent acceptability indices.  Similar to static case, it was shown that any dynamic coherent acceptability index $\alpha_t$ can be uniquely  characterized by an increasing  family of dynamic coherent risk measures\footnote{We recall that a coherent risk measure is a function $\rho:L^\infty\to\overline{\bR}$, that is monotone decreasing, subaditive, homogenous and cash-additive \cite{ArtznerDelbaenEberHeath1999}.  For dynamic version of coherent risk measures see for instance \cite{Riedel2004,ArtznerDelbaenEberHeathKu2007,BCZ2010}.} $\set{\rho_t^x}_{x\in\bR}$.

In the (static) conic finance framework, the time zero ask price $p^\textrm{ask}$ of a future payoff $X$ is defined as follows:
\begin{equation} \label{eq:intro1}
p^\textrm{ask} = \inf \set{a\in\bR \mid  \alpha(a-X^*)\geq \gamma},
\end{equation}
where $\gamma$ is a fixed pre-specified level of acceptability, and $X^*$ is the discounted value of $X$. That is, the ask price is the smallest amount of cash paid today, such that the combined cashflow $a-X^*$ is acceptable at least at level $\gamma$. Moreover, it can be shown that the ask price admits the representation $p^\textrm{ask}=\rho^\gamma(-X^*)$, which is useful for computational purposes. The bid price $p^\textrm{bid}$ is defined similarly.  Additionally, one can easily incorporate in this definition the possibility to hedge part of the risk by trading in an underlying market. Also in \cite{Madan2010}, it was shown that the conic finance pricing framework does not admit arbitrage, and can be used as a tool to shrink the arbitrage--free price interval. The extension of conic finance to multiperiod markets is quite delicate, especially if the underlying securities pay dividends and bear transaction costs themselves. The main challenges are due to the fact that the wealth process associated with a self-financing trading strategy is not a linear functional of trading strategies.
The dynamic conic finance theory was first studied in \cite{BCIR2012}. This was done for the case of discrete time and finite probability space.
There, the authors also investigated the connection between dynamic conic finance framework and classical arbitrage theory, based on the arbitrage theory for the corresponding markets developed in \cite{BieleckiCialencoRodriguez2012}.

Although the dynamic conic finance theory is a flexible nonlinear pricing framework it does not fully capture the liquidity risk. More precisely, due to scale invariance of the dynamic coherent acceptability indices, the bid/ask prices are homogeneous in the number of shares traded. However, the typical market phenomenon is that  the more shares one buys the higher price per share one pays; similarly, more share one sells, lower price per share is received. It turns out, as first observed in  \cite{RosazzaGianinSgarra2012}, that replacing the scale invariance postulate by sub-scale invariance yields a pricing framework that captures the liquidity charge describe above (see also \cite{BionNadal2009a}). Accordingly, these authors develop a `dynamic conic finance' framework generated by sub-scale invariant acceptability indices. They consider a continuous time setup for pricing terminal payoffs defined on a general probability space.  Similar to original conic finance case \cite{Madan2010}, the authors derive a representation theorem for bid/ask prices in terms of convex risk measures, and consequently in terms of solutions of some Backward Stochastic Differential Equations (BSDEs) and g-expectations.

This paper builds upon the ideas described above, and provides a fairly general theoretical pricing framework that can be applied to a large class of financial markets.
The paper is divided into two major parts, and the main contributions of this paper can be summarized as follows:
\begin{itemize}
  \item In the first part we develop the theory of dynamic subscale invariant performance measures, on a general probability space, and is the discrete time  setup. We prove a representation theorem of such measures in terms of a family of dynamic convex risk measures. Moreover, we also provide a representation of dynamic risk measures in terms of g-expectations, and solutions of \bsdes \ with convex drivers. We investigate the existence and uniqueness of the solutions, and provide a comparison theorem for corresponding \bsdes.
  \item The second part discusses a market model for dividend paying securities. We introduce the pricing operators that are defined in terms of dynamic acceptability indices, and find various properties of these operators. Using these pricing operators,  we define the bid and ask prices for the underlying securities and then for derivatives in this market.  We show that the obtained market model is arbitrage free, and we also prove a series of important, and desired,  properties of these prices.
\end{itemize}

The paper is organized as follows.  Section~2 is devoted to Backward Stochastic Difference Equations (\bsdes). Although the theory of Backward Stochastic Differential Equations is a mature field (cf. \cite{MaYong1999Book}), there are only few papers \cite{Stadje2009,CohenElliott2011,CohenElliott2009a,CheriditoStadje2013} devoted to their discrete counterpart.
Even though the existing results on \bsdes\ are quite general, they do not meet our specific needs, which prompted us to establish existence and uniqueness of the solutions for a (large) class of \bsdes \ relevant to our needs. As in the existing literature, we define the non-linear expectation, or g-expectation, in terms of the solution of a \bsde, and we  prove a series of standard properties of g-expectations. The  results obtained in this section are new, although general ideas of the proofs are similar to those in \cite{Stadje2009,CohenElliott2011,CheriditoStadje2013}.

In Section~3 we define and study the notion of Dynamic Acceptability Index (DAI), which, essentially, is a Dynamic Coherent Acceptability Index \cite{BCZ2010} with scale invariance property replaced with sub-scale invariance. It turns out that in this case a DAI can be generated by a family of dynamic convex risk measures, that are  introduced and studied in this section. Special attention is paid to time consistency of these risk and performance measures, which is conveniently  used in the sequel. Similar to \cite{RosazzaGianin2006,RosazzaGianinSgarra2012} we show that a dynamic convex risk measure can be characterized by a solution of \bsde \ with convex driver or, equivalently,  by the corresponding g-expectation.

Even though the results from Section~2 and Section~3 serve as background for next two sections, they are also of independent  interest, and they can be applied to other areas, such as  stochastic processes and mathematical finance. For example, the risk and performance measure can be used in risk management, and the established results are natural continuation of the axiomatic theory of risk and performance measures.

Section~4 is devoted to dynamic conic finance. In Section~4.1, we start by defining  a market model consisting of a banking account and $K$ securities (or assets). We assume that these $K$ assets may pay dividends, and that the amount of dividend paid can be different depending  on whether one holds a long or a short position. The banking account is the only asset that trades with no transaction costs. The prices of the securities are given by bid and ask \textit{pricing operators}. The pricing operators are not necessarily homogeneous (of degree one) in the number of shares traded. We present two examples of such assets -- a commonly traded stock and a Credit Default Swap (CDS), and we describe in detail the corresponding dividend streams and pricing operators. We introduce the relevant financial definitions in this market model, such as value process, self-financing trading strategy, and arbitrage. The main challenge comes from the fact the value process is not a linear functional with respect to trading strategies, and as a consequence, one needs to define appropriately the notion of self-financing trading strategy and that of arbitrage, by taking into account the set-up value and liquidation value  associated with trading strategy (see Section~4.2).
Next, in Section~4.3, we define some of the main objects of this paper -- \textit{the acceptability bid and ask prices} -- by using the dynamic acceptability indices. The general idea is similar to \eqref{eq:intro1}. We fix a dynamic acceptability index $\alpha$ and an acceptability level $\gamma$ and we define the acceptability ask/bid price as the minimum/maximum amount added to the dividend stream today that makes the combined cash-flow acceptable at level $\gamma$. Applying the results from Sections~2 and 3, we provide a representation of acceptability bid and ask prices in terms of nonlinear expectations associated with the corresponding family of convex drivers.
Also here, we prove a series of fundamental properties of these prices, such as that bid is always smaller than ask, convexity or concavity in number of shares traded, appropriate form of time consistency, etc. In Section~4.4, using the acceptability bid and ask prices, we build a market model that satisfies all the necessary properties postulated in Section~4.1, and we show that this market model is arbitrage free, in the sense of Section~4.2. We prove that bid price of a given dividend stream is always smaller than the ask price, even if one uses different dynamic acceptability indices or/and acceptability levels for bid and ask. In Section~5 we discuss the valuation of derivatives within this pricing framework.
We start with definitions of super-hedging cash-flows, and of dynamic acceptability bid and ask prices of derivatives, by taking into account the possibility to hedge using the underlying securities. In spirit of \cite{BCIR2012} we introduce the notion of \textit{no good deal} and prove that if there are no good deals, then the acceptability prices in the extended market are arbitrage free prices. Similar to Section~4, we provide a series of results about the proposed pricing framework. In particular we give sufficient conditions under which the bid price is equal to ask price and the trades will take place given that both parties are using conic finance pricing methodology. It should be mentioned that time consistency of dynamic acceptability indices plays a central role in developing the dynamic conic finance proposed here.

For convenience, a number of technical results are deferred to the Appendix.

Finally, we want to mention that we provide illustrative examples of the results obtained in the paper throughout, starting with concrete examples of drivers for \bsdes, that give examples of risk measures and acceptability indices. We also discuss several examples of market models. Due to already lengthy manuscript, the numerical implementation of the results from this paper, as well as application of them to real market data, will be investigated in a sequel.

\section{Backward Stochastic Difference Equations}\label{se:bsde}

Let $T$ be a fixed and finite time horizon, and let $\cT:=\{0,1,\ldots,T\}$.
We consider a filtered probability space $(\Omega, \cF, \{\cF_t\}_{t=0}^T, \bP)$, with $\cF_0=\set{\emptyset, \Omega}$ and $\cF=\cF_T$.
Throughout, we will use the notations $L^p(\cF_t):=L^2(\Omega,\cF_t,\bP)$, $p\geq 1, \ t\in\cT$. Also, we will denote by $\cX$ the set of all adapted and square integrable stochastic processes on $(\Omega, \cF, \{\cF_t\}_{t=0}^T, \bP)$.
We reserve the notation $\Delta$ for the backward difference operator $\Delta X_t:=X_t-X_{t-1}$, $t\in\cT$, where $X$ is a stochastic process, and we also take the convention $\Delta X_0:=X_0$. In what follows, all equalities and inequalities will be understood in $\bP$-almost surely sense.
We recall that the predictable quadratic variation $\langle X\rangle_t$ of a stochastic process $X$ is defined as a predictable process, starting at zero, and such that $X^2_t-\langle X\rangle_t$ is a martingale with respect to filtration $\set{\cF_t}$. It can be shown that $\Delta\langle X\rangle_t=\bE[(\Delta X_t)^2|\cF_{t-1}]$.

In the sequel, the function $g:\cT\times\Omega\times\bR\rightarrow\bR$ will play the role of a driver for considered Backward Stochastic Difference Equations (\bsdes), and we will assume that it satisfies\\[.05in]
\textbf{Assumption A:}\\[-0.6cm]
\begin{enumerate}[\textrm{A}1.]
\item
the mapping $(t,\omega)\mapsto g(t,\omega,z)$ is predictable for any $z\in\bR$;
\item
the function $z\mapsto g(t,\omega,z)$ is uniformly Lipschitz continuous, i.e. there exists $c_t(\omega)\in L^\infty(\cF_{t-1})$, such that for any $z_1,z_2\in\bR, t\in \cT$,
$$
|g(t,\omega,z_1)-g(t,\omega,z_2)|\leq c_t(\omega)|z_1-z_2|;
$$
\item
$g(t,\omega,0)=0$ for any $t\in\cT$.
\end{enumerate}
Throughout the paper we will denote by $c$ the Lipschitz constant of function $g$, that is
$$
c_t(\omega)=\essinf\Big\{l_t(\omega)\in L^\infty(\cF_{t-1}):|g(t,\omega,z_1)-g(t,\omega,z_2)|\leq l_t(\omega)|z_1-z_2|,\,\omega\in\Omega,\,z_1,\,z_2\in\bR\Big\},
$$
for any $t\in\cT$.
Also, we will suppress the  explicit dependence on $\omega$, if no confusion arise; for example, we may write $g(t,z)$ instead of $g(t,\omega,z)$.

We consider the following Backward Stochastic Difference Equation (\bsde),
\begin{equation}\label{eq:bsde}
  Y_t=Y_T+\sum_{t<s\leq T}g(s,Z_s)\Delta\langle W\rangle_s-\sum_{t<s\leq T}Z_s\Delta W_s+M_T-M_t, \quad t\in\cT,
\end{equation}
with terminal condition $Y_T\in L^2(\cF_T)$, and where $W_t$ is a fixed square integrable martingale process with independent increments, and such that $\Delta\langle W\rangle_t\neq0$ for any $t\in\cT$.
As already mentioned, the function $g$ is usually referred to as the driver of the \bsde\ \eqref{eq:bsde}.

As one may expect, due to its `backward' nature, and similar to continuous time BSDEs, a solution is a triple of processes, rather than just an adapted process. Next, we give the precise definition of a solution of \bsde \ \eqref{eq:bsde}.

\begin{definition}
A solution to BS$\Delta$E \eqref{eq:bsde} is a triple of processes $(Y,Z,M)$ such that: $(Y_t,Z_t,M_t)\in L^2(\cF_t)\times L^2(\cF_{t-1})\times L^2(\cF_t)$, it  satisfies equality \eqref{eq:bsde} for all $t\in\cT$, and $M$ is a martingale process strongly orthogonal\footnote{We say that the process $M$ is strongly orthogonal to $W$ if the process $W_tM_t$ is a martingale.} to $W$.
\end{definition}

In general, for fixed $Z$ and $M$, $Y$ that satisfies \eqref{eq:bsde} is not necessarily an adapted process.
For example, by taking the terminal condition $X\in\cF_T$, and putting $Z_t=M_t=0$ for all $t\in\cT$, then $Y_t=X$ for any $t\in\cT$. In this case, $Y$ is not an adapted process unless $\cF_t=\cF_T$, $t\in\cT$.
However, we will show that due to Galtchouk-Kunita-Watanabe decomposition (cf.~\cite[Theorem~10.18]{FollmerSchiedBook2004}), there exists $Z$ and $M$ along with $Y$ such that they $(Y,Z,M)$ is a solution of \eqref{eq:bsde}.
We recall that $W$ is said to have the predictable representation property, if $\cF=\cF^W$ implies that for any square integrable martingale process $X$, there exists a predictable process $Z$, such that $\Delta X_t=Z_t\Delta W_t$.
It can be shown that if $\cF=\cF^W$ and $W$ has predictable representation property, then $M_t=0$ for all $t\in\cT$.
If $W$ does not satisfy the predictable representation property, then a martingale process $M$ which is orthogonal to $W$ is indeed needed to ensure that the solution of \eqref{eq:bsde} is well defined.

With slide abuse of notations, sometimes will refer to process $Y$ as solution of \bsde \ \eqref{eq:bsde}, rather than saying process $Y$ from the solution $(Y,Z,M)$.

Without going into technical details, we want to mention that \eqref{eq:bsde} essentially corresponds to a discrete version of a Backward Stochastic Differential Equation (BSDE) of the form
\begin{equation}\label{eq:BSDE-cont}
-dY_t=g(t,Y_t,Z_t)dt-Z_tdB_t,\quad Y_T=\xi, \quad t\geq 0,
\end{equation}
where $\xi$ is a square integrable terminal condition, $g$ is a progressively measurable function (also called the driver), and $B$ is Brownian motion.
Existence, uniqueness, as well as numerous properties, of the solution of such BSDEs  are well studied and understood (cf.~\cite{PardouxPeng1990,Peng1997,ElKarouiPengQuenez1997,Briand2000}).
If, it is assumed that $W$ from \eqref{eq:bsde} is a symmetric random walk such that $\bP(\Delta W_t=\pm1)=\frac{1}{2}$, then the solution of \eqref{eq:bsde} will converge weakly to the solution of \eqref{eq:BSDE-cont} as $\Delta t\rightarrow0$.
Rather than taking \bsdes\ as an approximations of BSDEs, we will study the properties of \bsdes\ in their own right, and use them as fundamental tools to study dynamic risk measures and dynamic acceptability indices.
As mentioned above, the theory for discrete-time counterpart of \eqref{eq:BSDE-cont} is still a developing field, and one goal of this paper is to establish some fundamental properties of solutions of these type of equations tailored to our needs.
Similar work on \ \bsdes \ has been done in \cite{CohenElliott2009a}, \cite{CohenElliott2011}, \cite{CheriditoStadje2013}, and \cite{Stadje2009}, but unlike them, the driver $g$ considered in our paper does not depend on $Y_t$. One reason to consider these type of drivers comes from the fact that
the cash-additivity and convexity properties of dynamic risk measures generated by the solution of \eqref{eq:bsde} imply that the driver $g$ does not dependent on $Y$, and hence we focus only on drivers of the form $g(t,z)$ in the first place.

\subsection{Existence and Uniqueness of Solutions}

In this section we will present a general result on existence and uniqueness of the solution of \bsde \ \eqref{eq:bsde}.

\begin{theorem}\label{th:exist}
  Assume that the driver $g$ satisfies Assumption~A, and that the terminal condition $Y_T\in L^2(\cF_T)$. Then,
  there exists a unique solution of equation \eqref{eq:bsde}.
\end{theorem}

\begin{proof} First, we will prove the existence using backward induction argument.
  Given $Y_T\in L^2(\cF_T)$, we consider the equation
  \begin{equation}\label{eq:onestep@T}
  Y_{T-1}=Y_T+g(T,Z_T)\Delta\langle W\rangle_T-Z_T\Delta W_T+\Delta M_T,
  \end{equation}
  where $Z_T, \Delta M_T$ and $Y_{T-1}$ are the unknowns.
  Since $Y_T\in L^2(\cF_T)$, then
  $$
  \bE[\bE[Y_T|\cF_{T-1}]^2]\leq\bE[\bE[(Y_T)^2|\cF_{T-1}]]=\bE[(Y_T)^2]<\infty.
  $$
  Hence, $Y_T-\bE[Y_T|\cF_{T-1}]$ is a square integrable martingale difference, so it admits the Galtchouk-Kunita-Watanabe decomposition, which implies that there exist $Z_T\in\cF_{T-1}$, $Z_T\Delta W_T\in L^2(\cF_T)$, $\Delta M_T\in L^2(\cF_T)$ such that $\bE[\Delta M_T|\cF_{T-1}]=0$, $\bE[\Delta M_T\Delta W_T|\cF_{T-1}]=0$ and
  \begin{equation}\label{eq:YT}
  Y_T-\bE[Y_T|\cF_{T-1}]=Z_T\Delta W_T-\Delta M_T.
  \end{equation}
  We multiply both sides of last identity by $\Delta W_T$, and then apply $\bE[\ \cdot\ |\cF_{T-1}]$ to both sides. This implies that
  $\bE[Y_T\Delta W_T|\cF_{T-1}]=Z_T\Delta\langle W\rangle_T$.
  Therefore,
  \begin{align}\label{eq:ZT}
  Z_T=\frac{\bE[Y_T\Delta W_T|\cF_{T-1}]}{\Delta\langle W\rangle_T}.
  \end{align}
  Since $W$ has independent increments, then $\Delta\langle W\rangle_T=\bE[\Delta W^2_T|\cF_{T-1}]=\bE[\Delta W^2_T]=:C$, and by our initial assumption $C\neq0$.
  Hence, we deduce
  \begin{align*}
    \bE[Z_T^2]&=\frac{\bE[\bE[Y_T\Delta W_T|\cF_{T-1}]^2]}{C^2}\leq\frac{\bE[C\bE[Y^2_T|\cF_{T-1}]]}{C^2}=\frac{\bE[Y^2_T]}{C}<\infty.
  \end{align*}
  With $Z_T$ and $\Delta M_T$ known, taking into account \eqref{eq:YT} and \eqref{eq:onestep@T}, we conclude that $Y_{T-1}$ must be given by
  \begin{align}\label{eq:YT-1}
  Y_{T-1}=\bE[Y_T|\cF_{T-1}]+g(T,Z_T)\Delta\langle W\rangle_T.
  \end{align}
From here,   due to A1 and Proposition~\ref{pr:predictable}, we get that $Y_{T-1}\in\cF_{T-1}$.
  Also, using A2, A3 and the fact that $Z_T\in L^2(\cF_{T-1})$, we have
  $$
  \bE[(g(T,Z_T)\Delta\langle W\rangle_T)^2]=\bE[g(T,Z_T)^2\bE[\Delta\langle W\rangle_T^2|\cF_{T-1}]]\leq C^2\|c_T\|_\infty^2\bE[Z_T^2]<\infty\,
  $$
  and thus $Y_{T-1}\in L^2(\cF_{T-1})$.
  Therefore, we determined $Y_{T-1}$, $Z_T$ and $\Delta M_T$.

 %Next, having $Y_{T-1}\in L^2(\cF_{T-1})$, using the same arguments, we find $(Y_{T-2}, Z_{T-1}, \Delta M_{T-1})$.
 We continue this backward procedure for any finite number of steps smaller than $T$: having $Y_{t+1}\in L^2(\cF_{t+1})$ for some fixed $t\in\set{0,1,\ldots,T-1}$, by similar arguments as above, we find $(Y_{t}, Z_{t+1}, \Delta M_{t+1})\in L^2(\cF_t)\times L^2(\cF_{t})\times L^2(\cF_{t+1})$, such that
  \begin{equation}\label{eq:onestepbsde}
    Y_t=Y_{t+1}+g(t+1,Z_{t+1})\Delta\langle W\rangle_{t+1}-Z_{t+1}\Delta W_{t+1}+\Delta M_{t+1},\quad t=0,\ldots,T-1,
  \end{equation}
with
  \begin{align}
    Y_{t+1}-\bE[Y_{t+1}|\cF_t]& =Z_{t+1}\Delta W_{t+1}-\Delta M_{t+1}, \label{eq:gkwdecomp} \\
  Z_{t+1} & =\frac{\bE[Y_{t+1}\Delta W_{t+1}|\cF_t]}{\Delta\langle W\rangle_{t+1}}, \label{eq:zs} \\
    Y_t & =\bE[Y_{t+1}|\cF_t]+g(t+1,Z_{t+1})\Delta\langle W\rangle_{t+1} \label{eq:ys}.
  \end{align}
  By taking the convention $Z_0=0, M_0=0$, and letting $M_t:=M_0+\sum^t_{s=1}\Delta M_s$, we have that \eqref{eq:bsde} holds true for all $t\in\cT$.
  Moreover, $M$ is a square integrable martingale process. Finally, since
  \begin{align*}
    \bE[M_tW_t|\cF_{t-1}]&=\bE[(\sum^t_{s=1}\Delta M_s)W_t|\cF_{t-1}]\\
    &=\sum^{t-1}_{s=1}\Delta M_s\bE[W_t|\cF_{t-1}]+\bE[\Delta M_t(W_{t-1}+\Delta W_t)|\cF_{t-1}]\\
    &=M_{t-1}W_{t-1}, \quad t=1,\ldots,T,
  \end{align*}
we conclude that $M$ is strongly orthogonal to $W$, which concludes the proof of existence of the solution.

  Next, we will prove the uniqueness.
  Assume there are two solutions $(Y^1_t,Z^1_t,M^1_t)$ and $(Y^2_t,Z^2_t,M^2_t)$, $t\in\cT$, of BS$\Delta$E \eqref{eq:bsde} with terminal condition $Y_T$.
  Then
  \begin{align} \label{eq:uniq2}
    Y^1_{T-1}-Y^2_{T-1}=(g(T,Z^1_T)-g(T,Z^2_T))\Delta\langle W\rangle_T - (Z^1_T-Z^2_T)+\Delta M^1_T-\Delta M^2_T.
  \end{align}
  From \eqref{eq:zs}, we have that $Z^1_T=\frac{\bE[Y_T\Delta W_T|\cF_{T-1}]}{\Delta\langle W\rangle_T}=Z^2_T$.
  Hence, by \eqref{eq:gkwdecomp}, we get
  $$
  \Delta M^1_T=-Y_T+\bE[Y_T|\cF_{T-1}]+Z^1_T\Delta W_T=-Y_T+\bE[Y_T|\cF_{T-1}]+Z^2_T\Delta W_T=\Delta M^2_T.
  $$
From here, in view of \eqref{eq:uniq2}, we immediately conclude that $Y^1_{T-1}=Y^2_{T-1}$.
Inductively, and by using the convention that $Z^1_0=Z^2_0=0$, $M^1_0=M^2_0=0$, we get that $(Y^1_t,Z^1_t,M^1_t)=(Y^2_t,Z^2_t,M^2_t)$, for any $t\in\cT$.
  Therefore, the solution is unique, and this concludes the proof.
\end{proof}

\begin{remark}
  Throughout, we will make intensive use of the recurrent expressions \eqref{eq:onestepbsde}, \eqref{eq:gkwdecomp}, \eqref{eq:zs} and \eqref{eq:ys} that actually determine (or
  define) the unique solution of \bsde.

\end{remark}

\begin{remark}
 Note that, if $\cF=\cF^W$ and if $W$ has the predictable representation property, then, in view of \eqref{eq:gkwdecomp}, we have that $\Delta M_t=0$, and since $M_0=0$, we conclude that $M_t=0$, $t\in\cT$, in \eqref{eq:bsde}. Saying differently, if $\cF=\cF^W$ and $W$ has the predictable representation property
then there exists a pair of processes $(Y_t,Z_t)$, $t\in\cT$, that is the unique solution of equation
  $$
  Y_t=Y_T+\sum_{t<s\leq T}g(s,Z_s)\Delta\langle W\rangle_s-\sum_{t<s\leq T}Z_s\Delta W_s,\ \quad t\in\cT.
  $$

  One important example of martingale $W$ that has the predictable representation property is the symmetric random walk (see Proposition~\ref{pr:randomwalk}).
\end{remark}

\subsection{Comparison Results}

We will prove a series of comparison type results about the solutions of \bsdes.
These results, besides being of fundamental importance for the theory of \bsdes \  itself, will also serve as key ingredients for describing the risk measures developed later on in the paper.
More precisely, the version of the comparison theorem provided here is tailored to our needs and it will be used to prove the monotonicity property of proposed risk measures.

We start with an auxiliary result.

\begin{lemma}\label{le:linearcomparison}
Consider \bsde \ \eqref{eq:bsde}, and assume that the driver $g$ satisfies Assumption~A, and that the terminal condition $Y_T\geq0$.
Also, suppose that for a fixed $t\in\cT$,  $g(s,z)=x_sz$, $s\in\{t,\ldots,T\}$, where $x$ is such that $1+x_s\Delta W_s>0$, for any $s\in\{t,\ldots,T\}$.
Then,  $Y_s\geq0$ for all $s\in\{t,\ldots,T\}$.
Moreover, if $Y_t=0$ on $A\in\cF_t$, then  $Y_s=0$ on $A$, for all $s\in \{t,\ldots,T\}$.
\end{lemma}

\begin{proof}
  First note that Assumption~A and Theorem \ref{th:exist} guarantee that the solution $(Y,Z,M)$ of \eqref{eq:bsde} exists.
  Fix $t\in\cT$, assume that $g(s,z)=x_sz, s\in\{t,\ldots,T\}$ and $1+x_s\Delta W_s>0, s\in\{t,\ldots,T\}$.
  Then,  by \eqref{eq:ys} and \eqref{eq:zs},
  \begin{align*}
    Y_{s-1}&=\bE[Y_s|\cF_{s-1}]+x_sZ_s\Delta\langle W\rangle_s\\
    &=\bE[Y_s|\cF_{s-1}]+x_s\frac{\bE[Y_s\Delta W_s|\cF_{s-1}]}{\Delta\langle W\rangle_s}\Delta\langle W\rangle_s\\
    &=\bE[Y_s|\cF_{s-1}]+x_s\bE[Y_s\Delta W_s|\cF_{s-1}].
  \end{align*}
Recall that $g$ is predictable, and since $g(s,z)=x_sz$, we have that $x_s$ is $\cF_{s-1}$-measurable.
  Thus,
  \begin{equation}\label{eq:ys2}
  Y_{s-1}=\bE[Y_s(1+x_s\Delta W_s)|\cF_{s-1}].
  \end{equation}
  From here, if $Y_s\geq0$ for some $s\in\{t+1,\ldots,T\}$, using the assumption that $1+x_s\Delta W_s>0$, we get that  $Y_{s-1}\geq0$.
  Hence, since $Y_T\geq0$, we conclude that  $Y_s\geq0$ for all $s\in\{t,\ldots,T\}$.

If $\1_AY_t=0$ for some $A\in\cF_t$, then, by the above, $\1_AY_s\geq 0$ for all $s\in\set{t,\ldots,T}$.
  Moreover, by \eqref{eq:ys2} we also have $\1_AY_t=\bE[\1_AY_{t+1}(1+x_{t+1}\Delta W_{t+1}|\cF_t]=0$, and since $1+x_{t+1}\Delta W_{t+1}0$, we get $\1_AY_{t+1}(\omega)=0$.
  Similarly, we deduce that $\1_AY_s=0, s\in\{t+2,\ldots,T\}$.

  This concludes the proof.
\end{proof}

Lemma~\ref{le:linearcomparison} depicts the comparison result for drivers $g$ of the form $g(t,z)=x_tz$.
Using this result, we will prove next the comparison theorem for a general BS$\Delta$E.

\begin{theorem}\label{th:comp}
  Assume that  $g^1$, $g^2$ satisfy Assumption A, and $Y^1_T, \ Y^2_T\in L^2(\cF_T)$, and suppose that for every $t\in\cT$, the following conditions hold true:
  \begin{enumerate}[1)]
  \item
  $Y^1_T\geq Y^2_T$;
  \item
  $g^1(s,z)\geq g^2(s,z)$, $s\in\{t,\ldots,T\}$, $z\in\bR$;
  \item
  $|c^1_s\Delta W_s|<1$, $s\in\{t,\ldots,T\}$, where $c^1$ is the Lipschitz coefficient of $g^1$ as defined in A2.
  \end{enumerate}
  Denote by $Y^i, i=1,2$, the solution of \eqref{eq:bsde}, that corresponds to driver $g^i$, and terminal condition $Y_T^i$, for $i=1,2$.
   Then, $Y^1_s\geq Y^2_s$, for all $s\in\{t,\ldots,T\}$.
  Moreover, the comparison is strict, in the sense that if $Y^1_t=Y^2_t$ on $A\in\cF_t$, then $\1_AY^1_s=\1_AY^2_s$ and $\1_Ag^1(s,Z^2_s)=\1_Ag^2(s,Z^2_s)$, for all $s\in\{t,\ldots,T\}$.
\end{theorem}

\begin{proof}
  We prove by backward induction that $Y^1_t\geq Y^2_t$, for every $t\in\cT$.
  By Assumption~1), the statement holds true for $t=T$.
  Assume that  $Y^1_t\geq Y^2_t$ for some fixed $t\in\cT$.
  Then,  by \eqref{eq:onestepbsde}
  \begin{align*}
    Y^1_{s-1}-Y^2_{s-1}=&Y^1_s-Y^2_s+(g^1(s,Z^1_s)-g^2(s,Z^2_s))\Delta\langle W\rangle_s-(Z^1_s-Z^2_s)\Delta W_s+\Delta (M^1_s-M^2_s)\\
 =&\Big[Y^1_s-Y^2_s+(g^1(s,Z^1_s)-g^1(s,Z^2_s))\Delta\langle W\rangle_s-(Z^1_s-Z^2_s)\Delta W_s+\Delta (M^1_s-M^2_s)\Big]\\
  &+\Big[(g^1(s,Z^2_s)-g^2(s,Z^2_s))\Delta\langle W\rangle_s\Big]\\
  =:&\widetilde{Y}_{s-1}+\overline{Y}_{s-1}.
  \end{align*}
  Clearly, Assumption~2) implies that  $\overline{Y}_{s-1}\geq0$.
  As for $\widetilde{Y}_{s-1}$, we write it as
     \begin{align}\label{eq:ytilde}
    \widetilde{Y}_{s-1}
    &=Y^1_s-Y^2_s+\frac{g^1(s,Z^1_s)-g^1(s,Z^2_s)}{Z^1_s-Z^2_s}(Z^1_s-Z^2_s)\Delta\langle W\rangle_s-(Z^1_s-Z^2_s)\Delta W_s+\Delta (M^1_s-M^2_s),
  \end{align}
  where, as usually, $0/0=0$.
    Let us now define  $\widetilde{g}(t,z)=\frac{g^1(t,Z^1_t)-g^1(t,Z^2_t)}{Z^1_t-Z^2_t}z$, for $t\in \cT$ and $z\in \bR$.
  Since $g^1(t,z)$ satisfies Assumption A, then by Proposition~\ref{pr:predictable}, $\widetilde{g}(t,z)$ is predictable,  and by Assumption~A2
  $$
  |\widetilde{g}(t,z_1)-\widetilde{g}(t,z_2)|=|\frac{g^1(t,Z^1_t)-g^1(t,Z^2_t)}{Z^1_t-Z^2_t}||z_1-z_2|\leq c^1_t|z_1-z_2|.
  $$
  Moreover, $\widetilde{g}(t,0)=0$, and hence $\widetilde{g}(t,z)$ satisfies Assumption~A, and $\widetilde{Y}_{s-1}$ is the solution to BS$\Delta$E \eqref{eq:ytilde} with driver $\widetilde{g}(t,z)$ and terminal condition $Y^1_s-Y^2_s$.
Since $|\frac{g^1(t,Z^1_t)-g^1(t,Z^2_t)}{Z^1_t-Z^2_t}| \leq c^1_t$, in view of Assumption~3), $|\frac{g^1(t,Z^1_t)-g^1(t,Z^2_t)}{Z^1_t-Z^2_t} \Delta W_t| < 1$, and thus
  $1+\frac{g^1(t,Z^1_t)-g^1(t,Z^2_t)}{Z^1_t-Z^2_t}\Delta W_t>0$.
  From here, using Lemma~\ref{le:linearcomparison}, we get that $\widetilde{Y}_{s-1}\geq0$.

From the above arguments, we have that $Y^1_{s-1}-Y^2_{s-1}=\widetilde{Y}_{s-1}+\overline{Y}_{s-1}\geq0$, and consequently, by induction argument $Y^1_s\geq Y^2_s$, $s\in\{t,\ldots,T\}$.

  Finally, if $\1_AY^1_t=\1_AY^2_t$, for some $A\in\cF_t$, then $\1_A\widetilde{Y}_t=\1_A\overline{Y}_t=0$.
  By Lemma~\ref{le:linearcomparison},  $\1_AY^1_{t+1}=\1_AY^2_{t+1}$.
  Since $\1_A\overline{Y}_t=0$ and $\1_Ag^1(t+1,Z^2_{t+1})\geq \1_Ag^2(t+1,Z^2_{t+1})$, then $\1_Ag^1(t+1,Z^2_{t+1})=\1_Ag^2(t+1,Z^2_{t+1})$.
  Similarly, for $t+1<s\leq T$, one gets that $\1_AY^1_s=\1_AY^2_s$ and $\1_Ag^1(s,Z^2_s)=\1_Ag^2(s,Z^2_s)$.

  The proof is complete.
\end{proof}

\begin{corollary}\label{co:comp2}
  Let $g$ be a driver that satisfies Assumption~A, and let $Y^1_T, Y^2_T\in L^2(\cF_T)$ be two terminal conditions such that $Y^1_T\geq Y^2_T$
  Assume that $|c_t\Delta W_t|<1$, $t\in\cT$, where $c_t$ is the Lipschitz coefficient of $g$.
  Then, $Y^1_t\geq Y^2_t$, $t\in\cT$.
  Moreover, the comparison is strict in the sense that if $\1_AY^1_t=\1_AY^2_t$, for some $t\in\cT$, $A\in\cF_t$, then $\1_AY^1_s=\1_AY^2_s$, $s=t,\ldots,T$.
\end{corollary}

\begin{remark}
Using the same ideas, one can show that Theorem~\ref{th:comp} holds true if Assumption~3) is replaced by the following assumption:
  \begin{align}\label{eq:3prime}
    g^1\left(s,\frac{\bE[Y^1\Delta W_s|\cF_{s-1}]}{\Delta\langle W\rangle_s}\right)-g^1\left(s,\frac{\bE[Y^2\Delta W_s|\cF_{s-1}]}{\Delta\langle W\rangle_s}\right)\geq\frac{\bE[Y^2-Y^1|\cF_{s-1}]}{\Delta\langle W\rangle_s},
  \end{align}
  for $Y^1, Y^2\in L^2(\cF_s), Y^1\geq Y^2$, $t< s\leq T$, and the equality reached if and only if $Y^1=Y^2$.

Assumption \eqref{eq:3prime} is weaker than assumption 3) in Theorem~\ref{th:comp}.
  Indeed, assuming  that $|c^1_s\Delta W_s|<1$, and $Y^1\geq Y^2$,  we then have that
  \begin{align*}
    \frac{\bE[Y^1-Y^2|\cF_{s-1}]}{\Delta\langle W\rangle_s}&\geq\frac{\bE[|c^1_s\Delta W_s|(Y^1-Y^2)|\cF_{s-1}]}{\Delta\langle W\rangle_s}\\
    &\geq c^1_s\Big|\frac{\bE[\Delta W_s(Y^1-Y^2)|\cF_{s-1}]}{\Delta\langle W\rangle_s}\Big|=c^1_s\Big|\frac{\bE[Y^1\Delta W_s|\cF_{s-1}]}{\Delta\langle W\rangle_s}-\frac{\bE[Y^2\Delta W_s|\cF_{s-1}]}{\Delta\langle W\rangle_s}\Big|\\
    &\geq \Big|g^1\Big(s,\frac{\bE[Y^1\Delta W_s|\cF_{s-1}]}{\Delta\langle W\rangle_s}\Big)-g^1\Big(s,\frac{\bE[Y^2\Delta W_s|\cF_{s-1}]}{\Delta\langle W\rangle_s}\Big)\Big|.
  \end{align*}
  Hence, inequality \eqref{eq:3prime} holds true.

We feel that assumption \eqref{eq:3prime} is less intuitive, and cumbersome, and for sake of ease of exposition, in what follows we will use assumption 3) from Theorem~\ref{th:comp}.
\end{remark}

\begin{remark}
  In \cite{Stadje2009}, the author proves a comparison result for \bsdes \ in the limit sense, as time step goes to zero.
  It was assumed that that the driver $g$ and the martingale process $W$ satisfy are such that
  \begin{align*}
    &|g(t,z_1)-g(t,z_2)|\leq K(1+(|z_1|_\infty\vee|z_2|_\infty)^{\alpha/2})|z_1-z_2|_\infty;\\
    &\lim_{\Delta t\rightarrow0}\|\Delta W_t\|_\infty\Delta\langle W\rangle_t^{-(\alpha/4)}=0,\ \alpha\in[0,2),
  \end{align*}
  along with some additional technical conditions.
In this case, by taking $\alpha=0$, and small enough $\Delta t$, we note that $g$ and $W$ will satisfy Assumption B2 from our setup, and hence  our comparison result holds true.
  In either \cite{Stadje2009} or our setup, it is required to have control on both the driver $g$ and the noise $W$ for the comparison theorem to hold.
Finally we want to mention that while the conditions from  \cite{Stadje2009} and conditions proposed in this work overlap (in some sense), neither one implies the other.
With $\alpha=0$, the condition on the driver from  \cite{Stadje2009} is stronger; while for $\alpha>0$, the assumption on the driver is weaker but the condition satisfied by $W$ is more restricted.
\end{remark}

In this work, we will mostly work with drivers that satisfy the comparison principle, and for brevity we will call such drivers regular, with precise definition as follows:
\begin{definition}
  Let $g$ be a driver that satisfies Assumption~A, and let $Y^1_T, Y^2_T\in L^2(\cF_T)$ be two terminal conditions.
  \begin{enumerate}
    \item
    We say that the \textit{comparison result holds true} for \bsde\ with driver $g$, if $Y^1_T\geq Y^2_T$ implies that $Y^1_t\geq Y^2_t$, $t\in\cT$, and \textit{the comparison is strict} if $\1_AY^1_t=\1_AY^2_t$ for some $t\in\cT$, and $A\in\cF_t$, then $\1_AY^1_s=\1_AY^2_s$, $s=t,\ldots,T$.
    \item
    The driver $g$ is called \textit{regular driver} if the comparison result holds true for \bsde\ with driver~$g$.
  \end{enumerate}
\end{definition}

Lemma~\ref{le:linearcomparison} implies that a linear driver $g(t,z)=x_tz$ is regular if $1+x_t\Delta W_t>0$, for every $t\in\cT$.
Also, Corollary~\ref{co:comp2} implies that a general driver $g$ is regular if $|c_t\Delta W_t|<1$ for every $t\in\cT$.
Next, we present several examples of regular drivers, where we consider $W$ as a martingale process such that $\Delta W_t$ is uniformly bounded.

\begin{example}\label{ex:coherentdriver}
Let the driver $g(t,z)=c_t|z|$, with $c$ such that $\|c_t\|_\infty<\frac{1}{\|\Delta W_t\|_\infty}$, is a regular driver.
We will show in the next section that such driver generates a family of coherent dynamic risk measures.
\end{example}

\begin{example}\label{ex:convexdriver}
  Let us put
  $$
  g(t,z)=\frac{K}{(K+1)\|\Delta W_t\|_\infty}\ln(\frac{1}{3}+\frac{1}{3}e^{-z}+\frac{1}{3}e^{z}),
  $$
  where $K\in\bR^+$ is fixed.
  As such, $g(t,z)$ is predictable, $g(t,0)=0$ and $g(t,z)$ is Lipschitz due to the fact that its derivative with respect to $z$ takes value in $(-\frac{K}{(K+1)\|\Delta W_t\|_\infty},\frac{K}{(K+1)\|\Delta W_t\|_\infty})$.
Moreover, the Lipschitz coefficient $c_t$ is such that $|c_t\Delta W_t|<1$, according to the fact that $|\frac{\partial g(t,z)}{\partial z}|\leq\frac{K-1}{K\|\Delta W_t\|_\infty}$.
Thus, the driver $g$ is regular, and the corresponding BSDE has the following form
$$
Y_t=Y_T+\sum_{t<s\leq T}\frac{K}{(K+1)\|\Delta W_t\|_\infty}\ln(\frac{1}{3}+\frac{1}{3}e^{-Z_s}+\frac{1}{3}e^{Z_s})\Delta\langle W\rangle_t-\sum_{t<s\leq T}Z_s\Delta W_s+M_T-M_t.
$$
We will see in the next section, this \bsde\ plays an important role in our study, and it is related to so called convex dynamic risk measures.
\end{example}

\subsection{g-Expectations}

In the seminal paper \cite{Peng1997}, the author introduced a relationship between solutions of BSDEs (in continuous time) and so called nonlinear expectations or g-expectations (see also \cite{CoquetHuMeminPeng2002}). Later, the theory of nonlinear expectations was successfully applied to some problems from mathematical finance in the context of theory of risk measures. For more details we refer the reader to \cite{RosazzaGianin2006,Barrieu2007IP,RosazzaGianinSgarra2012} and references therein.
Similar approach can be adopted for the case of BS$\Delta$Es, which will be the main goal of this section.

Towards this end,  we will assume that the driver $g$ is a regular driver, and that the terminal condition $X\in L^2(\cF_T)$.
As before, we denote by $(Y_t, Z_t, M_t)$, $t\in\cT$, the solution of the corresponding \bsde.
Analogous to the existing literature on BSDEs (cf. \cite{Peng1997}), we define \textit{the conditional $g$-expectation} $\cE_g\big[X\big|\cF_t\big]$ of a random variable $X$ given $\cF_t$, the random variable given by $\cE_g\big[X\big|\cF_t\big]:= Y_t$.

In what follows, it will be convenient to view the space $L^{2}(\cF_T)$ as an $L^\infty(\cF_t)$-module, for every fixed $t\in\cT$; saying differently, the random variables from $L^\infty(\cF_t)$ will play the role of scalers for the linear space $L^{2}(\cF_T)$. For more details on general theory of $L^0$-modules, and their relationship to theory of risk and performance measures,  we refer the reader to \cite{FilipovicKupperVogelpoth2009,KupperVogelpoth2009,BCDK2013}.

\begin{remark}\label{le:gexpprime}
Throughout the paper we will use the following result, which follows immediately  from the uniqueness of the solutions and backward nature of \bsdes.
  Let $g$ be a driver that satisfies Assumption~A, and assume that there exists a triple of processes $(Y',Z',M')$ such that
  $$
  Y'_u=X+\sum_{u<s\leq T}g(s,Z'_s)\Delta\langle W\rangle_s-\sum_{u<s\leq T}Z'_s\Delta W_s+M'_T-M'_u, \quad u=t,\ldots,T,
  $$
  where $Y'_u\in L^2(\cF_u)$, $Z'_u\in L^2(\cF_{u-1})$, $M'_u\in L^2(\cF_u)$, and $\bE[M'_uW_u|\cF_{u-1}]=M'_{u-1}W_{u-1}$.
  Then, $Y'_t=\cE_g[X|\cF_t]$.
\end{remark}

Next result provides some fundamental properties of $g$-expectations, such as monotonicity, tower property, convexity when driver is convex, and homogeneity when driver is homogenous.

\begin{proposition}\label{pr:gexp4}
For any regular driver $g$, the conditional $g$-expectation satisfies the following properties:
  \begin{itemize}
    \item[(i)]
    $\cE_g[\mu|\cF_t]=\mu$, for any $\mu\in\bR$, $t\in\cT$;
    \item[(ii)]
    if $X^1\geq X^2$, $X^1,X^2\in L^2(\cF_T)$, then $\cE_g\big[X^1\big|\cF_t\big]\geq\cE_g\big[X^2\big|\cF_t\big]$, for any $t\in\cT$.
    Moreover, if $\1_A\cE_g\big[X^1\big|\cF_t\big]=\1_A\cE_g\big[X^2\big|\cF_t\big]$, for some $t\in\cT$, and $A\in\cF_t$, then $\1_AX^1=\1_AX^2$;
    \item[(iii)]
    $\cE_g\big[\cE_g\big[X\big|\cF_s\big]\big|\cF_t\big]=\cE_g\big[X\big|\cF_{s\wedge t}\big]$, for any $X\in L^2(\cF_T)$, $s,t\in\cT$;
    \item[(iv)]
    $\cE_g\big[\1_AX\big|\cF_t\big]=\1_A\cE_g\big[X\big|\cF_t\big]$, for any $X\in L^2(\cF_T)$, $A\in\cF_t$, $t\in\cT$;
    \item[(v)]
    $\cE_g\big[X+\beta\big|\cF_t\big]=\cE_g\big[X\big|\cF_t\big]+\beta$, for any $X\in L^2(\cF_T)$, $\beta\in L^2(\cF_t)$, $t\in\cT$;
    \item[(vi)]
    if $g(u,\omega,\cdot)$ is convex, for any $u\in\{t,\ldots,T\}$, $\omega\in\Omega$, that is
    $$
    g(u,\omega,\mu z_1+(1-\mu)z_2)\leq\mu g(u,\omega,z_1)+(1-\mu)g(u,\omega,z_2), \quad z_1,\,z_2\in\bR,\, \mu\in\bR,\, 0\leq\mu\leq1,
    $$
    then
    $$
    \cE_g[\lambda X^1+(1-\lambda)X^2|\cF_t]\leq\lambda\cE_g[X^1|\cF_t]+(1-\lambda)\cE_g[X^2|\cF_t],
    $$
    for any $X^1, X^2\in L^2(\cF_T)$, $\lambda\in L^\infty(\cF_t)$, $0\leq\lambda\leq1$.
    \item[(vii)]
    if $g(u,\omega,\cdot)$ is homogeneous, for any $u\in\{t,\ldots,T\}$, $\omega\in\Omega$, that is
    $$
    g(u,\omega,\mu z)=\mu g(u,\omega,z), \quad z,\,\mu\in\bR,
    $$
    then
    $$
    \cE_g[\lambda X|\cF_t]=\lambda\cE_g[X|\cF_t], \quad X\in L^2(\cF_T),\, \lambda\in L^\infty(\cF_t).
    $$
  \end{itemize}
\end{proposition}

\begin{proof}

  (i) If $Y_T=\mu\in\bR$, then according to \eqref{eq:ZT}, we get that
    $$
    Z_T=\frac{\bE[Y_T\Delta W_T|\cF_{T-1}]}{\Delta\langle W\rangle_T}=\mu\frac{\bE[\Delta W_T|\cF_{T-1}]}{\Delta\langle W\rangle_T}=0.
    $$
    Hence, by \eqref{eq:YT-1} and Assumption~A3, we have that
    $$
    Y_{T-1}=\bE[Y_T|\cF_{T-1}]+g(T,Z_T)\Delta\langle W\rangle_T=\bE[Y_T|\cF_{T-1}]=\mu.
    $$
    Inductively, in view of \eqref{eq:zs} and \eqref{eq:ys}, we have that $Y_t=\mu$ for all $t\in\cT$.
    Hence, $\cE_g[\mu|\cF_t]=\mu$, $t\in\cT$.

        \smallskip

    (ii) It follows imediatly from Theorem~\ref{th:comp}.

       \smallskip

    (iii) Assume that $t\leq s$, and let $(Y,Z,M)$ be the solution of \bsde\ with terminal condition $X$.
    Then,
    \begin{align*}
      &Y_u=X+\sum_{u<r\leq T}g(r,Z_r)\Delta\langle W\rangle_r-\sum_{u<r\leq T}Z_r\Delta W_r+M_T-M_u, \quad u\in\cT.
    \end{align*}
    By considering $u=t$ and $u=s$, we immediately  get
    \begin{align}\label{eq:ytys}
    Y_t=Y_s+\sum_{t<r\leq s}g(r,Z_r)\Delta\langle W\rangle_r-\sum_{t<r\leq s}Z_r\Delta W_r+M_s-M_t.
    \end{align}
    Next, we consider the \bsde\ with driver $g$, terminal time $s$ and terminal condition $Y_s$.
    By \eqref{eq:ytys} and definition of $g$-expectation, we have that $Y_t=\cE_g[Y_s|\cF_t]$, which implies that $\cE_g[X|\cF_t]=\cE_g[\cE_g[X|\cF_s]|\cF_t]$.

    Now, let us assume that $t> s$.
    For $t=T$, then property follows from the definition of $g$-expectation
    For $t<T$, we us consider the \bsde\ with driver $g$ and terminal condition (at time $T$) $\cE_g[X|\cF_s]$.
       In view of \eqref{eq:ZT}, we have that
    $$
    Z_T=\frac{\bE[Y_T\Delta W_T|\cF_{T-1}]}{\Delta\langle W\rangle_T}=\cE_g[X|\cF_s]\frac{\bE[\Delta W_T|\cF_{T-1}]}{\Delta\langle W\rangle_T}=0.
    $$
    Hence, by \eqref{eq:YT-1} and Assumption~A3, it is true that
    $$
    Y_{T-1}=\bE[Y_T|\cF_{T-1}]+g(T,Z_T)\Delta\langle W\rangle_T=\bE[\cE_g[X|\cF_s]|\cF_{T-1}]=\cE_g[X|\cF_s].
    $$
    Inductively, using \eqref{eq:zs} and \eqref{eq:ys}, we conclude that $Y_t=\cE_g[X|\cF_s]$.
    Consequently, by the definition of $g$-expectation, we finally get that
    $\cE_g[\cE_g[X|\cF_s]|\cF_t]=\cE_g[X|\cF_s].$

      \smallskip
        (iv)
    Fix $t\in\cT$, $X\in L^2(\cF_T), \ A\in\cF_t$, and let $(Y,Z,M)$ be the solution of BS$\Delta$E with terminal condition $X$.
    Note that $\1_Ag(u,Z_u)=g(u,\1_AZ_u)$, $u=t+1,\ldots,T$, and thus
   \begin{align}\label{eq:1ayu}
    \1_AY_u=\1_AX+\sum_{u<s\leq T}g(s,\1_AZ_s)\Delta\langle W\rangle_s-\sum_{u<s\leq T}\1_AZ_s\Delta W_s+\1_AM_T-\1_AM_u,
    \end{align}
    for every $u=t,\ldots,T$.
    Also note that since $A\in\cF_t$, and $u\geq t+1$, we have that $\1_AY_u\in L^2(\cF_u)$, $\1_AZ_u\in L^2(\cF_{u-1})$, and $\1_AM_u\in L^2(\cF_u)$, for any $u=t+1,\ldots,T$.
    Due to that fact that $M$ and $W$ are orthogonal, we note that
    $$
    \bE[\1_AM_uW_u|\cF_{u-1}]=\1_A\bE[M_uW_u|\cF_u]=\1_AM_{u-1}W_{u-1}.
    $$
    Therefore, in view of \eqref{eq:1ayu}, and Lemma~\ref{le:gexpprime}, we have that $1_AY_t=\cE_g[\1_AX|\cF_t]$, which implies that $\1_A\cE_g[X|\cF_t]=\cE_g[\1_AX|\cF_t]$.

        \smallskip
    (v) If $(Y_u,Z_u,M_u), u\in\cT$, be the solution of BS$\Delta$E with terminal condition $X$, then
    \begin{align}\label{eq:yuplusbeta}
    Y_u+\beta=X+\beta+\sum_{u<s\leq T}g(s,Z_s)\Delta\langle W\rangle_s-\sum_{u<s\leq T}Z_s\Delta W_s+M_T-M_u, \quad u=t,\ldots,T,
    \end{align}
   for any $\beta\in L^2(\cF_t)$.
    Clearly $Y_u+\beta\in L^2(\cF_u)$, and hence, by Lemma~\ref{le:gexpprime} and \eqref{eq:yuplusbeta}, we conclude that $Y_t+\beta=\cE_g[Y_T+\beta|\cF_t]$, which implies that
    \begin{align*}
      \cE_g[X|\cF_t]+\beta=X+\beta=\cE_g[X+\beta|\cF_t].
    \end{align*}

    \smallskip
    (vi)
    Let $(Y^1_u,Z^1_u,M^1_u)$, and respectively $Y^2_u,Z^2_u,M^2_u$, $u\in\cT$, be the solution of BS$\Delta$E with terminal condition $X^1$, and respectively $X^2$.
    Assuming that $g(u,\omega,\cdot)$ is convex, and using identity \eqref{eq:bsde}, we have that

     \begin{align}
      \lambda Y^1_t+(1-\lambda)Y^2_t =& \lambda X^1+(1-\lambda)X^2+\lambda\sum_{t<u\leq T}g(u,\omega,Z^1_u)\Delta\langle W\rangle_u \nonumber \\
      &+(1-\lambda)\sum_{t<u\leq T}g(u,\omega,Z^2_u)\Delta\langle W\rangle_u-\lambda\sum_{t<u\leq T}Z^1_u\Delta W_u \nonumber  \\
      &-(1-\lambda)\sum_{t<u\leq T}Z^2_u\Delta W_u+\lambda(M^1_T-M^1_t)+(1-\lambda)(M^2_T-M^2_t) \nonumber  \\
      \geq&\lambda X^1+(1-\lambda)X^2+\sum_{t<u\leq T}g(u,\omega,\lambda Z^1_u+(1-\lambda)Z^2_u)\Delta\langle W\rangle_u \nonumber  \\
      &-\sum_{t<u\leq T}(\lambda Z^1_u+(1-\lambda)Z^2_u)\Delta\langle W\rangle_u+\lambda M^1_T+(1-\lambda)M^2_T \nonumber \\
      &-\lambda M^1_t-(1-\lambda)M^2_t, \label{eq:aux0}
    \end{align}
    where $\lambda\in L^\infty(\cF_t)$, $0\leq\lambda\leq1$.
    Next we consider the process $Y'$ defined as follows
    \begin{align*}
      Y'_u=&\lambda X^1+(1-\lambda)X^2+\sum_{u<s\leq T}g(s,\lambda Z^1_s+(1-\lambda)Z^2_s)\Delta\langle W\rangle_s\\
      &-\sum_{u<s\leq T}(\lambda Z^1_s+(1-\lambda)Z^2_s)\Delta\langle W\rangle_s+\lambda M^1_T+(1-\lambda)M^2_T\\
      &-\lambda M^1_u-(1-\lambda)M^2_u, \quad u=t,\ldots,T.
    \end{align*}
    Clearly $Y'_u\in L^2(\cF_u)$, $\lambda Z^1_u+(1-\lambda)Z^2_u\in L^2(\cF_{u-1})$, $\lambda M^1_u+(1-\lambda)M^2_u\in L^2(\cF_u)$, and $\bE[(\lambda M^1_u+(1-\lambda)M^2_u)W_u|\cF_{u-1}]=(\lambda M^1_{u-1}+(1-\lambda)M^2_{u-1})W_{u-1}$, for any $u=t+1,\ldots,T$.
    Therefore, by Lemma~\ref{le:gexpprime}, we have that $Y'_t=\cE_g[\lambda Y^1_T+(1-\lambda)Y^2_T|\cF_t]$, combined with \eqref{eq:aux0} concludes the proof.

        \smallskip
    (vii)
        The proof is similar to the proof of (vi) and we omit it here.

\end{proof}

In what follows, we will call a driver  $g$ convex, if $g(t,\omega,\cdot)$ is convex, and  $g$ positive homogeneous, if $g(t,\omega,\cdot)$ is positive homogeneous, for any $t\in\cT$, $\omega\in\Omega$.
Also, we will simple say that $\cE_g[\ \cdot\ |\cF_t]$ is convex (rather than convex in $L^0$-module sense) if
  $$
  \cE_g[\lambda X^1+(1-\lambda)X^2|\cF_t]\leq\lambda\cE_g[X^1|\cF_t]+(1-\lambda)\cE_g[X^2|\cF_t],
  $$
  for any $\lambda\in L^\infty(\cF_t)$, $0\leq\lambda\leq1$.

  Proposition~\ref{pr:gexp4} shows that $g$-expectation (or nonlinear expectation) shares many properties with usual conditional expectation.
  However, as name suggests, generally speaking it is not linear.  Next two results show that  the $g$-expectation is linear if and only if the driver is regular and linear.

\begin{proposition}\label{pr:lineardriver}
  Assume that $g$ is a regular linear driver.
  Then $\cE_g[\ \cdot\ |\cF_t]$ is linear.
  Moreover, there exists a probability measure $\bQ\sim\bP$ such that $\bE_\bQ[X|\cF_t]=\cE_g[X|\cF_t]$ for all $X\in L^2(\cF_T)$.
\end{proposition}

\begin{proof}
  Since $g(t,z)$ is regular, then $\cE_g[\ \cdot\ |\cF_t]$ is well defined and $\cE_g[0|\cF_t]=0$.
Assuming that $(Y^i,Z^i,M^i), \ i=1,2,$ are the solutions of \bsde \ \eqref{eq:bsde} with terminal condition $X^i, \ i=1,2$, for any $a,b\in\cF_t$, we get
  \begin{align*}
    aY^1_t+bY^2_t=&aY^1_T+bY^2_T+\sum_{t<s\leq T}x_s(aZ^1_s+bZ^2_s)\Delta\langle W\rangle_s-\sum_{t<s\leq T}(aZ^1_s+bZ^2_s)\Delta W_s\\
    &+(aM^1_T+bM^2_T)-(aM^1_t+bM^2_t).
  \end{align*}
Moreover, $aM^1_t+bM^2_t$ is orthogonal to $W_t$, and therefore, $aY^1_t+bY^2_t$ is the solution of BS$\Delta$E with terminal condition $aY^1_T+bY^2_T$.
  Hence, linearity of $\cE_g[\cdot|\cF_t]$ follows.

For any  $A\in\cF$ we define $\bQ(A):=\cE_g[\1_A]$.
  First, we will verify that $\bQ$ is a probability measure.
  Since $1+x_t\Delta W_t>0$, then according to Lemma \ref{le:linearcomparison} we have that $\bQ(A)=\cE_g[\1_A]\geq0$.
  If $A=\emptyset$, then $\1_A=0$ almost surely and therefore $\bQ(A)=\cE_g[0]=0$.
  If $A=\Omega$, then $\1_A=1$ and hence $\bQ(A)=\cE_g[1]=1$.
  Let $(A_i)_{i=1}^\infty$ be such that $A_i\in\cF, \ i\in\bN$, and $A_i\bigcap A_j=0$, for $i\neq j$.
  Let $(Y,Z,M)$ be the solution of BS$\Delta$E with terminal condition $\sum^\infty_{i=1}\1_{A_i}$.
  Then,
  $$
  \bQ(\bigcup_{i=1}^\infty A_i)=\cE_g[\sum^\infty_{i=1}\1_{A_i}]=\sum^\infty_{i=1}\1_{A_i}+\sum_{t<s\leq T}x(s)Z_s\Delta\langle W\rangle_s-\sum_{t<s\leq T}Z_s\Delta W_s+M_T-M_t.
  $$
  By \eqref{eq:ZT}, we also have that
  $$
  Z_T=\frac{\bE[\sum^\infty_{i=1}\1_{A_i}\Delta W_T|\cF_{T-1}]}{\Delta\langle W\rangle_T}.
  $$
From here, by Dominated Convergence Theorem, and again by \eqref{eq:ZT}, we continue
  $$
  Z_T=\sum^\infty_{i=1} \frac{\bE[\1_{A_i}\Delta W_T|\cF_{T-1}]}{\Delta\langle W\rangle_T}=\sum^\infty_{i=1}Z^i_T,
  $$
  where $Z^i_T$ is the part of the solution corresponding to terminal condition $\1_{A_i}$, $i\in\bN$.
  Similarly, we have that $Y_{T-1}=\sum^\infty_{i=1}Y^i_{T-1}$.
  Inductively, we have that $Y_0=\sum^\infty_{i=1}Y^i_0$, or $\bQ[\bigcup_{i=1}^\infty A_i]=\sum^\infty_{i=1}\bQ[A_i]$.
  Therefore, $\bQ$ is a probability measure.

  Next, we will show that $\bQ$ is equivalent to $\bP$.
  If $A\in\cF$, and $\bP(A)=0$, then $\1_A=0$, $\bP$-a.s., and thus $\bQ(A)=\cE_g[\1_A]=0$.
  Conversely, if $\bQ(A)=0$, then, by Lemma \ref{le:linearcomparison}, $\1_A=0$ $\bP$-a.s., and hence $\bP(A)=0$.
  Thus, $\bQ$ is equivalent to $\bP$.

  We are left to show that $\bE_\bQ[X|\cF_t]=\cE_g[X|\cF_t]$, for any $X\in L^2(\cF)$, $t\in\cT$.
  First, we prove the statement for $t=0$.
  For a simple random variable $X=\sum^n_{i=1}\1_{A_i}a_i\in L^2(\cF)$, with $n\in\bN$, $A_i\in\cF$, $a_i\in\bR$, using linearity of $g$-expectations proved above, we have that
  $  \bE_\bQ[X]=\cE_g[X]$.
  For a general random variable we use the standard approximation procedure. If $X\in L^2_+(\cF)$, then there exists an increasing, and positive   sequence of simple random variable $0\leq X_1\leq X_2\leq\ldots\leq X_n\leq\ldots\leq X$ such that $\lim_{n\rightarrow\infty}X_n=X$, and
  $$
  \bE_\bQ[X]=\lim_{n\rightarrow\infty}\bE_\bQ[X_n] = \lim_{n\rightarrow\infty}\cE_g[X_n].
  $$
Using similar procedure as in the first part of the proof, one can show that $Y_0^{X_n}\overset{a.s.}\longrightarrow Y_0^{X}$, and thus
$\cE_g[X]=\bE_\bQ[X]$.
Finally for the general case, $X\in L^2(\cF)$, it is enough to split $X$ into its positive and negative part.

 From here, and Proposition~\ref{pr:gexp4}, for any $t\in\{1,\ldots,T\}$, $A\in\cF_t$, and $X\in L^2(\cF)$,  we also have
    \begin{align*}
    \bE_\bQ[\1_A\cE_g[X|\cF_t]]=\cE_g[\1_A\cE_g[X|\cF_t]]=\cE_g[\cE_g[\1_AX|\cF_t]]=\cE_g[\1_AX]=\bE_\bQ[\1_AX],
  \end{align*}
  which implies that $\cE_g[ \ \cdot \ |\cF_t]$ is equal to the usual conditional expectation.

  This concludes the proof.

  \end{proof}

\begin{proposition}\label{pr:lineargexp}
  Assume that $g$ is a regular driver, and $\cE_g[\ \cdot\ |\cF_t]$ is linear for any $t\in\cT$.
  Then, there exists a process $x_t$, such that $x_t$ is predictable, $|x_t|\leq c_t$, where $c_t$ is the Lipschitz coefficient of $g(t,\cdot)$, and $g(t,z)=x_tz$, $t\in\cT$.
\end{proposition}

\begin{proof}
  Fix $t\in\cT$.
  Let us consider a \bsde\ from time $t-1$ to terminal time $t$, and  with driver $g$.
  For any fixed $z^1$, $z^2\in\bR$, it is straightforward to verify that $Z^1_t=az^1$, $M^1_t=0$, $Y^1_{t-1}=g(t,az^1)\Delta\langle W\rangle_t$, and respectively $Z^2_t=bz^2$, $M^2_t=0$, $Y^2_{t-1}=g(t,bz^2)\Delta\langle W\rangle_t$, is the solution to this \bsde\ with terminal values $az^1\Delta W_t$, and respectively $bz^2\Delta W_t$, where $a$, $b$ are arbitrary real numbers.
  Consequently, $Z^0_t=az^1+bz^2$, $M^0_t=0$, $Y^0_{t-1}=g(t,az^1+bz^2)\Delta\langle W\rangle_t$ is the solution to the \bsde\ with terminal condition $(az^1+bz^2)\Delta W_t$.

  By the definition of $g$-expectation, and using its linearity, we have
  \begin{align*}
    g(t,az^1+bz^2)\Delta\langle W\rangle_t&=\cE_g[(az^1+bz^2)\Delta W_t|\cF_{t-1}]\\
    &=a\cE_g[z^1\Delta W_t|\cF_{t-1}]+b\cE_g[z^2\Delta W_t|\cF_{t-1}]\\
    &=ag(t,z^1)\Delta\langle W\rangle_t+bg(t,z^2)\Delta\langle W\rangle_t.
  \end{align*}
From here, since $\Delta\langle W\rangle_t\neq0$, and since $z^1$ and $z^2$ were arbitrarily chosen, we have that
  $$
  g(t,\omega,az^1+bz^2)=ag(t,\omega,z^1)+bg(t,\omega,z^2), \quad \omega\in\Omega,\ z^1,\,z^2,\,a,\,b\in\bR.
  $$
 Note that since $g(t,\omega,\cdot)$ is Lipschitz, hence continuous, there exists a random variable $x_t(\omega)$ such that $g(t,\omega,z)=x_t(\omega)z$, $\omega\in\Omega$, $z\in\bR$.

   Finally, recall that $g$ is predictable, and thus $x_t$ is predictable too.
  Moreover, by Assumption~A2,
  $$
  |x_t(\omega)z_1-x_t(\omega)z_2|=|g(t,\omega,z_1)-g(t,\omega,z_2)|\leq c_t(\omega)|z_1-z_2|
  $$
  for any $\omega\in\Omega$, $z_1,z_2\in\bR$, $t\in\cT$.
  Therefore,  $|x_t|\leq c_t$, and this completes the proof.

\end{proof}

\section{Dynamic Convex Risk Measures and Dynamic Acceptability Indices via $g$-Expectation}

In this section we will explore the connections between g-Expectation and Dynamic Convex Risk Measures (DCRMs), and consequently the relationship between DCRMs and Dynamic Acceptability Indices (DAIs).

In the seminal paper \cite{ArtznerDelbaenEberHeath1999}, the authors proposed an axiomatic approach to  defining risk measures that are meant to give a numerical value of the riskiness of a given financial contract  or portfolio.
Since then, an extensive body of work was devoted to exploration of the  axiomatic approach to risk measures, and it is beyond the scope of this paper to list all the relevant literature on this subject. We refer the reader to \cite{DrapeauKupper2010} for an excellent overview of the static (one period of time) risk measures, as well as to the survey paper \cite{AcciaioPenner2010} on dynamic risk measures.
The values of risk measures can be interpreted as the capital requirement for the purpose of regulating the risk assumed by market participants (typically, by banks). In particular, the risk measures are generalizations of the well-known Value-At-Risk ($\var$).
Following a similar axiomatic approach, Cherny and Madan \cite{ChernyMadan2009} introduced the notion of coherent acceptability index -- function defined on a set of random variables that takes positive values and that are monotone, quasi-concave, and scale invariant. Coherent acceptability indices can be viewed as generalizations of performance measures such as Sharpe Ratio, Gain to Loss Ratio, Risk Adjusted Return on Capital.  Coherent acceptability indices appear to be a tool very well tailored to assessing both risk and reward of a given cashflow.  The dynamic version of coherent acceptability indices was studied in \cite{BCZ2010}.

For a robust representations of general dynamic quasi-concave, monotone and local maps see, for instance, \cite{BCDK2013}.

As it was shown in \cite{RosazzaGianin2006} there is a direct connection between convex risk measures and nonlinear expectations, and consequently there exists a direct link between convex risk measures and BSDEs. These connections were further studied in \cite{CohenElliott2011}, \cite{ElliottSiuCohen2011}, and \cite{Stadje2009} for the case of discrete time setups, thus establishing a relationship between \bsdes \ and DRM for terminal cashflows (random variables).

In \cite{ChernyMadan2009,BCZ2010} the authors proved that any (dynamic) coherent acceptability index can be associated with a family of (dynamic) coherent risk measures.
Consequently, in \cite{RosazzaGianinSgarra2012} the authors study the relationship  between dynamic \textit{sub-scale} invariant performance measures  and dynamic \textit{convex} risk measures and their connections to \bsdes.
The aim of this section is to develop a unified framework for assessing the risk and performance of cash-flows in a dynamic, discrete time setup.
It is well known that one of the key properties of dynamic risk and performance measures is their time consistency, and in this paper, we also pay special attention to this property. We refer the reader to \cite{BCP2014} for a thorough discussion of various forms of time consistency of risk/performance measures.

Besides the usual applications of these measures to risk management, and assessment of portfolios performance, we will show in next sections that dynamic sub-scale invariant acceptability indices, nonlinear expectations and theory of \bsdes \ can be successfully applied to build to build a general, arbitrage free, nonlinear pricing methodology in complex derivative markets.

\bigskip

For sake of consistency, we will follow the setup from \cite{BCIR2012,BCZ2010,BieleckiCialencoRodriguez2012} adapted to a general probability space.
A cash-flow, also called a dividend process, denoted as $D=\set{D_t}_{t=0}^T$, is any real valued, square integrable, stochastic process adapted to filtration $\bF$.
The set of all cash-flows is denoted by $\cD$, that is
$$
\cD:=\big\{(D_t)^T_{t=0}: (D_t)^T_{t=0} \mbox{ is an adapted process},\ D_t\in L^2(\cF_t),\ t\in\cT\big\}.
$$
From financial point of view, an element $D\in\cD$ should be interpreted as a cash-flow associated with a portfolio, or general  financial instrument, such that the amount $D_t$ is received/paid by the holder at time $t\in\cT$. More details, and for more general setup, please see next chapter.

For any $D\in\cD$, $\lambda\in L^\infty(\cF_t)$, we define the following multiplicative operator
$$
\lambda \cdot_t D:=(0,\ldots,0,\lambda D_t,\ldots,\lambda D_T).
$$
Note that for any $t\in\cT$, the set $\cD$ is closed under the multiplication $\cdot_t$. In what follows, if no confusion arise, we will simply write $\lambda D$ instead of $\lambda\cdot_t D$.
We also define an order $\succeq_t$ on $\cD$, and say that $D^1\succeq_t D^2$, whenever $\sum^T_{s=t}D^1_s\geq\sum^T_{s=t}D^2_s$, $t\in\cT$.
Hence, $D^1\succeq_t D^2$ if the future cumulative cash-flow is respectively larger.

\subsection{Dynamic Convex Risk Measures via $g$-Expectation}

We start by recalling the definition of Dynamic Convex/Coherent Risk Measures.

\begin{definition}[Dynamic Convex Risk Measure]
  A dynamic convex risk measure is a function $\rho:\cT\times\cD\times\Omega\rightarrow\bR$ that satisfies the following properties:
  \begin{itemize}
    \item[R1.]
    \textbf{Adapted.} $\rho_t(D)$ is $\cF_t$-measurable, for any $t\in\cT, D\in\cD$.
    \item[R2.]
    \textbf{Local.} $\1_A\rho_t(D)=\1_A\rho_t(\1_A D)$, for any $t\in\cT, A\in\cF_t, D\in\cD$.
    \item[R3.]
    \textbf{Convex.} $\rho_t(\lambda D+(1-\lambda)D')\leq\lambda\rho_t(D)+(1-\lambda)\rho_t(D')$, for any $t\in\cT$, $\lambda\in L^\infty(\cF_t)$, $0\leq\lambda\leq1$, $D,D'\in\cD$.
    \item[R4.]
    \textbf{Monotone.} If $D\succeq_tD'$ for some $t\in\cT$, $D,D'\in\cD$, then $\rho_t(D)\leq\rho_t(D')$.
    \item[R5.]
    \textbf{Cash-additive.} $\rho_t(D+m\1_{\{s\}})=\rho_t(D)-m$ for any $t\in\cT$, $D\in\cD$, and $\cF_t$-measurable random variable $m$ and $s\geq t$.
    \item[R6.] \textbf{Time consistent.}
    $\rho_t(D)=\rho_t(-\rho_{t+1}(D)\1_{\{t+1\}})-D_t$ for any $t=0,1,\ldots,T-1$ and $D\in\cD$.
  \end{itemize}
\end{definition}

If $\rho$ furthermore is
  \begin{itemize}
    \item[R7.]
    \textbf{Positive-homogeneous.} $\rho_t(\lambda D)=\lambda\rho_t(D)$, $\lambda\in L^\infty_+(\cF_t)$,
  \end{itemize}
  then $\rho:\cT\times\cD\times\Omega\rightarrow\bR$ is called a \textit{dynamic coherent risk measure}.

Properties R1-R7 have a clear financial interpellation:  adaptiveness means that the measurements are consistent with the flow of information;
the locality property essentially means that the values of the risk measure restricted to a set $A\in\cF$ remain invariant with respect to the values of the arguments outside of the same set $A\in\cF$, and in particular, the events that will not happen in the future do not change the value of the measure today; convexity implies that diversification reduces the risk; monotonicity implies that a cashflow with higher payoffs bears less risk. Cash-additivity means that adding $\$m$ to a portfolio at any time in the future reduces the overall risk by the same amount $m$. From the regulatory perspective, the value of a risk measure is typically interpreted as the minimal capital requirement for a bank.
There exists various forms of time consistency of risk and performance measures and we refer the reader to recent paper \cite{BCP2014} where the authors give a systematic approach to time consistency of LM-measures (local and monotone functions). The time consistency R6 for DCRM considered here is known in the existing literature as strong time consistency.
It can be shown that property R6, combined with R4, is equivalent to the following property: if $D_t=D_t'$ and $\rho_{t+1}(D)=\rho_{t+1}(D')$, then $\rho_t(D)=\rho_t(D')$, i.e. if two cashflows bear the same risk tomorrow, and they pay the same dividend today, then today these two cashflows are assess at the same risk level. Saying differently, the risk is measured consistently in time. Finally, positive-homogeneity means the risk of a rescaled cashflow is rescaled by the same factor.

Similar to, \cite{CohenElliott2011,ElliottSiuCohen2011,Stadje2009} that address discrete time case, and \cite{RosazzaGianin2006,RosazzaGianinSgarra2012} that consider the continuous time setup, we will show that the solution of \bsdes \ \eqref{eq:bsde} with convex drivers, more precisely the corresponding $g$-Expectation, generates a DCRM.
In the sequel, for any regular and convex driver $g$, we will denote by $\rho^g$ the function defined as follows
\begin{equation}\label{eq:defRhoG}
\rho^g_t(D)=\cE_g[-\sum^T_{s=t}D_s|\cF_t], \quad t\in\cT, \ D\in\cD.
\end{equation}

\begin{theorem}\label{th:dr1}
  Assume that $g$ is a convex and  regular driver.
  Then, the function $\rho^g_t$ is a  DCRM.
\end{theorem}

\begin{proof} Invoking Proposition~\ref{pr:gexp4}, it is straightforward to show that  $\rho^g_t$ satisfies Properties R1--R5.

 We will show that $\rho^g$ is time consistent.
  By Proposition~\ref{pr:gexp4}~(iii) and (v), we immediately  have that
    \begin{align*}
      \rho^g_t(-\rho^g_{t+1}(D)\1_{\{t+1\}})-D_t=&\cE_g[\rho^g_{t+1}(D)|\cF_t]-D_t
      =\cE_g[\cE_g[-\sum^T_{s=t+1}D_s|\cF_{t+1}]|\cF_t]-D_t\\
      =&\cE_g[\cE_g[-\sum^T_{s=t}D_s|\cF_{t+1}]|\cF_t]
      =\cE_g[-\sum^T_{s=t}D_s|\cF_t]\\
      =&\rho^g_t(D),
    \end{align*}
for any $t\in\cT$, $D\in\cD$. This concludes the proof.
\end{proof}

As an immediate consequence of Theorem~\ref{th:dr1}, if additionally the driver $g$ is positive homogeneous in $z$, then $\rho^g$ is a time consistent coherent risk measure.

\begin{corollary}\label{co:cr}
  Assume that $g(t,z)$ is a convex and regular driver, such that $g(t,\cdot)$ is positive homogeneous.
  Then, $\rho^g_t$ is a dynamic coherent risk measure.
\end{corollary}

\begin{proof}
  In view of Theorem \ref{th:dr1}, we have that $\rho^g_t(D)$ is a time consistent DCRM.
  Moreover, for any fixed $t\in\cT$, due to Proposition \ref{pr:gexp4} (vii), we have that
  $$
  \rho^g_t(\lambda D)=\cE_g[-\lambda\sum^T_{s=t}D_s|\cF_t]=\lambda\cE_g[-\sum^T_{s=t}D_s|\cF_t]=\lambda\rho^g_t(D),
  $$
  for any $\lambda\in L^\infty_+(\cF_t)$, and $D\in\cD$.
  Hence, $\rho^g_t$ is a time consistent coherent risk measure.
\end{proof}

Next, we give some examples of DCRMs generated by various drivers via $g$-expectation.
Similar to Example~\ref{ex:coherentdriver} and \ref{ex:convexdriver}, in the following two examples, we will assume that $W$ is a martingale process such that $\Delta W_t$ is uniformly bounded.

\begin{example}[Coherent Case]\label{ex:coherentdriver2}
 We consider the driver $g(t,z)=c|z|$, for some fixed $c\in[0,1)$. It is easy to see that $g(t,z)$ is convex and positive homogeneous with respect to $z$.
  Then, by Corollary \ref{co:cr}, $\rho^g_t(D)=\cE_g[-\sum^T_{s=t}D_s|\cF_t]$ is a dynamic coherent risk measure.
\end{example}

\begin{example}[Convex Case]\label{ex:convexdriver2}
  Let us put
  $$
  g(t,z)=\frac{K}{(K+1)\|\Delta W_t\|_\infty}\ln(\frac{1}{3}+\frac{1}{3}e^{-z}+\frac{1}{3}e^{z}),
  $$
  where $K\in\bR^+$ is fixed.
  Note that $\frac{\partial g}{\partial z}$ is an increasing function with respect to $z$, and hence $g(t,\cdot)$ is a convex driver.
  By Theorem \ref{th:dr1}, we have that $\rho^g_t(D)=\cE_g[-\sum^T_{s=t}D_s|\cF_t]$ is a time consistent DCRM.
\end{example}

By Theorem \ref{th:dr1} any given convex regular driver $g$ generates DCRM.
Next result shows that the converse implication also holds true if $W$ is the symmetric random walk and the filtration is generated by $W$.

\begin{proposition}\label{pr:rhotog}
  Assume that $W$ is a symmetric random walk, $\cF=\cF^W$, and $\rho_t: L^\infty(\cF_T)\rightarrow L^\infty(\cF_t)$ is a DCRM.
  Also, let $g(t,z)=\frac{\rho_{t-1}(z\Delta W_t)}{\Delta\langle W\rangle_t}$.
  Then, $\rho_t(X)=\rho^g_t(X)$, for any $X\in L^\infty(\cF_T)$.
\end{proposition}

\begin{proof}
  Fix an $X\in L^\infty(\cF_T)$.
  Since $\rho_t(X)\in L^\infty(\cF_t)\subset L^2(\cF_t)$, for any $t\in\cT$, then, since $W$ has the martingale representation property (cf. Proposition~\ref{pr:randomwalk}),   there exists $Z_t\in\cF_{t-1}$ such that $\rho_t(X)=\bE[\rho_t(X)|\cF_{t-1}]+Z_t\Delta W_t$.
According to time-consistency property of $\rho$, we have that
\begin{align*}
\rho_t(X)-\rho_{t-1}(X)=&\rho_t(X)-\rho_{t-1}(-\rho_t(X))\\
=&\bE[\rho_t(X)|\cF_{t-1}]+Z_t\Delta W_t-\bE[\rho_t(X)|\cF_{t-1}]-\rho_{t-1}(Z_t\Delta W_t)\\
=&-\rho_{t-1}(Z_t\Delta W_t)+Z_t\Delta W_t.
\end{align*}
Therefore,
$$
\rho_{t-1}(X)=\rho_t(X)+\frac{\rho_{t-1}(Z_t\Delta W_t)}{\Delta\langle W\rangle_t}\Delta\langle W\rangle_t-Z_t\Delta W_t
$$
and hence, $\rho_t(X)=\rho^g_t(X)$.
\end{proof}

The next example is an application of Proposition~\ref{pr:rhotog}.

\begin{example}[Dynamic Entropic Risk Measure]\label{ex:derm}
 The dynamic entropic risk measure takes the following form %\textcolor[rgb]{0.98,0.00,0.00}{(cf. e.g. ....)}
  $$
  \rho^\gamma_t(X)=\gamma\ln(\bE[\exp(-\frac{X}{\gamma})|\cF_t])
  $$
  where $\gamma>0$. So the driver corresponding to entropic risk measure will be
  $$
  g(t,z)=\frac{\gamma}{\Delta\langle W\rangle_t}\ln(\bE[\exp(-\frac{z\Delta W_t}{\gamma})|\cF_{t-1}])=\frac{\gamma}{\Delta\langle W\rangle_t}\ln(\frac{1}{2}e^{-\frac{z}{\gamma}}+\frac{1}{2}e^{\frac{z}{\gamma}}).
  $$
  Similarly as in Example~\ref{ex:convexdriver2}, we have that $g(t,z)$ is a convex regular driver.
\end{example}

\subsection{Dynamic Acceptability Indices via $g$-Expectation}

In \cite{ChernyMadan2009} the authors introduced the notion of coherent acceptability index, meant to measure the performance of a given terminal cashflow, as a monotone, scale-invariant, and quasi-concave function defined on the set of all bounded random variables. The extension of these measures to a dynamic setup was studied in \cite{BCZ2010}, where the appropriate notion of time consistency for dynamic coherent acceptability indices was introduced. In \cite{BCZ2010}, the author follow a discrete time market setup on a finite probability space. Recently, Biagini and Bion-Nadal \cite{BiaginiBion-Nadal2012} studied these type of dynamic performance measure of terminal cash-flows on a general probability space, and discrete time market setup. For robust representations of general dynamic quasi-concave performance measures we refer the reader to \cite{BCDK2013}.
The aim of this section is to study dynamic performance measures that are sub-scale invariance.
Following \cite{RosazzaGianinSgarra2012}, who  argued that the scale invariance fails to capture all risks in an illiquid market, and, accordingly, postulated the sub-scale invariance instead, we replace the scale invariance condition used in  \cite{BCZ2010}, with a weaker assumption of sub-scale invariance.

\begin{definition}[Dynamic Acceptability Index]
  A dynamic acceptability index is a function $\alpha:\cT\times\cD\times\Omega\rightarrow[0,\infty]$ that satisfies the following properties:
  \begin{itemize}
    \item[I1.] \textbf{Adapted.} $\alpha_t(D)$ is $\cF_t$-measurable, for any $t\in\cT, D\in\cD$.
    \item[I2.] \textbf{Local.} $\1_A\alpha_t(D)=\1_A\alpha_t(\1_A D)$, for any $t\in\cT, A\in\cF_t, D\in\cD$.
    \item[I3.] \textbf{Quasi-concave.} If $\alpha_t(D)\geq m$ and $\alpha_t(D')\geq m$ for some positive $\cF_t$-measurable random variable $m$, and $D,D'\in\cD$, then $\alpha_t(\lambda D+(1-\lambda)D')\geq m$, for any $\lambda\in L^\infty(\cF_t)$, $0\leq\lambda\leq1$.
    \item[I4.] \textbf{Monotone.} If $D\succeq_tD'$ for some $t\in\cT$, and $D,D'\in\cD$, then $\alpha_t(D)\geq\alpha_t(D')$.
    \item[I5.] \textbf{Sub-scale Invariant.} $\alpha_t(\lambda D)\geq\alpha_t(D)$ for any $\lambda\in L^\infty(\cF_t)$, $0\leq\lambda\leq1$, $D\in\cD$, or, equivalently, $\alpha_t(\lambda D)\leq\alpha_t(D)$ for any $\lambda\in L^\infty(\cF_t)$, $\lambda\geq 1$, $D\in\cD$.
    \item[I6.] \textbf{Time Consistent.} For any $t\in\cT$, $D, D'\in\cD$, the following implication holds true
  $$
  \alpha_{t+1}(D)\geq m\geq\alpha_{t+1}(D') \quad  \Rightarrow \quad \alpha_t(D)\geq m\geq\alpha_t(D'),
  $$
  whenever $D_t\geq0\geq D'_t$, and $m$ being an non-negative $\cF_t$-measurable random variable.
  \end{itemize}
\end{definition}

\begin{remark}
  (i) We note that property I5 is weaker than the scale invariance property
  \begin{itemize}
    \item[I5'.] \textbf{Scale Invariant.} $\alpha_t(\lambda D)=\alpha_t(D)$ for any $\lambda\in L^\infty(\cF_t)$, $\lambda>0$, $D\in\cD$,
  \end{itemize}
   (ii) If property I5 from the definition of DAI is replaced with I5',  then $\alpha$ is called \textit{dynamic coherent acceptability index}.
\end{remark}

Analogously to \cite{BCZ2010,ChernyMadan2009}, we will show that a DAI can be generated by a family of DCRMs, and hence by a family of \bsdes.
For this, we will consider families of drivers that satisfy the following assumptions: \\[.2cm]
\textbf{Assumption G:}\\[-0.6cm]
\begin{itemize}
  \item[G1.]
  $g_{x_2}\geq g_{x_1}$ for $x_2\geq x_1>0$;
  \item[G2.]
  $g_x$ is a convex regular driver for any $x>0$.
  \item[G3.]
  $g_x=g_{x-}$ for any $(t,\omega,z)\in\bR^+\times\Omega\times\bR$;
\end{itemize}
Towards this end, with slight  abuse of notations, we will denote by $g$ the family of drivers $(g_x)_{x>0}$.

In what follows, for any family of drivers $g$ that satisfy assumption~G, we denote by $\alpha^g$ the following function
\begin{equation}\label{eq:alphadef}
\alpha^g_t(D)(\omega):=\sup\Big\{x\in \bR, x>0:\cE_{g_{x}}\Big[-\sum^T_{s=t}D_s\Big|\cF_t\Big](\omega)\leq0\Big\}, \quad \omega\in\Omega, t\in\cT, \ D\in\cD.
\end{equation}
Then we have the following theorem.

\begin{theorem}\label{th:gdai}
  Assume that the family of drivers $g=(g_x)_{x>0}$ satisfies Assumption G.
  Then, $\alpha^g_t$ is a dynamic acceptability index.
\end{theorem}

The proof will follow from a series of lemmas proved below.

\begin{lemma}
  Assume that the family of drivers $g=(g_x)_{x>0}$ satisfies Assumption G.
  Then, $\alpha^g_t$ satisfies properties I1-I5.
\end{lemma}

\begin{proof}
  We will show that $\alpha^g_t$ satisfies properties I1-I5.

\smallskip\noindent
    I1. Let us consider the set $A_\gamma=\{\omega\in\Omega:\alpha^g_t(D)(\omega)\geq \gamma\}$, where $\gamma\in\bR$, $t\in\cT$ and $D\in\cD$ are fixed.
    We want to show that $A_\gamma\in\cF_t$.
    If $\gamma\leq 0$, then it is clear that $A_\gamma=\Omega\in\cF_t$.
    For $\gamma>0$, we will prove that
    \begin{align*}
      \alpha^g_t(D)(\omega)\geq \gamma, \quad \omega\in\Omega,
    \end{align*}
    is equivalent to
    \begin{equation}\label{eq:eggamma}
      \cE_{g_\gamma}\Big[-\sum^T_{s=t}D_s\Big|\cF_t\Big](\omega)\leq0, \quad \omega\in\Omega.
    \end{equation}
    According to definition \eqref{eq:alphadef}, any $\omega\in\Omega$ such that $\alpha^g_t(D)(\omega)\geq\gamma$ satisfies that
    $$
    \sup\Big\{x\in \bR, x>0:\cE_{g_{x}}\Big[-\sum^T_{s=t}D_s\Big|\cF_t\Big](\omega)\leq0\Big\}\geq\gamma,
    $$
    which, by assumption G1, implies that
    $$
    \lim_{x\uparrow\gamma}\cE_{g_x}\Big[-\sum^T_{s=t}D_s\Big|\cF_t\Big](\omega)\leq0.
    $$
    Therefore, in order to show \eqref{eq:eggamma} holds for such $\omega$, we only need to verify that
    \begin{equation}\label{eq:limitgamma}
    \lim_{x\uparrow\gamma}\cE_{g_x}\Big[-\sum^T_{s=t}D_s\Big|\cF_t\Big](\omega)=\cE_{g_\gamma}\Big[-\sum^T_{s=t}D_s\Big|\cF_t\Big](\omega) \quad \omega\in A_\gamma.
    \end{equation}
    Let $(Y^x,Z^x,M^x)$, $0<x\leq\gamma$, be the solutions of BS$\Delta$Es with drivers $g_x$ and the terminal condition $Y^x_T=-\sum^T_{s=t}D_s$.
    Then, \eqref{eq:limitgamma} is implied by
    \begin{align}\label{eq:limityx}
      \lim_{x\uparrow\gamma}Y^x_t=Y^\gamma_t,
    \end{align}
    which holds clearly for $t=T$.

    Suppose that \eqref{eq:limityx} is true for some $t<u\leq T$.
    According to \eqref{eq:zs}, we have the representation:
    $$
    Z^x_u=\frac{\bE[Y^x_u\Delta W_u|\cF_{u-1}]}{\Delta\langle W\rangle_u}.
    $$
    Notice that $\{Y^x_u\}_{0<x<\gamma}$ is increasing with respect to $x$ because of Theorem~\ref{th:comp}.
    Hence, by dominated convergence theorem, we get that
    $$
    \lim_{x\uparrow\gamma}Z^x_u=\lim_{x\uparrow\gamma}\frac{\bE[Y^x_u\Delta W_u|\cF_{u-1}]}{\Delta\langle W\rangle_u}=\frac{\bE[Y^\gamma_u\Delta W_u|\cF_{u-1}]}{\Delta\langle W\rangle_u}=Z^\gamma_u.
    $$
    Next, by \eqref{eq:ys}, $Y^x_{u-1}$ can be represented by
    $$
    Y^x_{u-1}=\bE[Y^x_u|\cF_{u-1}]+g_x(u,Z^x_u)\Delta\langle W\rangle_u,
    $$
    where the following equality
    $$
    \lim_{x\uparrow\gamma}\bE[Y^x_u|\cF_{u-1}]=\bE[Y^\gamma_u|\cF_{u-1}]
    $$
    holds true due to dominated convergence theorem.
    Moreover, we have that
    \begin{align*}
      &|g_x(u,\omega,Z^x_u(\omega))-g_\gamma(u,\omega,Z^\gamma_u(\omega))|\\
      \leq&|g_x(u,\omega,Z^x_u(\omega))-g_x(u,\omega,Z^\gamma_u(\omega))|+|g_x(u,\omega,Z^\gamma_u(\omega))-g_\gamma(u,\omega,Z^\gamma_u(\omega))|\\
      \leq&c^x_t(\omega)|Z^x_u(\omega)-Z^\gamma_u(\omega)|+|g_x(u,\omega,Z^\gamma_u(\omega))-g_\gamma(u,\omega,Z^\gamma_u(\omega))|,
    \end{align*}
    for almost all $\omega\in\Omega$.
    Recall that $c^x_t(\omega)$ is defined as the smallest Lipschitz constant of $g_x$.
    Since that the family of drivers satisfy assumption G1 and G2, then by Proposition~\ref{pr:convexfunctions}, we have that $c^x_t\leq c^\gamma_t$.
    The following equality follows immediately:
    $$
    \lim_{x\uparrow\gamma}g_x(u,Z^x_u)=g_\gamma(u,Z^\gamma_u)
    $$
    Thus, we have proved
    $$
    \lim_{x\uparrow\gamma}Y^x_{u-1}=Y^\gamma_{u-1},
    $$
    and \eqref{eq:limityx} is true by induction.

    On the other hand, for any $\omega\in\Omega$ such that $\cE_{g_{\gamma}}\big[-\sum^T_{s=t}D_s\big|\cF_t\big](\omega)\leq0$.
    We get that
    $$
    \alpha^g_t(D)(\omega)=\sup\Big\{x\in \bR, x>0:\cE_{g_{x}}\Big[-\sum^T_{s=t}D_s\Big|\cF_t\Big](\omega)\leq0\Big\}\geq\gamma,
    $$
    which implies that $\omega\in A_\gamma$.
    In view of the above, we conclude that
    $$
    A_\gamma=\Big\{\omega\in\Omega: \alpha^g_t(D)(\omega)\geq\gamma\Big\}=\Big\{\omega\in\Omega: \cE_{g_{\gamma}}\Big[-\sum^T_{s=t}D_s\Big|\cF_t\Big](\omega)\leq0\Big\}.
    $$
    Since that $\cE_{g_\gamma}\big[-\sum^T_{s=t}D_s\big|\cF_t\big]$ is $\cF_t$-measurable, then $A_\gamma\in\cF_t$.
    This completes the proof of showing that $\alpha^g$ is adapted.

    \smallskip\noindent
    I2. Let us fix $t\in\cT$, $D\in\cD$, and $A\in\cF_t$.
    Then,  for almost all $\omega\in\Omega$, we have that
    \begin{align*}
      \1_A(\omega)\alpha^g_t(D)(\omega)=&\1_A(\omega)\sup\Big\{x\in\bR, x>0:\cE_{g_x}\Big[-\sum^T_{s=t}D_s\Big|\cF_t\Big](\omega)\leq0\Big\}\\
      =&\1_A(\omega)\sup\Big\{x\in\bR, x>0:\1_A(\omega)\cE_{g_x}\Big[-\sum^T_{s=t}D_s\Big|\cF_t\Big](\omega)\leq0\Big\}.
    \end{align*}
    By (v) in Proposition~\ref{pr:gexp4}, we deduce that $\1_A\cE_{g_x}[-\sum^T_{s=t}D_s|\cF_t]=\cE_{g_x}[-\1_A\sum^T_{s=t}D_s|\cF_t]$, and locality of $\alpha$ follows immediately.

    \smallskip\noindent
    I3. Let us fix $t\in\cT$, $D,\ D'\in\cD$.
    Also, let $\gamma>0$ be some $\cF_t$-measurable random variable such that $\alpha^g_t(D)(\omega)\geq\gamma(\omega)$ and $\alpha^g_t(D')(\omega)\geq\gamma(\omega)$ hold for almost all $\omega$.
    Fix one such $\omega$ and denote $\gamma^*=\gamma(\omega)$.
    Similar to the proof of adaptiveness, we have that $\alpha^g_t(D)(\omega)\geq\gamma^*$ and $\alpha^g_t(D')(\omega)\geq\gamma^*$ will imply that
    \begin{align*}
      \cE_{g_{\gamma(\omega)}}\Big[-\sum^T_{s=t}D_s\Big|\cF_t\Big](\omega)&=\cE_{g_{\gamma^*}}\Big[-\sum^T_{s=t}D_s\Big|\cF_t\Big](\omega)\leq0,\\
      \cE_{g_{\gamma(\omega)}}\Big[-\sum^T_{s=t}D'_s\Big|\cF_t\Big](\omega)&=\cE_{g_{\gamma^*}}\Big[-\sum^T_{s=t}D'_s\Big|\cF_t\Big](\omega)\leq0,
    \end{align*}
    respectively.
    Then according to (vi) in Proposition~\ref{pr:gexp4}, we get that
    $$
    \cE_{g_{\gamma(\omega)}}\Big[-\sum^T_{s=t}\Big(\lambda D_s+(1-\lambda)D'_s\Big)\Big|\cF_t\Big](\omega)\leq0,
    $$
    holds for almost all $\omega\in\Omega$, and it implies that
    $$
    \alpha^g_t(\lambda D+(1-\lambda)D')(\omega)=\sup\Big\{x\in\bR, x>0:\cE_{g_x}\Big[-\sum^T_{s=t}\Big(\lambda D_s+(1-\lambda)D'_s\Big)\Big|\cF_t\Big](\omega)\leq0\Big\}\geq\gamma(\omega).
    $$
    Therefore, the Quasi-concavity of $\alpha^g_t$ holds.

    \smallskip\noindent
    I4. Fix $t\in\cT$. Let $D$, $D'\in\cD$, and suppose that $D\succeq_tD'$.
    For any fixed $\cF_t$-measurable random variable $\gamma>0$ such that $\alpha^g_t(D')\geq\gamma$, we have that
    $$
    \cE_{g_{\gamma(\omega)}}\Big[-\sum^T_{s=t}D'_s\Big|\cF_t\Big](\omega)\leq0
    $$
    holds for almost every $\omega\in\Omega$.
    Due to the assumption $D\succeq_tD'$, which implies that $\sum^T_{s=t}D_s\geq\sum^T_{s=t}D'_s$,
    we get by (ii) in Proposition~\ref{pr:gexp4} that
    $$
    \cE_{g_{\gamma(\omega)}}\Big[-\sum^T_{s=t}D_s\Big|\cF_t\Big](\omega)\leq\cE_{g_{\gamma(\omega)}}\Big[-\sum^T_{s=t}D'_s\Big|\cF_t\Big](\omega)\leq0,
    $$
    and the statement $\alpha^g_t(D)\geq\gamma$ follows.
    Hence, $\alpha^g_t$ is monotone.

    \smallskip\noindent
    I5. Fix $t\in\cT$, $D\in\cD$.
    For any fixed $\cF_t$-measurable random variable $\gamma>0$ such that $\alpha^g_t(D)\geq\gamma$, we have that
    $$
    \cE_{g_{\gamma(\omega)}}\Big[-\sum^T_{s=t}D'_s\Big|\cF_t\Big](\omega)\leq0
    $$
    holds for almost every $\omega\in\Omega$.
    By convexity of $g$-expectation, it is true that
    \begin{align*}
      \cE_{g_{\gamma(\omega)}}\Big[-\lambda\sum^T_{s=t}D_s\Big|\cF_t\Big]&\leq\lambda\cE_{g_{\gamma(\omega)}}\Big[-\sum^T_{s=t}D_s\Big|\cF_t\Big]+(1-\lambda)\cE_{g_{\gamma(\omega)}}\Big[0\Big|\cF_t\Big]\\
      &=\lambda\cE_{g_{\gamma(\omega)}}\Big[-\sum^T_{s=t}D_s\Big|\cF_t\Big]\leq0,
    \end{align*}
    for any $\lambda\in L^\infty(\cF_t)$, $0\leq\lambda\leq1$, and almost all $\omega\in\Omega$.
    Therefore, we conclude that $\alpha^g_t(\lambda D)\geq\gamma$, which implies that $\alpha^g_t(\lambda D)\geq\alpha^g_t(D)$.

\end{proof}

\begin{remark}
  In Assumption G, G1 and G2 are natural conditions for constructing DAI via $g$-expectation.
  The reason to assume G3 comes from duality between DCRM and DAI.
  If a DAI is given, then DCRM can be defined as
  $$
  \rho^{\gamma}_t(D)=\essinf\{c\in\cF_t:\alpha^g_t(D+c\1_{\{t\}})\geq\gamma\},
  $$
  which is equivalent to
  $$
  \rho^{\gamma}_t(D)=\essinf\Big\{c\in\cF_t:\esssup\Big\{x>0:\cE_{g_x}\Big[-\sum^T_{s=t}D_s-c\1_{\{t\}}\Big|\cF_t\Big]\leq0\Big\}\geq\gamma\Big\}.
  $$
  So we have that,
  $$
  \rho^{\gamma}_t(D)=\essinf\Big\{c\in\cF_t:\esssup\Big\{x>0:\cE_{g_x}\Big[-\sum^T_{s=t}D_s\Big|\cF_t\Big]\leq c\1_{\{t\}}\Big\}\geq\gamma\Big\}.
  $$
  Random variable $\rho^\gamma_t(D)$ can be represented as
  $$
  \rho^{\gamma}_t(D)=\essinf\Big\{c\in\cF_t:\forall \varepsilon>0,\cE_{g_{\gamma-\varepsilon}}\Big[-\sum^T_{s=t}D_s\Big|\cF_t\Big]\leq c\Big\},
  $$
  Hence, we get that $\rho^\gamma_t(D)=\sup_{\varepsilon>0}\{\cE_{g_{\gamma-\varepsilon}}[-\sum^T_{s=t}D_s|\cF_t]\}=\cE_{g_{\gamma-}}[-\sum^T_{s=t}D_s|\cF_t]$.
  In order for duality to hold, it is required to have that $\rho^{g_\gamma}_t(D)=\rho^\gamma_t(D)$ which is equivalent to follows
  $$
  \cE_{g_{\gamma-}}\Big[-\sum^T_{s=t}D_s\Big|\cF_t\Big]=\cE_{g_\gamma}\Big[-\sum^T_{s=t}D_s\Big|\cF_t\Big].
  $$
  Finally, according to comparison theorem, we have that $g_\gamma=g_{\gamma-}$.
\end{remark}

\begin{lemma}\label{le:dcai}
  Assume that the family of drivers $g=(g_x)_{x>0}$ satisfies Assumption G.
  Also suppose that $g_x$ is positive homogeneous for any $x\in\bR^+$.
  Then, $\alpha^g_t$ is scale invariant.
\end{lemma}

\begin{proof}
  We need to show that $\alpha^g_t$ is scale invariant.

  Let $t\in\cT$, $D\in\cD$, and $\lambda\in L^\infty(\cF_t)$, $\lambda>0$.
  By definition, we have that
  $$
  \alpha^g_t(\lambda D)=\esssup\{x>0:\cE_{g_x}[-\lambda\sum^T_{s=t}D_s|\cF_t]\leq0\}.
  $$
  In view of the fact that $g_x(t,\cdot)$ is positive homogeneous for any $x>0$, Proposition~\ref{pr:gexp4} (vii) implies that
  $$
  \alpha^g_t(\lambda D)=\esssup\{x>0:\lambda\cE_{g_x}[-\sum^T_{s=t}D_s|\cF_t]\leq0\}.
  $$
  Since $\lambda>0$, then we get that $\lambda\cE_{g_x}[-\sum^T_{s=t}D_s|\cF_t]\leq0$ is equivalent to $\cE_{g_x}[-\sum^T_{s=t}D_s|\cF_t]\leq0$.
  Hence, $\alpha^g_t(\lambda D)=\esssup\{x>0:\cE_{g_x}[-\sum^T_{s=t}D_s|\cF_t]\leq0\}=\alpha^g_t(D)$.
  This concludes that $\alpha^g_t$ is a dynamic coherent acceptability index.
\end{proof}

By using the following lemma, we will prove that $\alpha^g_t$ is time consistent.

\begin{lemma}\label{le:simpleconsist}
   Assume that the family of drivers $g=(g_x)_{x>0}$ satisfies Assumption G.
   Also suppose that $D_t\geq0\geq D'_t$ for some $t\in\cT$, $D, D'\in\cD$, and there exists $c\in\bR^+, A\in\cF_t$ such that $\1_A\alpha^g_{t+1}(D)\geq\1_Ac\geq\1_A\alpha^g_{t+1}(D')$.
   Then $\1_A\alpha^g_{t}(D)\geq\1_Ac\geq\1_A\alpha^g_{t}(D')$.
\end{lemma}

\begin{proof}
  Suppose that $\1_A\alpha^g_t(D)\geq\1_Ac\geq\1_A\alpha^g_t(D')$ for some $t\in\cT$, $c\in\bR^+$, $A\in\cF_t$, and $D,D'\in\cD$ such that $D_t\geq0\geq D'_t$.
  Then, for any $\beta<c$, we have that $\cE_{g_\beta}[-\1_A\sum^T_{s=t+1}D_s|\cF_{t+1}]\leq0$. Since $D_t\geq0$, then according to Proposition~\ref{pr:gexp4} (ii) and (iii), we get that
  $$
  \cE_{g_\beta}\Big[-\1_A\sum^T_{s=t}D_s\Big|\cF_t\Big]=-\1_A D_t+\cE_{g_\beta}\Big[\cE_{g_\beta}\Big[-\1_A\sum^T_{s=t+1}D_s\Big|\cF_{t+1}\Big]\cF_t\Big]\leq0.
  $$
  Due to the fact that $\beta$ is arbitrary, the following is true:
  \begin{align*}
  \1_A\alpha^g_t(D)&=\1_A\esssup\Big\{x>0:\cE_{g_x}\Big[-\sum^T_{s=t}D_s\Big|\cF_t\Big]\leq0\Big\}\\
  &=\1_A\esssup\Big\{x>0:\cE_{g_x}\Big[-\1_A\sum^T_{s=t}D_s\Big|\cF_t\Big]\leq0\Big\}\\
  &\geq\1_Ac.
  \end{align*}
  On the other hand, by knowing that $\1_A\alpha^g_{t+1}(D')\leq \1_Ac$, we will prove $\1_A\alpha^g_t(D')\leq\1_Ac$ by contradiction.
  Assume that there exists some $A'\subset A$, $\bP(A')>0$ such that $\alpha^g_t(D')>c$ on $A'$. Then there exists a $\beta>c$ and $A''\subset A$, $\bP(A'')>0$ such that $\alpha^g_t(D)>\beta$ on $A"$.
  Hence, we have that
  \begin{align*}
    \cE_{g_\beta}\Big[-\1_{A''}\sum^T_{s=t}D'_s\Big|\cF_t\Big]\leq0.
  \end{align*}
  However, since $\alpha^g_{t+1}(D')\leq c$ on $A$, then $\alpha^g_{t+1}(D')<\beta$ on $A''$.
  In view of the fact that $D'_t\leq0$, we get for $\omega\in A''$ that
  $$
  \cE_{g_\beta}\Big[-\1_{A''}\sum^T_{s=t}D'_s\Big|\cF_t\Big](\omega)=-\1_{A''}(\omega)D'_t(\omega)+\cE_{g_\beta}\Big[\cE_{g_\beta}\Big[-1_{A''}\sum^T_{s=t+1}D'_s\Big|\cF_{t+1}\Big|\cF_t\Big](\omega)>0.
  $$
  Hence, there is a contradiction and such result implies that $\1_A\alpha^g_t(D')\leq\1_Ac$.
\end{proof}

Now we are ready to finish proving Theorem~\ref{th:gdai}, which is to verify that $\alpha^g_t$ is time consistent.

\begin{lemma}
  Assume that the family of drivers $g=(g_x)_{x>0}$ satisfies Assumption G.
  Then, $\alpha^g_t$ satisfies I6.
\end{lemma}

\begin{proof}
  For any $t\in\cT$, $D\in\cD$, if there exists a positive $\cF_t$-measurable random variable $m$ such that $\alpha^g_{t+1}(D)\geq m$, then there exists a sequence of simple random variables $\phi_n=\sum^k_{i=1}\1_{A^n_i}a^n_i$ where $A^n_i\in\cF_t$, $a^n_i\in\bR^+$, $n\in\{1,2,\ldots\}$, such that $\phi_n\leq\phi_{n+1}$ and $\lim_{n\rightarrow\infty}\phi_n=m$.
  Hence, we have that $\alpha^g_{t+1}(D)\geq\phi_n$ for any $n$.
  Since that $D_t\geq0$, then according to locality of $\alpha^g_t$ and Lemma \ref{le:simpleconsist}, The following is true
  $$
  \alpha^g_t(D)\geq\phi_n,\ \quad n\in\{1,2,\ldots\}.
  $$
  Therefore, we conclude that $\alpha^g_t(D)\geq m$.

  Assume that there exists a positive $\cF_t$-measurable random variable $m$ such that $\alpha^g_{t+1}(D')\leq m$.
  Fix $N\in\bR^+$,
    and for any $\omega\in\Omega$, define
  \begin{equation*}
    m_N(\omega)=
    \left\{
    \begin{array}{lll}
      m(\omega) &\mbox{if $m(\omega)<N$},\\
      N &\mbox{if $m(\omega)\geq N$}.
    \end{array}
    \right.
  \end{equation*}
  It is clear that $(\alpha^g_{t+1})_N(D')\leq m_N$ and $m_N$ is bounded.
  Hence, there exists a sequence of simple random variables $\phi_n$ such that $\phi_n\geq m_N$, and $\lim_{n\rightarrow\infty}\phi_n=m_N$.
  Therefore $(\alpha^g_{t+1})_N(D')\leq\phi_n$ for all $n$.
  Since $D'_t\leq 0$, then by locality of $\alpha^g_t$ and Lemma \ref{le:simpleconsist}, we have that $(\alpha^g_t)_N(D')\leq\phi_n$ for any $n\in\{1,2,\ldots,\}$.
  Moreover, it implies that $(\alpha^g_t)_N(D')\leq m_N$.
  Let $N$ go to infinity, we conclude that $\alpha^g_t(D')\leq m$.
\end{proof}

Before proceed to the next section, we give several examples of DAIs generated by $g$-expectations.
While discussing these examples, we take $W$ as a symmetric random walk.

\begin{example}[Quasi-concave Case]
  Let $g=(g_x)_{x>0}$ be in the form of $g_x(t,z)=\frac{x}{(x+1)}\ln(\frac{1}{3}+\frac{1}{3}e^{-z}+\frac{1}{3}e^{z})$.
  Similar to Example~\ref{ex:convexdriver2}, we have that each $g_x(t,z)$ a convex regular driver.
  For fixed $t, z$, $g_x(t,z)$ is increasing with respect to $x$.
  Therefore, $g$ satisfies Assumption G.
  Moreover, $\alpha^g_t(D)$ is a DAI.
\end{example}

\begin{example}[Coherent Case]
  Let $g=(g_x)_{x>0}$ be in the form of $g_x(t,z)=\frac{x}{x+1}|z|$.
  Then $g$ satisfies Assumption G.
  Moreover, $g_x(t,\cdot)$ is positive homogeneous.
  According to Lemma~\ref{le:dcai}, $\alpha^g_t(D)$ is a dynamic coherent acceptability index.
\end{example}

\begin{example}[Entropic Case]
Let $g=(g_x)_{x>0}$ be such that
$$
g_x(t,z)=\frac{1}{x\Delta\langle W\rangle_t}\ln(\frac{1}{2}e^{-xz}+\frac{1}{2}e^{xz}).
$$
Then, due to Example~\ref{ex:derm}, such family of drivers satisfies Assumption~G, and will generate a DAI, to which we refer as DAI corresponding to entropic risk measures.
\end{example}

\section{Dynamic Conic Finance as Market Model}

Cherny and Madan \cite{MadanCherny2010} proposed the conic finance framework for pricing non-dividend paying securities using static acceptability indices.
In \cite{BCIR2012}, the authors generalized such technique to a dynamic framework that allows cash flows to pay dividends and be subjected to transaction costs by using dynamic coherent acceptability indices that were obtained in \cite{BCZ2010}.
Nevertheless, in \cite{AcerbiScandolo2008} and \cite{RosazzaGianinSgarra2012}, the authors presented a systematic criticism to the positive homogeneity and sublinearity assumptions frequently adopted in the framework of coherent risk measures, and Bion-Nadal \cite{BionNadal2009a} introduced a dynamic approach to bid and ask prices taking into account both transaction costs and liquidity risk, based on Time Consistent Pricing Procedures.
In this section, we will build a market model in which the securities are priced by an acceptability method.
Such set-up is subjected to transaction costs and liquidity risk, by using the dynamic quasi-concave acceptability indices via $g$-expectations we developed earlier.

\subsection{Market Set-up}\label{se:marketsetup}

In this section we retain the same probabilistic framework and the same notations as in the previous sections.
In particular, we consider an underlying probability space $(\Omega,\cF,\bF=\{\cF_t\}_{t\in \cT},\bP)$, and we assume that all processes considered below are $\bF$-adapted and appropriately integrable.
 On this probability space we consider a market consisting of a banking account (or money market account) and $K$ securities.
Throughout, we pick the $T$ as the largest time horizon.
 We also adopt the convention that all security prices are already discounted with the banking account.
 Our market is further characterized as follows:

\begin{enumerate}[(M1)]
  \item
 The process $(D^{\text{ask},i}_t)_{t\in \cT}\in\cD$, with $D^{\ask,i}_0=0$, represents the dividend process associated with \textit{holding a long position of 1 share} of the $i$th security, $i=1,\ldots,K$.
 Correspondingly, $(D^{\text{bid},i}_t)_{t\in \cT}\in\cD$, with $D^{\bid,i}_0=0$, is the dividend process associated with \textit{holding a short position of 1 share} of security $i=1,\ldots,K$.
  \item
$P^{\text{ask}}_t(\varphi,D^{\ask/\bid,i})$, respectively $P^{\text{bid}}_t(\varphi,D^{\ask/\bid,i})$, denotes the ex-dividend prices of purchasing, respectively selling, $\varphi\in L^\infty_+(\cF_t)$ shares of cash flows $D^{\ask,i}$ or $D^{\bid,i}$ that are associated with security $i\in\set{1, \ldots,K}$ at time $t\in\cT$.
  We also assume that the pricing operators $P^{\ask}_t,P^{\bid}_t:L^\infty_+(\cF_t)\times\cD\rightarrow L^2(\cF_t)$ are such that $P^{\ask}_t(0,\cdot)=P^{\bid}_t(0,\cdot)=0$, $t\in\cT$.
  \item
    The dividend process associated with holding 1 unit of the banking account is given by processes $D^0=(0,\ldots,0,1)\in\cD$.
    Moreover, $\varphi\in L^\infty_+(\cF_t)$ units of the banking account are purchased/sold at time $t$ for price $P^{\text{ask}}_t(\varphi,D^0)=P^{\text{bid}}_t(\varphi,D^0)=\varphi\cdot 1=\varphi$.
    In particular, $P^{\text{ask}}_t(1,D^0)=P^{\text{bid}}_t(1,D^0)=1$.\footnote{This is consistent with our convention that all prices are discounted by the banking account, so that the price of one unit of the banking account is $1$ at any time $t\in\cT$.}
\end{enumerate}

\ref{m1} Let $\cM$ be the set of all dividend processes in this market, i.e.  $\cM=\set{D^0,D^{\ask/\bid,i}, i=1,\ldots,K}$.
We will use the notation $(\cM,P^{\ask},P^{\bid})$ to denote our market model.

The market model $(\cM,P^{\ask},P^{\bid})$ such that $D^{\ask,i}=D^{\bid,i}$ and $P^{\ask}_t(\varphi,D^i)=P^{\bid}_t(\varphi,D^i)$, for any $i=1,\ldots, K$, is called frictionless market model.

\smallskip

\begin{remark}\label{rem:nonhom}
In accordance with our  framework, it is generally assumed that $D^{\ask,i}\neq D^{\bid,i}$ and $P^{\ask}_t(\varphi,D^i)\neq P^{\bid}_t(\varphi,D^i)$.
We also remark that in this paper we do not postulate that the prices $P^{\text{ask}}_t(\varphi,D)$ and $ P^{\text{bid}}_t(\varphi,D)$ are homogeneous (of degree one) in $\varphi.$ In other words, we acknowledge the fact  that in practice the unit price of a security typically depends on the size of the position in the security (cf. Example \ref{ex:mktdata} below). This is due, for the most part, to market liquidity considerations. As we shall see later, the bid/ask prices generated by our acceptability method are not necessarily homogeneous.
\end{remark}

We now illustrate the processes introduced above in the context of dividend paying stock and Credit Default Swap (CDS) contract.

\begin{example}
  Denote by $S_T$ the fundamental value associated to 1 share of a dividend paying stock after dividend payment at time $T$.
  The dividend paid by 1 share of the stock at each time node $t=1,\ldots,T$, is denoted by $D_t$, regardless of what position the investor is in.
  Therefore, the dividend process associated with 1 share of the stock is
  $$
  D^{\ask}=D^{\bid}=\{0,D_1,\ldots,D_{T-1},D_T+S_T\}.
  $$
In this case, the ex-dividend ask and bid price process $P^{\text{ask}}$ and $P^{\text{bid}}$ are the market quoted prices for selling, respectively buying stock $S$; see also Example~\ref{ex:mktdata}.
\end{example}

\begin{example}
  A CDS contract is an agreement between the protection buyer and the protection seller, in which the protection buyer pays a regular fixed premium up to occurrence of a pre-specified credit event, in return, the seller promises a compensation to the buyer.
  Typically, CDS contracts are traded on over-the-counter markets in which dealers quote CDS spreads to investors.
  Consider a CDS contract that is initiated at $t=0$, expires at $t=T$ with nominal value $\cN$, the protection buyer pays a spread $\kappa^{\text{ask}}$ to the dealer at each time node in exchange for a compensation $\delta$ at default time $\tau$; the protection seller receives a spread $\kappa^{\text{bid}}$ from the dealer at each time node and he needs to pay $\delta$ to the dealer at $\tau$.
  Please note that $\kappa^{\ask}$, $\kappa^{\bid}$ and $\delta$ all depends on $\cN$, and such dependence is not necessarily linear.

  The dividend processes $D^{\text{ask}}$ and $D^{\text{bid}}$ associated to buying and selling the CDS with specifications above, respectively, satisfy
  \begin{align*}
    \sum^t_{s=0}D^{\text{ask}}_s:=\1_{\{\tau\leq t\}}\delta-\kappa^{\text{ask}}\sum^t_{s=1}\1_{\{s<\tau\}},\ \sum^t_{s=0}D^{\text{bid}}_t:=\1_{\{\tau\leq t\}}\delta-\kappa^{\text{bid}}\sum^t_{s=1}\1_{\{s<\tau\}},
  \end{align*}
  for $t=1,\ldots,T$. In this case, the ex-dividend ask and bid price process $P^{\text{ask}}$ and $P^{\text{bid}}$ specify the mark-to-market values of the CDS.
  In general, $P^{\text{ask}}$ and $P^{\text{bid}}$ are also not positively homogeneous with respect to $\cN$; a property confirmed by practitioners.
\end{example}

We close this subsection by illustrating the inhomogeneity of  prices in a order-driven market.

\begin{example}\label{ex:mktdata}
  In an order-driven market, orders to buy and sell are centralized in a limit order book available to market participants and orders are executed against the best available offers in the limit order book.
\begin{table}[!hbp]
  \centering
  \caption{Order Book of AAPL (Yahoo Finance 10:46AM EST 12/04/2014)}\label{tb:aapl}
\begin{tabular}{|c|c|c|c|}
  \hline
  Bid Price & Bid Size & Ask Price & Ask Size\\
  \hline
  116.59 & 400 & 116.61 & 200 \\
  \hline
  116.58 & 400 & 116.62 & 700 \\
  \hline
  116.57 & 800 & 116.63 & 543 \\
  \hline
  116.56 & 500 & 116.64 & 643 \\
  \hline
  116.55 & 543 & 116.65 & 343 \\
  \hline
\end{tabular}
\end{table}

Table \ref{tb:aapl} is the limit order book of Apple Inc (AAPL) publicly traded stock.
As we can see, there are up to 200 shares available for purchase at a price of \$116.61 per share. Hence, $P^{\text{ask}}(1)=116.61$, and for $0\leq\varphi\leq200$, $P^{\text{ask}}(\varphi)=\varphi P^{\text{ask}}(1)$.
Similarly, for $200<\varphi\leq900$, we have that $P^{\text{ask}}(\varphi)=200\cdot116.61+(\varphi-200)\cdot336.62>\varphi P^{\text{ask}}(1)$.
Thus, the ask price $P^{\text{ask}}(\cdot)$ is not homogenous in number of shares traded. Moreover, it is easy to note that $P^{\text{ask}}(\cdot)$ is a convex function.
Similarly, the function $P^{\text{bid}}(\cdot)$ is not homogenous and it is concave.
\end{example}

\subsection{Self-financing Trading Strategies and Arbitrage}

Due to nonlinearity of the prices and presence of transaction costs, the classical definition of self-financing trading strategy and arbitrage are not suitable for the market model proposed above. In this section we define the notion of self-financing trading strategy by using the general concept that a self-financing trading strategy is a trading strategy that does not allow injection or substraction of money during trading periods. Similarly, the notion of arbitrage is build on the idea that a self-financing trading strategy can not yield a riskless profit. Moreover, following market practice, we will allow that an investor can simultaneously have both long and short positions of the same  security at the same time, i.e. the long and short positions in the same security at the same time are not necessarily netted out.

\begin{definition}\label{def:self-fin}
A \textit{trading strategy} is a predictable process $\phi:=\big\{(\phi^0_t,\phi^{l,1}_t,\phi^{s,1}_t,\ldots,\phi^{l,K}_t,\phi^{s,K}_t)\big\}^T_{t=1}$, where $\phi^0_t\in~L^\infty(\cF_{t-1})$ is the number of units of banking account held from time $t-1$ to $t$; $\phi^{l,i}_t\in L^\infty_+(\cF_{t-1})$ is the number of shares in long position of security $i$ held from time $t-1$ to $t$; and $\phi^{s,i}_t\in L^\infty_+(\cF_{t-1})$ is the number of shares in short position of security $i$ held from time $t-1$ to $t$. Sometimes, we will use the notation $\phi^i_t=(\phi^{l,i}_t,\phi^{s,i}_t)$, for $i=1,\ldots,K.$
\end{definition}

\begin{definition}\label{def:vp}
Let $\phi$ be a trading strategy.
\begin{enumerate}[(V1)]
\item
The \textit{set-up cost process} $\widetilde{V}(\phi)$ associated with $\phi$ is defined as
\begin{align*}
\widetilde{V}_t(\phi)=\phi^0_{t+1}+\sum^K_{i=1}\Big(P^{\text{ask}}_t(\phi^{l,i}_{t+1}, D^{\text{ask},i})-P^{\text{bid}}_t(\phi^{s,i}_{t+1},D^{\text{bid},i})\Big),\ \quad t=0,\ldots,T-1.
\end{align*}
\item
The \textit{liquidation value process} $V(\phi)$ associated with $\phi$ is defined as
\begin{align*}
V_t(\phi)=\phi^0_t+\sum^K_{i=1}\Big(P^{\text{bid}}_t(\phi^{l,i}_t, D^{\text{ask},i})-P^{\text{ask}}_t(\phi^{s,i}_t, D^{\text{bid},i})\Big)\\
  +\sum^K_{i=1}\Big(\phi^{l,i}_tD^\text{ask,i}_t-\phi^{s,i}_tD^\text{bid,i}_t\Big),\ \quad t=1,\ldots,T.
\end{align*}
\end{enumerate}
\end{definition}

For each $t\in\cT$, an investor could have both long and short positions of a security at the same time.
The process $\widetilde{V}(\phi)$ represents the cost of setting up the portfolio $\phi$,
and $V_t(\phi)$ is interpreted as the liquidation value of the portfolio at time $t$ (before any time $t$ transactions), including any dividends acquired from time $t-1$ to time $t$.
In classical frictionless market, the liquidation value of a portfolio is the same as the set-up cost.
However in our case, generally speaking, $V_t(\phi)\neq \widetilde{V}_t(\phi)$ due to transaction costs and non-homogeneity of $P^{\ask}$ and $P^{\bid}$.

\begin{definition}\label{def:sf}
  A trading strategy $\phi$ is \textit{self-financing} if
  \begin{align}
    &\Delta\phi^0_{t+1}+\sum^K_{i=1}\Big(\1_{\Delta\phi^{l,i}_{t+1}\geq0}P^{\ask}_t(\Delta\phi^{l,i}_{t+1},D^{\ask,i})-\1_{\Delta\phi^{l,i}_{t+1}<0}P^{\bid}_t(-\Delta\phi^{l,i}_{t+1},D^{\ask,i})\\
    &-\1_{\Delta\phi^{s,i}_{t+1}\geq0}P^{\bid}_t(\Delta\phi^{s,i}_{t+1},D^{\bid})+\1_{\Delta\phi^{s,i}_{t+1}<0}P^{\ask}_t(-\Delta\phi^{s,i}_{t+1},D^{\bid})\Big) \nonumber\\
    & \quad = \sum^K_{i=1}(\phi^{l,i}_tD^{\text{ask},i}_t-\phi^{s,i}_tD^{\text{bid},i}_t). \label{eq:self-fin}
  \end{align}
  for $t=0,\ldots,T-1$.
\end{definition}

Definition \ref{def:sf} provides a natural interpretation of self-financing condition in market with friction.
The cash flows that are being bought or sold should depend on the both positions before and after re-balance at each time $t$.
All the money that is used for getting to the new position is equal to the dividends acquired from time $t-1$ to time $t$.
Therefore, no money flows in or out of the portfolio.

Next, we will introduce the concept of arbitrage that is relevant for our theory.

\begin{definition}\label{def:arb1}   An \textit{arbitrage opportunity} at time $t$, $t\in\cT$, is a self-financing strategy, such that  $V_T(\phi)-\widetilde V_t(\phi)\geq0$ and $\bP(V_T(\phi)-\widetilde V_t(\phi)>0)>0$.  We call a market \textit{arbitrage free} at time $t$ if there exists no arbitrage opportunity in the model at time $t$.
\end{definition}

In what follows, we will provide two other characterizations of arbitrage opportunities, which are useful in our paper.
Towards this end, we first define the following sets,

\begin{align}
&\cS(t)=\cS(t, P^{\ask}, P^{\bid}) :=
\begin{cases}
\{ \phi : \phi \text{ is self-financing, } \widetilde{V}_0(\phi)=0 \}, & t=0, \nonumber \\[0.05in]
\{ \phi : \phi \text{ is self-financing, } \phi_s=0 \textrm{ for all }  s\leq t\}, & t=1,\ldots,T. \nonumber
\end{cases}
\end{align}

Note that for $\phi\in\cS(t)$, we have $\phi_t=0$. Due to our assumption that $P^{\text{bid}}_t(0,\cdot)=P^{\text{ask}}_t(0,\cdot)=0$, we have that $\widetilde{V}_t(\phi)=V_t(\phi)=0$.
Next, we introduce the set of cash flows generated by strategies in $\cS(t)$:
\begin{align}\label{eq:H0}
  \cH^0(t)=\Big\{\Big(\underbrace{0,\ldots,0}_{t},\Delta V_{t+1}(\phi),\ldots,\Delta V_T(\phi)\Big):\phi\in\cS(t)\Big\},\ \quad t\in\cT.
\end{align}
Since $P^{\ask}_t,P^{\bid}_t:L^\infty_+(\cF_t)\times\cD\rightarrow \cD$, we have that $\widetilde{V}_t(\phi)$, $V_t(\phi)\in L^2(\cF_t)$ for any self-financing trading strategy $\phi$, and therefore $\cH^0(t)\subset\cD$.

Next result gives two characterization of arbitrage.

\begin{proposition}\label{pr:eqarb}
  The following statements are equivalent:
  \begin{enumerate}[(1)]
    \item
    There exists an arbitrage opportunity at time $t$.
    \item
    There exists a strategy $\psi\in\cS(t)$, such that $V_T(\psi)\geq0$ and $\bP(V_T(\psi)>0)>0$.
    \item
    There exists a cash flow $({0,\ldots,0},H_{t+1},\ldots,H_T)\in\cH^0(t)$, such that $\sum^T_{s=t+1}H_s\geq0$ and $\bP(\sum^T_{s=t+1}H_s>0)>0$.
  \end{enumerate}
\end{proposition}

\begin{proof}
  For  a fixed $t\in\cT$, we will show that (1) is equivalent to (2), and (2) is equivalent to (3).

\smallskip

\noindent   $(2)\Rightarrow(1)$ Assume that there exists $\psi\in\cS(t)$ such that $V_T(\psi)\geq0$ and $\bP(V_T(\psi)>0)>0$. Since $\psi\in\cS(t)$, then $\widetilde{V}_t(\psi)=0$, and (2) follows immediately.

\smallskip  \noindent $(1)\Rightarrow(2)$
Assume that $\phi$ is gives an arbitrage opportunity at time $t$.
  We define a trading strategy $\psi$ as follows
  \begin{align}
\psi^0_u=\psi^{l,i}_u=\psi^{s,i}_u=0,\ \quad & u=0,\ldots,t,\ i=0,\ldots,K, \nonumber \\[0.05in]
\psi^0_u=\phi^0_u-\widetilde{V}_t(\phi),\ \quad & u=t+1,\ldots,T, \nonumber \\[0.05in]
\psi^{l/s,i}_u=\phi^{l/s,i}_u,\ \quad & u=t+1,\ldots,T,\ i=1,\ldots,K. \nonumber
\end{align}
It is straightforward to see that $\widetilde{V}_0(\psi)=0$, and $\psi_u=0$ for $u\leq t$.
To show that $\psi$ is self-financing, first notice that
\begin{align}
    \Delta\psi^0_{t+1}&+\sum^K_{i=1}\Big(\1_{\Delta\psi^{l,i}_{t+1}\geq0}P^{\ask}_t(\Delta\psi^{l,i}_{t+1},D^{\ask,i})-1_{\Delta\psi^{l,i}_{t+1}<0}P^{\bid}_t(-\Delta\psi^{l,i}_{t+1},D^{\ask,i}) \nonumber \\
    &-\1_{\Delta\psi^{s,i}_{t+1}\geq0}P^{\bid}_t(\Delta\psi^{s,i}_{t+1},D^{\bid,i})+\1_{\Delta\psi^{s,i}_{t+1}<0}P^{\ask}_t(-\Delta\psi^{s,i}_{t+1},D^{\bid,i})\Big) \nonumber \\
    =&\phi^0_{t+1}-\widetilde{V}_t(\phi)+\sum^K_{i=1}\Big(P^{\ask}_t(\phi^{l,i}_{t+1},D^{\ask,i})-P^{\bid}_t(\phi^{s,i}_{t+1},D^{\bid,i})\Big) \nonumber \\
    =&\widetilde{V}_t(\phi)-\widetilde{V}_t(\phi)=0 \nonumber \\
    =&\sum^K_{i=1}(\psi^{l,i}_tD^{\text{ask},i}_t-\psi^{i,s}_tD^{\text{bid},i}_t).\label{eq:sft}
  \end{align}
  Since $\Delta\psi_u=\Delta\phi_u$ for any $u\geq t+1$, and $\phi$ is a self-financing trading strategy, then also in view of \eqref{eq:sft}, we conclude that $\psi$ is a self-financing strategy.
  Hence, $\psi\in\cS(t)$, and $V_T(\psi)=V_T(\phi)-V_t(\phi)$ which implies that $V_T(\psi)\geq0$, $\bP(V_T(\psi)>0)>0$.

 \smallskip \noindent
 $(2)\Rightarrow(3)$ Assume that $\psi\in\cS(t)$, and $V_T(\psi)\geq0$, $\bP(V_T(\psi)>0)>0$. Then, we define the process $H$ as follows
  \begin{align}
  &H_s :=
  \begin{cases}
  0, & s=0,\ldots,t, \nonumber \\
  \Delta V_s(\psi), & s=t+1, \ldots, T. \nonumber
  \end{cases}
  \end{align}
  Then $H\in\cH^0(t)$, $\sum^T_{s=t+1}H_s=\sum^T_{s=0}\Delta V_s(\psi)=V_T(\psi)\geq0$, and thus $\bP(\sum^T_{s=t+1}H_s>0)=\bP(V_T(\psi)>0)>0$.

\smallskip \noindent
  $(3)\Rightarrow(2)$ Now, suppose that there exists a cash flow $\big(0,\ldots,0,\widehat H_{t+1},\ldots,\widehat H_T\big)\in\cH^0(t)$ such that $\sum^T_{s=t+1}\widehat H_s\geq0$ and
  $\bP(\sum^T_{s=t+1}\widehat H_s>0)>0$.
  Then, by definition of $\cH^0(t)$ there exists a $\psi\in\cS(t)$ such that $V_t(\psi)=0$, $\Delta V_s(\psi)=\widehat H_s$, $s\in\{t+1,\ldots,T\}$ and
  $$
  V_T(\psi)=\sum^T_{s=0}\Delta V_s(\psi)=\sum^T_{s=t+1}\widehat H_s\geq0.
  $$
  Moreover, $\bP(V_T(\psi)>0)=\bP(\sum^T_{s=t+1}\widehat H_s>0)>0$.
  Thus,  (3) holds true.

This concludes the proof.
\end{proof}

 With the results from  Proposition~\ref{pr:eqarb} at hand, we say that the no-arbitrage condition (NA) for $\cH^0(t)$ holds if (3) does not hold, and throughout this section, we will characterize arbitrage opportunities by properties (2) or (3) as in Proposition~\ref{pr:eqarb}.

 Clearly, since $\cS(t+1)\subset \cS(t)$, absence of arbitrage at time $t\in\set{0,1,\ldots,T-1}$ implies absence of arbitrage at any future times $s$, $s=t+1,\ldots,T-1$.
 In particular, if a market is arbitrage free at time $0$, then such market is arbitrage free at any time $t\in\cT$.
 Hence, to show that there is no arbitrage opportunity in the market, it is enough to show no arbitrage at time 0.
 Accordingly, we have the following definition.

 \begin{definition}
   A market is called arbitrage free if there exists no arbitrage opportunity in the model at time 0.
 \end{definition}

 It is important to observe though, that, contrary to the  classical frictionless market model, absence of arbitrage at time $t=1,\ldots,T-1,$ in models considered here does not (in general) imply absence of arbitrage at time $s$, where $s< t$.
 This will be illustrated in the following example.

\begin{example}\label{example2}
  Let $T=2$ ,$\Omega=\{\omega_1,\omega_2,\omega_3,\omega_4\}$, and we consider a market with one banking account and one security paying no dividend.
  Assume that the pricing operator is homogeneous with respect to the number shares traded, and the price process of the security is given in Table~\ref{tb:stockprice}.

  \begin{table}[!h]
  \centering
  \caption{Stock Price Dynamics, Example~\ref{example2}.}\label{tb:stockprice}
\begin{tabular}{|c|c|c|c|c|}
  \hline
   & $P^{\ask}(\omega_1)/P^{\bid}(\omega_1)$ & $P^{\ask}(\omega_2)/P^{\bid}(\omega_2)$ & $P^{\ask}(\omega_3)/P^{\bid}(\omega_3)$ & $P^{\ask}(\omega_4)/P^{\bid}(\omega_4)$\\
  \hline
  $T=0$ & 10/10 & 10/10 & 10/10 & 10/10 \\
  \hline
  $T=1$ & 12/11 & 12/11 & 11/10 & 11/10 \\
  \hline
  $T=2$ & 13/12 & 11/10 & 12/11 & 10/9 \\
  \hline
\end{tabular}
\end{table}
  Consider the trading strategy $\psi_1=(-10,1)$, $\psi_2(\omega_1,\omega_2)=(1,0)$ and $\psi_2(\omega_3,\omega_4)=(0,0)$. Thus, $\psi$ is a self-financing strategy such that $\widetilde{V}_0=0$, $V_1(\omega_1,\omega_2)=V_2(\omega_1,\omega_2)=1$ and $V_1(\omega_3,\omega_4)=V_2(\omega_3,\omega_4)=0$. According to Definition \ref{def:arb1}, it is an arbitrage opportunity at time $0$.
  However, it is not hard to observe that for any $\varphi\in\cS(1)$, $\varphi$ could not be an arbitrage opportunity at time $1$.
\end{example}

\subsection{Pricing Operators}\label{sec:PricingOperators}

In this section we will introduce some pricing operators for cash flows $D\in\cD$ through an acceptability method and show some properties of such prices.
Then, in the next section, we will show that a market, in which the fundamental assets are priced according to our pricing operators, satisfies the properties postulated in (M2) and (M3), and this market is arbitrage free in the sense of Definition \ref{def:arb1}.

For any random variable $a$ that is $\cF_t$-measurable, and process $X=(X_t)_{t\in\cT}$, we will use the following notation:
\begin{align*}
  \delta_t(a)&=\1_{\{t\}}a:=\{0,\ldots,0,a,0,\ldots,0\},\\
  \delta^+_t(X)&=\{0,\ldots,0,X_{t+1},\ldots,X_T\}.
\end{align*}

For any $D\in\cD$, $\delta^+_t(D)$ represents the the future cash flow, at time  $t\in\cT$, of the dividend stream $D$.
We are going to evaluate this future dividend cash flow, in other words, to calculate the ex-dividend prices of $D$, by using an acceptability based method.

Assume that an investor wants to buy $\varphi$ shares of cash flow $D\in\cD$ at time $t\in\cT$, where $\varphi\in\cF_{t}$, $\varphi\geq0$, then the market as the counterparty will charge $P^{\text{ask}}_t(\varphi,D)$ at time $t$ and deliver $\delta^+_t(D)$ to the buyer.
Thus, the corresponding cash flow for the market is $\{0,\ldots,0,P^{\text{ask}}_t(\varphi,D),-\varphi D_{t+1},\ldots,-\varphi D_T\}$.
To decide the proper $P^{\text{ask}}_t(\varphi,D)$, the market will choose the smallest $P^{\text{ask},D}_t(\varphi)$ such that\\
$\{0,\ldots,0,P^{\text{ask}}_t(\varphi,D),-\varphi D_{t+1},\ldots,-\varphi D_T\}$ is acceptable with respect to some acceptability index $\alpha^g$ at some level $\gamma$.
Similarly, the market will choose the largest $P^{\text{bid}}_t(\varphi,D)$ such that the cash flow\\
$\{0,\ldots,0,-P^{\text{bid}}_t(\varphi,D),\varphi D_{t+1},\ldots,\varphi D_T\}$ is acceptable at level $\gamma$, when an investor is selling $\varphi\geq0$ shares of $D\in\cD$ at time $t\in\cT$.

Throughout this section, we will always consider the family of drivers $g=(g_x)_{x>0}$ satisfy Assumption G.
We proceed by defining the acceptability ask price $a^{g,\gamma}_t$ and the acceptability bid price $b^{g,\gamma}_t$.

\begin{definition}\label{def:ask&bid}
  Let $g=(g_x)_{x>0}$ be a family of drivers.
  The \textit{acceptability ask price} of $\varphi\in L^\infty_+(\cF_t)$ shares of the cash flow $D\in\cD$, at level $\gamma$, at time $t\in\cT$ is defined as
  \begin{equation}\label{eq:a1}
  a^{g,\gamma}_t(\varphi, D)=\essinf\{a\in\cF_t:\alpha^g_t(\delta_t(a)-\delta^+_t(\varphi D))\geq\gamma\};
  \end{equation}
  and the \textit{acceptability bid price} of $\varphi\in L^\infty_+(\cF_t)$ shares of $D\in\cD$, at level $\gamma$, at time $t\in\cT$ is defined as
  \begin{equation}\label{eq:b1}
  b^{g,\gamma}_t(\varphi, D)=\esssup\{b\in\cF_t:\alpha^g_t(\delta^+_t(\varphi D)-\delta_t(b))\geq\gamma\}.
  \end{equation}
\end{definition}

For $\varphi\in L^\infty_+(\cF_t)$ and $D\in\cD$, $a^{g,\gamma}_t(\varphi,D)$ and $b^{g,\gamma}_t(\varphi,D)$ are the ex-dividend prices at time $t$, therefore, they are independent of $D_0,\ldots,D_t$.

\begin{remark}
  Note that in Definition \ref{def:ask&bid}, $\varphi$ is a $\cF_t$-measurable random variable, thus by applying the pricing operators $a^{g,\gamma}_t$ and $b^{g,\gamma}_t$ to cash flows generated by any trading strategy $\phi$, we will get well-defined set-up cost process $\widetilde{V}_t(\phi)$ and liquidation value process $V_t(\phi)$.
  Also by observing the fact that $a^{g,\gamma}_t(\varphi, D)=a^{g,\gamma}_t(1,\varphi D)$ and $b^{g,\gamma}_t(\varphi, D)=b^{g,\gamma}_t(1,\varphi D)$, we will prove most results for $a^{g,\gamma}_t(1,D)$ and $b^{g,\gamma}_t(1,D)$.
  Then such results are also true for $a^{g,\gamma}_t(\varphi,D)$ and $b^{g,\gamma}_t(\varphi,D)$.
\end{remark}

\begin{remark}
  We call $a^{g,\gamma}_t(1,D)$ the time $t$ acceptability ask price of $D$ at level $\gamma$, and $b^{g,\gamma}_t(1,D)$ the time $t$ acceptability bid price of $D$ at level $\gamma$.
  For simplicity, we will use the notation $a^{g,\gamma}_t(D)=a^{g,\gamma}_t(1,D)$ and $b^{g,\gamma}_t(D)=b^{g,\gamma}_t(1,D)$.
\end{remark}

\begin{remark}
   We observe that in Definition \ref{def:ask&bid} $\varphi$ is non-negative.
 This is consistent with the term ``buy/sell $\varphi$ shares of some security'' which is used in practice.
\end{remark}

Next, we provide some important properties of acceptability ask and bid prices.

\begin{theorem}\label{th:askbid1}
  The acceptability ask and bid prices of $D\in\cD$, at level $\gamma>0$, at time $t\in\cT$, satisfy the following properties:
  \begin{enumerate}[P1.]
  \item
  Representation:
  \begin{align*}
    a^{g,\gamma}_t(D)&=\cE_{g_\gamma}\Big[\sum^T_{s=t+1}D_s\Big|\cF_t\Big],\\
    b^{g,\gamma}_t(D)&=-\cE_{g_\gamma}\Big[-\sum^T_{s=t+1}D_s\Big|\cF_t\Big].
  \end{align*}
  \item
  No Arbitrage:
  $$a^{g,\gamma}_t(D)\geq b^{g,\gamma}_t(D).$$
  \item
  Convexity and Concavity:
  \begin{align*}
    a^{g,\gamma}_t(\lambda D^1+(1-\lambda)D^2)&\leq\lambda a^{g,\gamma}_t(D^1)+(1-\lambda)a^{g,\gamma}_t(D^2),\\
    b^{g,\gamma}_t(\lambda D^1+(1-\lambda)D^2)&\geq\lambda b^{g,\gamma}_t(D^1)+(1-\lambda)b^{g,\gamma}_t(D^2),
  \end{align*}
  for $D^1, D^2\in\cD$, $\lambda\in L^\infty_+(\cF_t)$, $0\leq\lambda\leq1$.
  \item
  Market Impact:
  \begin{align*}
    &a^{g,\gamma}_t(\lambda\varphi,D)\leq\lambda a^{g,\gamma}_t(\varphi,D),\ b^{g,\gamma}_t(\lambda\varphi,D)\geq\lambda b^{g,\gamma}_t(\varphi,D),\ \lambda, \varphi\in L^\infty_+(\cF_t),\quad  0\leq\lambda\leq1;\\
    &a^{g,\gamma}_t(\lambda\varphi,D)\geq\lambda a^{g,\gamma}_t(\varphi,D),\ b^{g,\gamma}_t(\lambda\varphi,D)\leq\lambda b^{g,\gamma}_t(\varphi,D),\ \lambda, \varphi\in L^\infty_+(\cF_t),\quad \lambda\geq1.
  \end{align*}
  \item
  Time Consistency:
  \begin{align*}
    a^{g,\gamma}_t(D)&=a^{g,\gamma}_t(\delta_{t+1}(D_{t+1}+a^{g,\gamma}_{t+1}(D))),\\
    b^{g,\gamma}_t(D)&=b^{g,\gamma}_t(\delta_{t+1}(D_{t+1}+b^{g,\gamma}_{t+1}(D))).
  \end{align*}
  \item
  Linearity (if the drivers are linear):
  If $g_\gamma(t,z)=x(t)z$, $t\in\cT$.
  Then, there exists a probability measure $\bQ\sim\bP$ such that
  $$a^{g,\gamma}_t(\varphi,D)=b^{g,\gamma}_t(\varphi,D)=\varphi\bE_{\bQ}\big[\sum^T_{s=t+1}D_s|\cF_t\big],$$
  for $\varphi\in L^\infty_+(\cF_t)$.
  \end{enumerate}
\end{theorem}

\begin{proof}
We will prove the results only for acceptability ask prices; the case of acceptability bid prices is treated similarly.

\smallskip\noindent
P1.
    Due to G3 of Assumption G, by similar arguments as in Theorem~\ref{th:gdai}, we get that $\alpha^g_t(X)\geq\gamma$ for $\gamma>0$ is equivalent to the fact that $\rho^{g_{\gamma}}_t(X)\leq0$.
  Also, in view of the definition of acceptability ask price, we have
  \begin{align*}
  a^{g,\gamma}_t(D)&=\essinf\{a\in\cF_t:\alpha^g_t(\delta_t(a)-\delta^+_t(D))\geq\gamma\}\\
  &=\essinf\left\{a\in\cF_t:\cE_{g_\gamma}\Big[-a+\sum^T_{s=t+1}D_s\Big|\cF_t\Big]\leq0\right\}\\
  &=\essinf\left\{a\in\cF_t:a\geq\cE_{g_\gamma}\Big[\sum^T_{s=t+1}D_s\Big|\cF_t\Big]\right\}\\
  &=\cE_{g_\gamma}\Big[\sum^T_{s=t+1}D_s\Big|\cF_t\Big].
  \end{align*}

\smallskip \noindent P2.
 % We show that $a^{g,\gamma}_t(D)-b^{g,\gamma}_t(D)\geq0$.
  %Due to P1., it is equivalent to show that $\cE_{g_\gamma}[\sum^T_{s=t+1}D_s|\cF_t]+\cE_{g_\gamma}[-\sum^T_{s=t+1}D_s|\cF_t]\geq0$.
  By convexity of $g$-expectation, we have that
  $$
  \frac{1}{2}\cE_{g_\gamma}\Big[\sum^T_{s=t+1}D_s\Big|\cF_t\Big]+\frac{1}{2}\cE_{g_\gamma}\Big[-\sum^T_{s=t+1}D_s\Big|\cF_t\Big]\geq\cE_{g_\gamma}\Big[\frac{1}{2}\Big(\sum^T_{s=t+1}D_s-\sum^T_{s=t+1}D_s\Big)\Big|\cF_t\Big]=0.
  $$
  Hence, due to property P1, we get $a^{g,\gamma}_t(D)\geq b^{g,\gamma}_t(D)$.

\smallskip \noindent P3. Property P3 follows from convexity of $g_\gamma(t,\cdot)$, convexity of the $g$-expectation, and P1.

\smallskip \noindent P4. By taking $D^1=D$, and $D^2=0$,
  for $\lambda,\varphi\in L^\infty(\cF_t)$, $0\leq\lambda\leq1$, we have
   $a^{g,\gamma}_t(\lambda\varphi D)\leq\lambda a^{g,\gamma}_t(\varphi D)$.
  Since $a^{g,\gamma}_t(\lambda\varphi D)=a^{g,\gamma}_t(\lambda\varphi,D)$ and $a^{g,\gamma}_t(\varphi D)=a^{g,\gamma}_t(\varphi,D)$, we immediately get that $a^{g,\gamma}_t(\lambda\varphi,D)\leq\lambda a^{g,\gamma}_t(\varphi,D)$.
  %Similarly, we could show that $b^{g,\gamma}_t(\lambda\varphi,D)\geq\lambda b^{g,\gamma}_t(\varphi,D)$.

  For $\lambda\in L^\infty(\cF_t)$, $\lambda\geq1$, $D''\in\cD$, we have that $a^{g,\gamma}_t(\varphi,\frac{D''}{\lambda})=a^{g,\gamma}_t(\frac{\varphi}{\lambda},D'')\leq\frac{1}{\lambda}a^{g,\gamma}_t(\varphi,D'')$.

\smallskip \noindent P5.
  According to P1 and Proposition \ref{pr:gexp4}~(iii), we have that
  \begin{align*}
    a^{g,\gamma}_t(\delta_{t+1}(D_{t+1}+a^{g,\gamma}_{t+1}(D)))&=\cE_{g_\gamma}\Big[D_{t+1}+\cE_{g_\gamma}\Big[\sum^T_{s=t+2}D_s\Big|\cF_{t+1}\Big]\Big|\cF_t\Big]\\
    &=\cE_{g_\gamma}\Big[\cE_{g_\gamma}\Big[\sum^T_{s=t+1}D_s\Big|\cF_{t+1}\Big]\Big|\cF_t\Big]
    =\cE_{g_\gamma}\Big[\sum^T_{s=t+1}D_s\Big|\cF_{t}\Big]\\
    &=a^{g,\gamma}_t(D).
  \end{align*}

 \smallskip \noindent P6.
  By Proposition \ref{pr:lineardriver}, there exists a probability measure $\bQ\sim\bP$, such that $\cE_{g_\gamma}[X|\cF_t]=\bE_\bQ[X|\cF_t]$ for all $X\in L^2(\cF_T)$.
  Then, using P1, for $\varphi\in\cF_t$, we obtain
  \begin{align*}
    a^{g,\gamma}_t(\varphi,D)&=a^{g,\gamma}_t(1,\varphi D)=\cE_{g_\gamma}[\varphi\sum^T_{s=t+1}D_s|\cF_t]=\bE_\bQ[\varphi\sum^T_{s=t+1}D_s|\cF_t]=\varphi\bE_\bQ[\sum^T_{s=t+1}D_s|\cF_t] .
    %b^{g,\gamma}_t(\varphi,D)&=b^{g,\gamma}_t(1,\varphi D)=-\cE_{g_\gamma}[-\varphi\sum^T_{s=t+1}D_s|\cF_t]=-\bE_\bQ[-\varphi\sum^T_{s=t+1}D_s|\cF_t]=\varphi\bE_\bQ[\sum^T_{s=t+1}D_s|\cF_t].
  \end{align*}
  Hence, $a^{g,\gamma}_t(\varphi,D) = \varphi\alpha_t^{g,\gamma}(1,D)$.

  This concludes the proof.

\end{proof}

\begin{remark}
By Property P4 in Theorem \ref{th:askbid1}.P4 we have that $a^{g,\gamma}_t(\lambda, D)\geq\lambda a^{g,\gamma}_t(D)$, for any $\lambda\geq1$.
  This indicates that if an investor is buying a cashflow, the price moves up against the buyer (the effect of market impact).
  Similar argument follows for the bid price.
  Such property of acceptability ask and bid prices is consistent with real market quotes, as it was shown for equity markets in Example \ref{ex:mktdata}.
\end{remark}

In case of classical risk-neutral pricing the discounted cumulative dividend prices of cashflows are martingales under an equivalent martingale measure $\bQ$.
In our pricing framework, a similar `martingale property' also holds true as shown in the next result.
For a given dividend stream $D\in\cD$, we define the \textit{acceptability cumulative dividend prices} at time $t$ as follows
\begin{align*}
a^{\text{cld},g,\gamma}_t(D)& :=\sum^t_{s=0}D_s+a^{g,\gamma}_t(D), \\
b^{\text{cld},g,\gamma}_t(D)& :=\sum^t_{s=0}D_s+b^{g,\gamma}_t(D).
\end{align*}
As an immediate consequence of Theorem~\ref{th:askbid1}.P5, we obtain.

\begin{corollary}\label{pr:tc}
  The acceptability ask and bid cumulative dividend prices of a cash flow $D$, at level $\gamma>0$ satisfy
  \begin{align*}
    a^{\text{cld},g,\gamma}_t(D)=a^{g,\gamma}_t(\delta_{t+1}(a^{\text{cld},g,\gamma}_{t+1}(D))),\\
    b^{\text{cld},g,\gamma}_t(D)=b^{g,\gamma}_t(\delta_{t+1}(b^{\text{cld},g,\gamma}_{t+1}(D))).
  \end{align*}
\end{corollary}

\begin{remark}
  This proposition is a counterpart of time consistency in case of linear pricing. The time $t$ cumulative dividend price of $D$ is equal to evaluating time $t+1$ cumulative dividend price at time $t$.
  We call such property the time consistency of acceptability ask/bid prices.
\end{remark}

So far, we have showed that if the market picks the same level $\gamma>0$ for the given family of drivers $g$ on both buying and selling side, then we have some nice properties of the ask and bid prices.
On the other hand, in general, market participants will choose different acceptability levels or even acceptability indices (different family of drivers) for buying and/or selling. In the rest of this section, we will provide some results regarding this possibilities.

\begin{proposition}\label{pr:ag1bg2}
  Let $g^1$ and $g^2$ be two families of drivers.
  Then, $a^{g^1,\gamma_1}_t(D)\geq b^{g^2,\gamma_2}_t(D)$, for any $D\in\cD$, $t\in\cT$, $\gamma_1,\gamma_2>0$.
\end{proposition}

\begin{proof}  We will prove the statement recursively, backward in time component.

  Let $A_t:=\{\omega\in\Omega: g^1_{\gamma_1}(\omega,t,z)\geq g^2_{\gamma_2}(\omega,t,z), z\in\bR\}$, $t\in\cT\setminus\{0\}$.
  since the drivers are predictable, both $A_t$ and $A^c_t$ are $\cF_{t-1}$ measurable, .

  By definition of $A_T$, we have that $\1_{A_T}g^1_{\gamma_1}(T,z)\geq\1_{A_T}g^2_{\gamma_2}(T,z)$ and $\1_{A^c_T}g^1_{\gamma_1}(T,z)\leq\1_{A^c_T}g^2_{\gamma_2}(T,z)$, for all $z\in\bR$. Hence, in view of Theorem \ref{th:comp} and Theorem \ref{th:askbid1}, we get that
  \begin{align*}
  \1_{A_T}\cE_{g^1_{\gamma_1}}[D_T|\cF_{T-1}]\geq\1_{A_T}\cE_{g^2_{\gamma_2}}[D_T|\cF_{T-1}]\geq-1_{A_T}\cE_{g^2_{\gamma_2}}[-D_T|\cF_{T-1}],\\
  \1_{A^c_T}\cE_{g^1_{\gamma_1}}[D_T|\cF_{T-1}]\geq-\1_{A^c_T}\cE_{g^1_{\gamma_1}}[-D_T|\cF_{T-1}]\geq-1_{A^c_T}\cE_{g^2_{\gamma_2}}[-D_T|\cF_{T-1}].
  \end{align*}
  Therefore, $\1_{A_T}a^{g^1,\gamma_1}_{T-1}(D)\geq\1_{A_T}b^{g^2,\gamma_2}_{T-1}(D)$ and $\1_{A^c_T}a^{g^1,\gamma_1}_{T-1}(D)\geq\1_{A^c_T}b^{g^2,\gamma_2}_{T-1}(D)$, and thus the statement holds true for $t=T$.

Note that  by definition of acceptability cumulative dividend prices, we have that
  \begin{align*}
    a^{g^1,\gamma_1}_{T-2}(D)+\sum^{T-2}_{s=1}D_s=a^{\text{cld},g^1,\gamma_1}_{T-2}(D),\\
    b^{g^2,\gamma_2}_{T-2}(D)+\sum^{T-2}_{s=1}D_s=b^{\text{cld},g^2,\gamma_2}_{T-2}(D).
  \end{align*}
  In view of Proposition \ref{pr:tc}, we also have that that
  \begin{align*}
    a^{\text{cld},g^1,\gamma_1}_{T-2}(D)=\cE_{g^1_{\gamma_1}}[a^{g^1,\gamma_1}_{T-1}(D)+\sum^{T-1}_{s=1}D_s|\cF_{T-2}]=\cE_{g^1_{\gamma_1}}[a^{g^1,\gamma_1}_{T-1}(D)+D_{T-1}|\cF_{T-2}]+\sum^{T-2}_{s=1}D_s,\\
    b^{\text{cld},g^2,\gamma_2}_{T-2}(D)=-\cE_{g^2_{\gamma_2}}[-b^{g^2,\gamma_2}_{T-1}(D)-\sum^{T-1}_{s=1}D_s|\cF_{T-2}]=-\cE_{g^2_{\gamma_2}}[-b^{g^2,\gamma_2}_{T-1}(D)-D_{T-1}|\cF_{T-2}]+\sum^{T-2}_{s=1}D_s.
  \end{align*}
  Thus, $a^{g^1,\gamma_1}_{T-2}(D)=\cE_{g^1_{\gamma_1}}[a^{g^1,\gamma_1}_{T-1}(D)+D_{T-1}|\cF_{T-2}]$ and $b^{g^2,\gamma_2}_{T-2}(D)=\cE_{g^2_{\gamma_2}}[b^{g^2,\gamma_2}_{T-1}(D)+D_{T-1}|\cF_{T-2}]$.
  In view of above, we have $a^{g^1,\gamma_1}_{T-1}(D)+D_{T-1}\geq b^{g^2,\gamma_2}_{T-1}(D)+D_{T-1}$.
  Consequently, applying again the comparison Theorem~\ref{th:comp} and Theorem~\ref{th:askbid1}, we deduce \begin{align*}
    \1_{A_{T-1}}\cE_{g^1_{\gamma_1}}[a^{g^1,\gamma_1}_{T-1}(D)+D_{T-1}|\cF_{T-2}]&\geq\1_{A_{T-1}}\cE_{g^2_{\gamma_2}}[b^{g^2,\gamma_2}_{T-1}(D)+D_{T-1}|\cF_{T-2}]\\
    &\geq-\1_{A_{T-1}}\cE_{g^2_{\gamma_2}}[-b^{g^2,\gamma_2}_{T-1}(D)-D_{T-1}|\cF_{T-2}],
  \end{align*}
  and
  \begin{align*}
    \1_{A^c_{T-1}}\cE_{g^1_{\gamma_1}}[a^{g^1,\gamma_1}_{T-1}(D)+D_{T-1}|\cF_{T-2}]&\geq-\1_{A^c_{T-1}}\cE_{g^1_{\gamma_1}}[-a^{g^1,\gamma_1}_{T-1}(D)-D_{T-1}|\cF_{T-2}]\\
    &\geq-\1_{A^c_{T-1}}\cE_{g^1_{\gamma_1}}[-b^{g^2,\gamma_2}_{T-1}(D)-D_{T-1}|\cF_{T-2}].
  \end{align*}
  Therefore, $a^{g^1,\gamma_1}_{T-2}(D)\geq b^{g^2,\gamma_2}_{T-2}(D)$. We continue this backward procedure for any finite number of steps till $t=0$.

The proof is complete.
\end{proof}

Proposition \ref{pr:ag1bg2} shows that regardless of what drivers and what level of acceptability one chooses for buying and selling side, the ask price will be greater than the bid price.
In particular, when the same family of drivers $g$ is chosen for both trading sides, then $a^{g,\gamma_1}_t(D)\geq b^{g,\gamma_2}_t(D)$, for any $D\in\cD$, $t\in\cT$, $\gamma_1,\gamma_2>0$.

Next result shows that the bid-ask spread increases  when acceptability level is increased.

\begin{proposition}\label{pr:spread} For any $\gamma_2\geq\gamma_1>0$, $t\in\cT$, and any $D\in\cD$,
  $$
  a^{g,\gamma_1}_t(D)\leq a^{g,\gamma_2}_t(D), \quad b^{g,\gamma_1}_t(D)\geq b^{g,\gamma_2}_t(D).
  $$
\end{proposition}

\begin{proof}
It is sufficient to note that
$$
\{a\in\cF_t:\alpha^g_t(\delta_t(a)-\delta^+_t(D))\geq\gamma_2\}\subseteq\{a\in\cF_t:\alpha^g_t(\delta_t(a)-\delta^+_t(D))\geq\gamma_1\},
$$
for any $\gamma_2\geq \gamma_1$.  Using the definition of acceptability ask and bid prices, the result follows at once.

\end{proof}

Suppose that two counterparties A and B are looking to make a trade on a cashflow $D$ at time $t$, such as, A is willing to sell $\varphi\in L^\infty_+(\cF_t)$ share of $D$, and B wants to buy $\varphi$ share of $D$. Assume that both parties are using acceptability pricing theory. Namely,
party A will use the family of drivers $g^1$, and level $\gamma_1$, to calculate his ask price, and party B will use $g^2$, and level $\gamma_2$, to calculate her bid price.
Clearly the trade will happen only if B's bid price meets A's ask price.
Note that Proposition~\ref{pr:ag1bg2} guarantees only that $a^{g^1,\gamma_1}_t(\varphi,D)\geq b^{g^2,\gamma_2}_t(\varphi,D)$, and hence, it is important to investigate under which conditions $a^{g^1,\gamma_1}_t(\varphi,D)=b^{g^2,\gamma_2}_t(\varphi,D)$.
Not to our surprise, this question has a close connection to linear pricing theory. As shown in the next result, in order for bid and ask prices to coincide, the drivers (and hence the prices) have to be locally linear.

\begin{proposition}\label{pr:aequalb2}
  Let $A\in\cF_t$, $g^1$ and $g^2$ be two families of drivers, $\gamma_1$, $\gamma_2>0$.
  Then,
   $$
   \1_Aa^{g^1,\gamma_1}_t(\varphi,D)=\1_Ab^{g^2,\gamma_2}_t(\varphi,D)
   $$
   if and only if there exists a driver $\tilde{g}(t,z)$ such that $\tilde{g}(t,\cdot)$ is linear, and
  \begin{align}
    \1_A\cE_{g^1_{\gamma_1}}[\lambda\sum^T_{s=t+1}D_s|\cF_t] &=\1_A\cE_{\tilde{g}}[\lambda\sum^T_{s=t+1}D_s|\cF_t]; \nonumber\\
    \1_A\cE_{g^2_{\gamma_2}}[-\lambda\sum^T_{s=t+1}D_s|\cF_t] & =\1_A\cE_{\tilde{g}}[-\lambda\sum^T_{s=t+1}D_s|\cF_t], \label{eq:bidEqualAsk}
  \end{align}
  for any $0\leq\lambda\leq\varphi$.
\end{proposition}

\begin{proof}
  ($\Longleftarrow$) If there exists a driver $\tilde{g}(t,z)$ such that $\tilde{g}$ is linear in $z$, and $\1_A\cE_{g^1_{\gamma_1}}[\lambda\sum^T_{s=t+1}D_s|\cF_t]=\1_A\cE_{\tilde{g}}[\lambda\sum^T_{s=t+1}D_s|\cF_t]$, $\1_A\cE_{g^2_{\gamma_2}}[-\lambda\sum^T_{s=t+1}D_s|\cF_t]=\1_A\cE_{\tilde{g}}[-\lambda\sum^T_{s=t+1}D_s|\cF_t]$ for $0\leq\lambda\leq\varphi$, then, by takeing $\lambda=\varphi$, we have
  \begin{align*}
    &\1_Aa^{g^1,\gamma_1}_t(\varphi,D)=\1_A\cE_{g^1_{\gamma_1}}\Big[\varphi\sum^T_{s=t+1}D_s\Big|\cF_t\Big]=\1_A\cE_{\tilde{g}}\Big[\varphi\sum^T_{s=t+1}D_s\Big|\cF_t\Big];\\
    &\1_Ab^{g^2,\gamma_2}_t(\varphi,D)=-\1_A\cE_{g^2_{\gamma_2}}\Big[-\varphi\sum^T_{s=t+1}D_s\Big|\cF_t\Big]=\1_A\cE_{\tilde{g}}\Big[\varphi\sum^T_{s=t+1}D_s\Big|\cF_t\Big],
  \end{align*}
  and hence, $\1_Aa^{g^1,\gamma_1}_t(\varphi,D)=\1_Ab^{g^2,\gamma_2}_t(\varphi,D)$.

\smallskip \noindent
  ($\Longrightarrow$) Assume that $\1_Aa^{g^1,\gamma_1}_t(\varphi,D)=\1_Ab^{g^2,\gamma_2}_t(\varphi,D)$. Then, for $0\leq\lambda\leq\varphi$, there exists $0\leq\lambda'\leq1$ such that $\lambda=\lambda'\varphi$. By Theorem~\ref{th:askbid1}.P4, we get
  $$
  \1_Aa^{g^1,\gamma_1}_t(\lambda,D)\leq\1_A\lambda'a^{g^1,\gamma_1}_t(\varphi,D)=\1_A\lambda'b^{g^2,\gamma_2}_t(\varphi,D)\leq\1_Ab^{g^2,\gamma_2}_t(\lambda,D),
  $$
  however, in view of Proposition \ref{pr:ag1bg2}, we have that $\1_Aa^{g^1,\gamma_1}_t(\lambda,D)\geq\1_Ab^{g^2,\gamma_2}_t(\lambda,D)$.
  Therefore,
  \begin{equation}\label{eq:aequalb2}
  \1_Aa^{g^1,\gamma_1}_t(\lambda,D)=\1_A\lambda'a^{g^1,\gamma_1}_t(\varphi,D)=\1_A\lambda'b^{g^2,\gamma_2}_t(\varphi,D)=\1_Ab^{g^2,\gamma_2}_t(\lambda,D),
  \end{equation}
  Let $(\cE_{g^1_{\gamma_1}}[\1_A\varphi\sum^T_{u=t+1}D_u|\cF_s],\widetilde{Z}_s,\widetilde{M}_s)$, $t\leq s\leq T$ be the solution of BS$\Delta$E corresponding to driver $g_\gamma$ and terminal condition $\1_A\varphi\sum^T_{s=t+1}D_s$. Let us define $x^t_s$ for $t+1\leq s\leq T$ as
\begin{equation}\label{eq:xs}
x_s=\left\{
  \begin{array}{l l}
    \frac{g^1_{\gamma_1}(s,\widetilde{Z}_s)}{\widetilde{Z}_s} & \quad \text{if $\widetilde{Z}_s\neq 0$ }\\
    0 & \quad \text{if $\widetilde{Z}_s= 0$}.
  \end{array} \right.
\end{equation}

  Next, we define $\tilde{g}(s,z)=x_sz$ for $t+1\leq s\leq T$, and $z\in \bR$.
  We will show that $\tilde{g}$ is the desired driver.

First, let us show that $\tilde g$ satisfies Assumption~A.
According to Proposition~\ref{pr:predictable}, $x_s$ defined in \eqref{eq:xs} is $\cF_{s-1}$-measurable, and thus $\tilde{g}(s,z)$ is $\cF_{s-1}$-measurable for any $z\in \bR,$ and so it satisfies A1.
   Since $g^1_{\gamma_1}$ satisfies assumption A2,  then $|x_s|=\frac{|g^1_{\gamma_1}(s,\widetilde{Z}_s)|}{|\widetilde{Z}_s|}\leq c^1_{\gamma_1}(s)$ on the set $\set{\widetilde{Z}_s\neq 0}$, where $c^1_{\gamma_1}(s)$ is the Lipschitz coefficient of $g^1_{\gamma_1}$, for  $s\in\{t+1,\ldots,T\}$. Of course, $|x_s|=0\leq c^1_{\gamma_1}(s)$ on the complement of $\set{\widetilde{Z}_s\neq 0}$, for $s\in\{t+1,\ldots,T\}$.
  Thus, $\tilde{g}$ satisfies A2. Clearly, $\tilde{g}$ satisfies A3, and thus it satisfies Assumption A.

Next, we will show that the identities \eqref{eq:bidEqualAsk} are fulfilled.
   By the construction of $\tilde{g}$, we have that $\1_A\cE_{g^1_{\gamma_1}}[\varphi\sum^T_{s=t+1}D_s|\cF_t]=\1_A\cE_{\tilde{g}}[\varphi\sum^T_{s=t+1}D_s|\cF_t]$, and thus, for $0\leq\lambda\leq\varphi$, with $\lambda=\lambda'\varphi$, we get
  \begin{align*}
    \1_A\cE_{\tilde{g}}\Big[\lambda\sum^T_{s=t+1}D_s\Big|\cF_t\Big]&=\1_A\lambda'\cE_{\tilde{g}}\Big[\varphi\sum^T_{s=t+1}D_s\Big|\cF_t\Big]
    =\1_A\lambda'\cE_{g^1_{\gamma_1}}\Big[\varphi\sum^T_{s=t+1}D_s\Big|\cF_t\Big]\\
    &=\1_A\cE_{g^1_{\gamma_1}}\Big[\lambda\sum^T_{s=t+1}D_s\Big|\cF_t\Big],
  \end{align*}
  where the last equality holds because of \eqref{eq:aequalb2}.
 Second identity from \eqref{eq:bidEqualAsk} is proved similarly.

  This concludes the proof.

\end{proof}

\begin{remark}
  Proposition \ref{pr:aequalb2} implies  that $a^{g^1,\gamma_1}_t(\varphi,D)= b^{g^2,\gamma_2}_t(\varphi,D)$ if and only if $a^{g^1,\gamma_1}_t(\lambda, D)= b^{g^2,\gamma_2}_t(\lambda, D)$, for any $0\leq\lambda\leq\varphi$.
  In other words, if two counterparties agree on the prices for $\varphi$ shares, then they will also agree on prices for any smaller (positive) number of shares $\lambda\varphi$.
\end{remark}

To conclude this section, we show that if $\alpha^{g^1,\gamma_1}_t(D)= b^{g^2,\gamma_2}_t(D)$ for all $D\in\cD$ and $t\in\cT$, then $g^1_{\gamma_1}$ and $g^2_{\gamma_2}$ have to be equal and linear.
This is one reason why the results in Proposition~\ref{pr:aequalb2} hold true only locally.

\begin{proposition}\label{pr:aequalb3}
  Let $g^1$ and $g^2$ be two families of drivers, and $\gamma_1, \gamma_2>0$.
  \begin{enumerate}
    \item[(i)]
    Assume that $a^{g^1,\gamma_1}_t(D)= b^{g^2,\gamma_2}_t(D)$ for any $D\in\cD$, and for a fixed $t\in\cT$.
    Then, $\cE_{g^1_{\gamma_1}}\big[\ \cdot\ |\cF_t]$=$\cE_{g^2_{\gamma_2}}\big[\ \cdot\ |\cF_t]$. Moreover, in this case the functional $\cE_{g^1_{\gamma_1}}\big[\ \cdot\ |\cF_t]$ is linear.
     \item[(ii)]
     Assume that $a^{g^1,\gamma_1}_t(D)= b^{g^2,\gamma_2}_t(D)$ for any $D\in\cD$, and any  $t\in\cT$. Then, there exists a driver $\widetilde{g}(t,z)$ such that $\widetilde{g}(t,\cdot)$ is linear, and $g^1_{\gamma_1}(t,z)=g^2_{\gamma_2}(t,z)=\widetilde{g}(t,z)$ for any $t\in\cT$, $z\in\bR$.
  \end{enumerate}
\end{proposition}

\begin{proof}
  Due to the assumption and the representations of bid/ask prices, we have
  \begin{equation}\label{eq:aequalb3}
  \cE_{g^1_{\gamma_1}}\Big[\sum^T_{s=t+1}D_s\Big|\cF_t\Big]=-\cE_{g^2_{\gamma_2}}\Big[-\sum^T_{s=t+1}D_s\Big|\cF_t\Big], \quad D\in\cD.
  \end{equation}
  Clearly, \eqref{eq:aequalb3} is also true for $-D$
  and therefore
  $$
  \cE_{g^2_{\gamma_2}}\Big[\sum^T_{s=t+1}D_s\Big|\cF_t\Big]=-\cE_{g^1_{\gamma_1}}\Big[-\sum^T_{s=t+1}D_s\Big|\cF_t\Big]\leq\cE_{g^1_{\gamma_1}}\Big[\sum^T_{s=t+1}D_s\Big|\cF_t\Big]=-\cE_{g^2_{\gamma_2}}\Big[-\sum^T_{s=t+1}D_s\Big|\cF_t\Big].
  $$
  Since $-\cE_{g^2_{\gamma_2}}[-\sum^T_{s=t+1}D_s|\cF_t]\leq\cE_{g^2_{\gamma_2}}[\sum^T_{s=t+1}D_s|\cF_t]$, then
  \begin{equation}\label{eq:aequalb4}
  \cE_{g^2_{\gamma_2}}\Big[\sum^T_{s=t+1}D_s\Big|\cF_t\Big]=-\cE_{g^1_{\gamma_1}}\Big[-\sum^T_{s=t+1}D_s\Big|\cF_t\Big]=\cE_{g^1_{\gamma_1}}\Big[\sum^T_{s=t+1}D_s\Big|\cF_t\Big]=-\cE_{g^2_{\gamma_2}}\Big[-\sum^T_{s=t+1}D_s\Big|\cF_t\Big]
  \end{equation}
  for any $D\in\cD$.

 In view of \eqref{eq:aequalb4}, we have that $a^{g^i,\gamma_i}_t(\varphi,D)=a^{g^i,\gamma_i}_t(1,\varphi D)=b^{g^i,\gamma_i}_t(1,\varphi D)=b^{g^i,\gamma_i}_t(\varphi,D)$, $i=1,2$, and thus, by Proposition \ref{pr:aequalb2}, we obtain that
  \begin{equation}\label{eq:aequalb5}
    \cE_{g^i_{\gamma_i}}\Big[\varphi\sum^T_{s=t+1}D_s\Big|\cF_t\Big]=\varphi\cE_{g^i_{\gamma_i}}\Big[\sum^T_{s=t+1}D_s\Big|\cF_t\Big], \quad i=1,2,
  \end{equation}
for any $\varphi\in L^\infty_+(\cF_t)$, $D\in\cD$.
Moreover, by \eqref{eq:aequalb4} again, equation \eqref{eq:aequalb5} is also true for $\varphi\in L^\infty_-(\cF_t)$.
  Hence, \eqref{eq:aequalb5} is true for any $\varphi\in L^\infty(\cF_t)$.

  To proceed, let $D^1, D^2\in\cD$, $a,b\in L^\infty(\cF_t)$. Then by convexity of $g$-expectation and \eqref{eq:aequalb5}, we have that
  \begin{align*}
    \cE_{g^i_{\gamma_i}}\Big[a\sum^T_{s=t+1}D^1_s+b\sum^T_{s=t+1}D^2_s\Big|\cF_t\Big]&\leq\frac{1}{2}\cE_{g^i_{\gamma_i}}\Big[2a\sum^T_{s=t+1}D^1_s\Big|\cF_t\Big]+\frac{1}{2}\cE_{g^i_{\gamma_i}}\Big[2b\sum^T_{s=t+1}D^2_s\Big|\cF_t\Big]\\
    &=a\cE_{g^i_{\gamma_i}}\Big[\sum^T_{s=t+1}D^1_s\Big|\cF_t\Big]+b\cE_{g^i_{\gamma_i}}\Big[\sum^T_{s=t+1}D^2_s\Big|\cF_t\Big].
  \end{align*}
  Due to \eqref{eq:aequalb4}, it is also true that
  $$
  \cE_{g^i_{\gamma_i}}\Big[a\sum^T_{s=t+1}D^1_s+b\sum^T_{s=t+1}D^2_s\Big|\cF_t\Big]\geq a\cE_{g^i_{\gamma_i}}\Big[\sum^T_{s=t+1}D^1_s\Big|\cF_t\Big]+b\cE_{g^i_{\gamma_i}}\Big[\sum^T_{s=t+1}D^2_s\Big|\cF_t\Big],
  $$
  and consequently, in view of the above, we deduce that
  \begin{equation}\label{eq:aequalb6}
  \cE_{g^i_{\gamma_i}}\Big[a\sum^T_{s=t+1}D^1_s+b\sum^T_{s=t+1}D^2_s\Big|\cF_t\Big]= a\cE_{g^i_{\gamma_i}}\Big[\sum^T_{s=t+1}D^1_s\Big|\cF_t\Big]+b\cE_{g^i_{\gamma_i}}\Big[\sum^T_{s=t+1}D^2_s\Big|\cF_t\Big].
  \end{equation}
  Thus, \eqref{eq:aequalb4}, \eqref{eq:aequalb5} and \eqref{eq:aequalb6} imply that $\cE_{g^1_{\gamma_1}}[\ \cdot\ |\cF_t]=\cE_{g^2_{\gamma_2}}[\ \cdot\ |\cF_t]$, and they are linear.

  If $a^{g^1,\gamma_1}_t(D)=b^{g^2,\gamma_2}_t(D)$ for any $D\in\cD$, $t\in\cT$, then $\cE_{g^1_{\gamma_1}}[\ \cdot\ |\cF_t]=\cE_{g^2_{\gamma_2}}[\ \cdot\ |\cF_t]$ and they are linear for any $t\in\cT$.
  By Proposition \ref{pr:lineargexp}, there exists a linear driver $\widetilde{g}(t,z)$ such that $g^1_{\gamma_1}(t,z)=g^2_{\gamma_2}(t,z)=\widetilde{g}(t,z)$, $t\in\cT$, $z\in\bR$.

This concludes the proof.
\end{proof}

\subsection{Market Models}\label{se:marketmodel}

In this section, we will consider some market models that follow the set-up introduced in Section~\ref{se:marketsetup}.
Recall that a market is denoted by a triple $(\cM,P^{\ask},P^{\bid})$, where $\cM$ is the subspace of $\cD$ that consists of processes $D^{\ask/\bid,i}$, $i=1,\ldots,K$, as in (M1), and $D^0=(0,\ldots,0,1)$ which is the dividend process of the banking account.
The functionals $P^{\ask}$ and $P^{\bid}$ allow to compute the ex-dividend prices of the cashflow $\widetilde{D}\in\cM$.
We will define $P^{\ask}$ and $P^{\bid}$ by using pricing operators introduced in Section~\ref{sec:PricingOperators}.

\subsubsection{Ask and Bid Prices Computed at the Same Acceptability Level}\label{se:samelevel}

Let $g$ be a family of drivers that satisfies Assumption G, and let $\gamma>0$.
We consider the market model $(\cM,a^{g,\gamma},b^{g,\gamma})$. Namely, we put
$$
P^{\ask}_t(\varphi,\widetilde{D})=a^{g,\gamma}_t(\varphi,\widetilde{D}), \quad P^{\bid}_t(\varphi,\widetilde{D})=b^{g,\gamma}_t(\varphi,\widetilde{D}),
$$
for any cashflow $\widetilde{D}\in\cM$, and $\varphi\in L^\infty_+(\cF_t)$.

Since $a^{g,\gamma}, b^{g,\gamma}:L^\infty_+(\cF_t)\times\cD\rightarrow L^2(\cF_t)$, for any $t\in\cT$, then such market is well-defined.
Moreover, due to Theorem~\ref{th:askbid1}, we also have the following representation
\begin{align*}
P^{\ask}_t(\varphi,\widetilde{D})&=\cE_{g_\gamma}\Big[\varphi\sum^T_{s=t+1}\widetilde{D}_s\Big|\cF_t\Big],\\
P^{\bid}_t(\varphi,\widetilde{D})&=-\cE_{g_\gamma}\Big[-\varphi\sum^T_{s=t+1}\widetilde{D}_s\Big|\cF_t\Big].
\end{align*}
According to Proposition~\ref{pr:gexp4}, we have that $P^{\ask}_t(0,\widetilde{D})=\cE_{g_\gamma}[0|\cF_t]=0$, and $P^{\bid}_t(0,\widetilde{D})=-\cE_{g_\gamma}[0|\cF_t]=0$.
Hence, the market satisfies assumption (M2).
In particular, for $D^0=(0,\ldots,0,1)$, it is clear that
\begin{align*}
  P^{\ask}_t(\varphi, D^0)&=a^{g,\gamma}_t(\varphi, D^0)=\cE_{g_\gamma}[\varphi|\cF_t]=\varphi,\\
  P^{\bid}_t(\varphi, D^0)&=b^{g,\gamma}_t(\varphi, D^0)=-\cE_{g_\gamma}[-\varphi|\cF_t]=\varphi,
\end{align*}
for any $\varphi\in L^\infty_+(\cF_t)$, which implies that the market satisfies (M3).

Next, we will show that our market satisfies the following important properties, proved over the course of two theorems:

\begin{enumerate}
  \item[(M4)]   The market is arbitrage-free.
  \item[(M5)] For any $\widetilde{D}\in\cM$, $\lambda\in L^\infty(\cF_t)$, and $0\leq\lambda\leq1$, $t\in\cT$,
\begin{align*}
P^{\ask}_t(\lambda\varphi^1+(1-\lambda)\varphi^2,\widetilde{D}) & \leq\lambda P^{\ask}_t(\varphi^1,\widetilde{D})+(1-\lambda)P^{\ask}_t(\varphi^2,\widetilde{D}), \\ P^{\bid}_t(\lambda\varphi^1+(1-\lambda)\varphi^2,\widetilde{D}) & \geq\lambda P^{\bid}_t(\varphi^1,\widetilde{D})+(1-\lambda)P^{\bid}_t(\varphi^2,\widetilde{D}).
  \end{align*}
  \item[(M6)] For any $\widetilde{D}\in\cM$, $\varphi^1, \lambda\in L^\infty_+(\cF_t)$, $\varphi^2\in L^\infty_-(\cF_t)$, $0\leq\lambda\leq1$, $i=1,\ldots,K$, and $t\in\cT$,
  $$
  \lambda P^{\ask}_t(\varphi^1,\widetilde{D})-(1-\lambda) P^{\bid}_t(-\varphi^2,\widetilde{D})\geq\1_{\vartheta\geq0}P^{\ask}_t(\vartheta,\widetilde{D})-\1_{\vartheta<0}P^{\bid}_t(-\vartheta,\widetilde{D}),
  $$
  where $\vartheta=\lambda\varphi^1+(1-\lambda)\varphi^2$.

\end{enumerate}
Properties (M5) and (M6) imply that diversification in the trading is favored.
If an investor nets his purchasing and selling in a convex way, then the cost of trading will be reduced.

Now we proceed by showing that such market is arbitrage free.

\begin{theorem}\label{th:arb2}
  The market model $(\cM,a^{g,\gamma},b^{g,\gamma})$, $\gamma>0$, is arbitrage free.
\end{theorem}

\begin{proof}
   Assume, the market admits an arbitrage, so that, according to  Proposition \ref{pr:eqarb},  there is a trading strategy $\phi\in\cS(0, \gamma):=\cS(0, a^{g,\gamma},b^{g,\gamma})$, such that  $V_T(\phi)\geq0$ and $\bP(V_T(\phi)>0)>0$. We will show that this leads to a  contradiction.

Note that
\begin{align*}
  V_T(\phi)=&\phi^0_T+\sum^K_{i=1}\Big(\phi^{l,i}_TD^{\text{ask},i}_T-\phi^{s,i}_TD^{\text{bid},i}_T\Big)\\
  =&\phi^0_1+\sum^T_{s=2}\Delta\phi^0_s+\sum^K_{i=1}\Big(\phi^{l,i}_TD^{\text{ask},i}_T-\phi^{s,i}_TD^{\text{bid},i}_T\Big).
\end{align*}
Since $\phi$ is self-financing, then we have
\begin{align*}
  \Delta\phi^0_u=&\sum^K_{i=1}\Big(\phi^{l,i}_{u-1}D^{\text{ask},i}_{u-1}-\phi^{s,i}_{u-1}D^{\text{bid},i}_{u-1}\Big)-\sum^K_{i=1}\Big(\1_{\Delta\phi^{l,i}_u\geq0}a^{g,\gamma}_{u-1}(\Delta\phi^{l,i}_u,D^{\ask,i})-1_{\Delta\phi^{l,i}_u<0}b^{g,\gamma}_{u-1}(-\Delta\phi^{l,i}_u,D^{\ask,i})\\
  &-\1_{\Delta\phi^{s,i}_u\geq0}b^{g,\gamma}_{u-1}(\Delta\phi^{s,i}_u,D^{\bid,i})+\1_{\Delta\phi^{s,i}_u<0}a^{g,\gamma}_{u-1}(\Delta\phi^{s,i}_u,D^{\bid,i})\Big),\ \quad u=2,\ldots,T.
\end{align*}
For convenience, we use the following notations:
\begin{align}\label{psi}
  \psi^0_t&=\phi^0_t,\ \quad t=1,\ldots,T,\nonumber \\
  \psi^i_t&=\phi^{l,i}_t,\ \quad i=1,\ldots,K,\ t=1,\ldots,T,\\
  \psi^j_t&=-\phi^{s,j-K}_t,\ \quad j=K+1,\ldots,2K,\ t=1,\ldots,T. \nonumber
\end{align}
Also, let us define $\widehat{D}$ as
\begin{align*}
  \widehat{D}^0&=D^0,\\
  \widehat{D}^i&=D^{\ask,i},\ \quad i=1,\ldots,K,\\
  \widehat{D}^j&=D^{\bid,j-K},\ \quad j=K+1,\ldots,2K.
\end{align*}
Then, we have that
\begin{align*}
  \Delta\psi^0_t=&\sum^{2K}_{i=1}\psi^i_{t-1}\widehat{D}^i_{t-1}-\sum^{2K}_{i=1}\cE_{g_\gamma}\Big[\Delta\psi^i_t\sum^T_{s=t}\widehat{D}^i_s\Big|\cF_{t-1}\Big],\ \quad t=2,\ldots,T,\\
  V_T(\phi)=&\psi^0_T+\sum^{2K}_{i=1}\psi^i_T\widehat{D}^i_T=\psi^0_1+\sum^T_{s=2}\Delta\psi^0_s+\sum^{2K}_{i=1}\psi^i_T\widehat{D}^i_T\\
  =&\psi^0_1+\sum^T_{s=2}\Big(\sum^{2K}_{i=1}\psi^i_{s-1}\widehat{D}^i_{s-1}-\sum^{2K}_{i=1}\cE_{g_\gamma}\Big[\Delta\psi^i_s\sum^T_{u=s}\widehat{D}^i_u\Big|\cF_{s-1}\Big]\Big)+\sum^{2K}_{i=1}\psi^i_T\widehat{D}^i_T\\
  =&\psi^0_1+\sum^T_{s=2}\Big(\sum^{2K}_{i=1}\psi^i_{s-1}\widehat{D}^i_{s-1}-\sum^{2K}_{i=1}\cE_{g_\gamma}\Big[\Delta\psi^i_s\sum^T_{u=s}\widehat{D}^i_u\Big|\cF_{s-1}\Big]\Big)+\sum^{2K}_{i=1}(\psi^i_1+\sum^T_{s=2}\Delta\psi^i_s)\widehat{D}^i_T.
\end{align*}
By multiplying both sides by $\frac{1}{2KT}$, and applying the conditional $g$-expectation $\cE_{g_\gamma}[\ \cdot\ |\cF_{T-1}]$ to both sides, in view of Proposition~\ref{pr:gexp4}.(v), we deduce
\begin{align}\label{TC1}
  \cE_{g_\gamma}\Big[\frac{1}{2KT}V_T(\phi)\Big|\cF_{T-1}\Big]=&\frac{1}{2TK}\psi^0_1+\frac{1}{2TK}\sum^T_{s=2}\sum^{2K}_{i=1}\psi^i_{s-1}\widehat{D}^i_{s-1}-\frac{1}{2TK}\sum^T_{s=2}\sum^{2K}_{i=1}\cE_{g_\gamma}\Big[\Delta\psi^i_s\sum^T_{u=s}\widehat{D}^i_u\Big|\cF_{s-1}\Big]\\
  &+\cE_{g_\gamma}\Big[\frac{1}{2KT}\sum^{2K}_{i=1}(\psi^i_1+\sum^T_{s=2}\Delta\psi^i_s)\widehat{D}^i_T\Big|\cF_{T-1}\Big] \nonumber \\
  \leq&\frac{1}{2KT}\psi^0_1+\frac{1}{2KT}\sum^T_{s=2}\sum^{2K}_{i=1}\psi^i_{s-1}\widehat{D}^i_{s-1}-\frac{1}{2KT}\sum^T_{s=2}\sum^{2K}_{i=1}\cE_{g_\gamma}\Big[\Delta\psi^i_s\sum^T_{u=s}\widehat{D}^i_u\Big|\cF_{s-1}\Big] \nonumber\\
  &+\frac{1}{2KT}\sum^{2K}_{i=1}\cE_{g_\gamma}\Big[\psi^i_1\widehat{D}^i_T\Big|\cF_{T-1}\Big]+\frac{1}{2KT}\sum^T_{s=2}\sum^{2K}_{i=1}\cE_{g_\gamma}\Big[\Delta\psi^i_s\widehat{D}^i_T\Big|\cF_{T-1}\Big]\nonumber \\
  =&\frac{1}{2KT}\psi^0_1+\frac{1}{2KT}\sum^T_{s=2}\sum^{2K}_{i=1}\psi^i_{s-1}\widehat{D}^i_{s-1}-\frac{1}{2KT}\sum^{T-1}_{s=2}\sum^{2K}_{i=1}\cE_{g_\gamma}\Big[\Delta\psi^i_s\sum^T_{u=s}\widehat{D}^i_u\Big|\cF_{s-1}\Big]\nonumber \\
  &+\frac{1}{2KT}\sum^{2K}_{i=1}\cE_{g_\gamma}\Big[\psi^i_1\widehat{D}^i_T\Big|\cF_{T-1}\Big]+\frac{1}{2KT}\sum^{T-1}_{s=2}\sum^{2K}_{i=1}\cE_{g_\gamma}\Big[\Delta\psi^i_s\widehat{D}^i_T\Big|\cF_{T-1}\Big].\nonumber
\end{align}
Since $V_T(\phi)\geq0$, $\bP(V_T(\phi)>0)>0$, then by Proposition~\ref{pr:gexp4}.(ii) applied to the left hand side of \eqref{TC1}, we have that
\[\cE_{g_\gamma}\Big[\frac{1}{2KT}V_T(\phi)\Big|\cF_{T-1}\Big]\geq 0,\quad \textrm{and}\quad \bP(\cE_{g_\gamma}\Big[\frac{1}{2KT}V_T(\phi)\Big|\cF_{T-1}\Big]>0)>0.\]
Consequently,
\begin{align*}
&\psi^0_1+\sum^{T-1}_{s=2}\sum^{2K}_{i=1}\psi^i_{s-1}\widehat{D}^i_{s-1}-\sum^{T-1}_{s=2}\sum^{2K}_{i=1}\cE_{g_\gamma}\Big[\Delta\psi^i_s\sum^T_{u=s}\widehat{D}^i_u\Big|\cF_{s-1}\Big]\\
&+\sum^{2K}_{i=1}\cE_{g_\gamma}\Big[\psi^i_1\sum^T_{u=T-1}\widehat{D}^i_u\Big|\cF_{T-1}\Big]+\sum^{T-1}_{s=2}\sum^{2K}_{i=1}\cE_{g_\gamma}\Big[\Delta\psi^i_s\sum^T_{u=T-1}\widehat{D}^i_u\Big|\cF_{T-1}\Big]\geq0,
\end{align*}
and the strict inequality holds on some set with positive probability.
Let us use the notation
\begin{align}
  \Pi_t:=&\psi^0_1+\sum^t_{s=2}\sum^{2K}_{i=1}\psi^i_{s-1}\widehat{D}^i_{s-1}-\sum^t_{s=2}\sum^{2K}_{i=1}\cE_{g_\gamma}\Big[\Delta\psi^i_s\sum^T_{u=s}\widehat{D}^i_u\Big|\cF_{s-1}\Big] \nonumber\\
&+\sum^{2K}_{i=1}\cE_{g_\gamma}\Big[\psi^i_1\sum^T_{u=t}\widehat{D}^i_s\Big|\cF_t\Big]+\sum^t_{s=2}\sum^{2K}_{i=1}\cE_{g_\gamma}\Big[\Delta\psi^i_s\sum^T_{u=t}\widehat{D}^i_u\Big|\cF_t\Big],\label{eq:pit}
\end{align}
where $t\in\{1,\ldots,T-1\}$.
We just showed that $\Pi_{T-1}\geq0$ and $\bP(\Pi_{T-1}>0)>0$.
Next, we are going to prove that $\Pi_t\geq0$ and $\bP(\Pi_t>0)>0$ will imply that $\Pi_{t-1}\geq0$, $\bP(\Pi_{t-1}>0)>0$, for any $t\in\{2,\ldots,T-1\}$.

Using Proposition~\ref{pr:gexp4}.(iii)-(iv), we obtain that
\begin{align*}
  \cE_{g_\gamma}\Big[\frac{1}{2Kt}\Pi_t\Big|\cF_{t-1}\Big]=&\frac{1}{2Kt}\Big(\psi^0_1+\sum^t_{s=2}\sum^{2K}_{i=1}\psi^i_{s-1}\widehat{D}^i_{s-1}-\sum^t_{s=2}\sum^{2K}_{i=1}\cE_{g_\gamma}\Big[\Delta\psi^i_s\sum^T_{u=s}\widehat{D}^i_u\Big|\cF_{s-1}\Big]\Big)\\
  &+\cE_{g_\gamma}\Big[\frac{1}{2Kt}\sum^{2K}_{i=1}\cE_{g_\gamma}\Big[\psi^i_1\sum^T_{u=t}\widehat{D}^i_u\Big|\cF_t\Big]+\frac{1}{2Kt}\sum^t_{s=2}\sum^{2K}_{i=1}\cE_{g_\gamma}\Big[\Delta\psi^i_s\sum^T_{u=t}\widehat{D}^i_u\Big|\cF_t\Big]\Big|\cF_{t-1}\Big]\\
  \leq&\frac{1}{2Kt}\Big(\psi^0_1+\sum^t_{s=2}\sum^{2K}_{i=1}\psi^i_{s-1}\widehat{D}^i_{s-1}-\sum^t_{s=2}\sum^{2K}_{i=1}\cE_{g_\gamma}\Big[\Delta\psi^i_s\sum^T_{u=s}\widehat{D}^i_u\Big|\cF_{s-1}\Big]\Big)\\
  &+\frac{1}{2Kt}\sum^{2K}_{i=1}\cE_{g_\gamma}\Big[\psi^i_1\sum^T_{u=t}\widehat{D}^i_u\Big|\cF_{t-1}\Big]+\frac{1}{2Kt}\sum^t_{s=2}\sum^{2K}_{i=1}\cE_{g_\gamma}\Big[\Delta\psi^i_s\sum^T_{u=t}\widehat{D}^i_u\Big|\cF_{t-1}\Big]\\
  =&\frac{1}{2Kt}\Big(\psi^0_1+\sum^t_{s=2}\sum^{2K}_{i=1}\psi^i_{s-1}\widehat{D}^i_{s-1}-\sum^{t-1}_{s=2}\sum^{2K}_{i=1}\cE_{g_\gamma}\Big[\Delta\psi^i_s\sum^T_{u=s}\widehat{D}^i_u\Big|\cF_{s-1}\Big]\Big)\\
  &+\frac{1}{2Kt}\sum^{2K}_{i=1}\cE_{g_\gamma}\Big[\psi^i_1\sum^T_{u=t}\widehat{D}^i_u\Big|\cF_{t-1}\Big]+\frac{1}{2Kt}\sum^{t-1}_{s=2}\sum^{2K}_{i=1}\cE_{g_\gamma}\Big[\Delta\psi^i_s\sum^T_{u=t}\widehat{D}^i_u\Big|\cF_{t-1}\Big]\\
  =&\frac{1}{2Kt}\Big(\psi^0_1+\sum^{t-1}_{s=2}\sum^{2K}_{i=1}\psi^i_{s-1}\widehat{D}^i_{s-1}-\sum^{t-1}_{s=2}\sum^{2K}_{i=1}\cE_{g_\gamma}\Big[\Delta\psi^i_s\sum^T_{u=s}\widehat{D}^i_u\Big|\cF_{s-1}\Big]\Big)\\
  &+\frac{1}{2Kt}\sum^{2K}_{i=1}\cE_{g_\gamma}\Big[\psi^i_1\sum^T_{u=t-1}\widehat{D}^i_u\Big|\cF_{t-1}\Big]+\frac{1}{2Kt}\sum^{t-1}_{s=2}\sum^{2K}_{i=1}\cE_{g_\gamma}\Big[\Delta\psi^i_s\sum^T_{u=t-1}\widehat{D}^i_u\Big|\cF_{t-1}\Big]\\
  =&\frac{1}{2Kt}\Pi_{t-1},
\end{align*}
for $t\in\{2,\ldots,T-1\}$.
Since $\Pi_t\geq0$ and $\bP(\Pi_t>0)>0$, then according to Proposition~\ref{pr:gexp4}.(ii), we get that $\cE_{g_\gamma}\Big[\frac{1}{2Kt}\Pi_t\Big|\cF_{t-1}\Big]$ and $\bP(\cE_{g_\gamma}\Big[\frac{1}{2Kt}\Pi_t\Big|\cF_{t-1}\Big]>0)>0$, and thus $\Pi_{t-1}\geq0$, $\bP(\Pi_{t-1}>0)>0$.
Hence, by backward induction, it is true that
\begin{align*}
  \Pi_1=\psi^0_1+\sum^{2K}_{i=1}\cE_{g_\gamma}\Big[\psi^i_1\sum^{T}_{s=1}\widehat{D}^i_s\Big|\cF_1\Big]\geq0,
\end{align*}
and $\bP(\Pi_1>0)>0$. Consequently, in view of \eqref{psi} and the representation for  $a^{g,\gamma}$ and $b^{g,\gamma}$ (cf. Theorem \ref{th:askbid1}) we obtain that
%Multiply both sides by $\frac{1}{2K}$ and apply $\cE_{g_\gamma}[\ \cdot\ ]$ to both sides, we will have that
\begin{align*}
  0\leq&\frac{1}{2K}\psi^0_1+\frac{1}{2K}\sum^{2K}_{i=1}\cE_{g_\gamma}\Big[\psi^i_1\sum^{T}_{s=1}\widehat{D}^i_s\Big]=\frac{1}{2K}\phi^0_1+\frac{1}{2K}\sum^K_{i=1}\Big(a^{g,\gamma}_0(\phi^{l,i}_1,D^{\ask,i})-b^{g,\gamma}_0(-\phi^{s,i}_1,D^{\bid,i})\Big),
\end{align*}
and
\begin{align}\label{TC2}
  \bP\Big(\phi^0_1+\sum^K_{i=1}\Big(a^{g,\gamma}_0(\phi^{l,i}_1,D^{\ask,i})-b^{g,\gamma}_0(-\phi^{s,i}_1,D^{\bid,i})\Big)>0\Big)>0.
\end{align}
Since $\phi$ is a self-financing, then
$$
\psi^0_1=-\sum^K_{i=1}\Big(a^{g,\gamma}_0(\phi^{l,i}_1,D^{\ask,i})-b^{g,\gamma}_0(-\phi^{s,i}_1,D^{\bid,i})\Big).
$$
Thus, \eqref{TC2} means that  $\bP(0>0)>0$, which is a contradiction.
\end{proof}

\smallskip

\begin{proposition}
  Assume that $(\cM,P^{\ask},P^{\bid})=(\cM,a^{g,\gamma},b^{g,\gamma})$, $\gamma>0$. Then, properties (M5) and (M6) hold true.
\end{proposition}

\begin{proof}
Due to Theorem~\ref{th:askbid1}.P3, we have that $P^{\ask}_t(\lambda\varphi^1+(1-\lambda)\varphi^2,\widetilde{D})\leq\lambda P^{\ask}_t(\varphi^1,\widetilde{D})+(1-\lambda)P^{\ask}_t(\varphi^2,\widetilde{D})$, and $P^{\bid}_t(\lambda\varphi^1+(1-\lambda)\varphi^2,\widetilde{D})\geq\lambda P^{\bid}_t(\varphi^1,\widetilde{D})+(1-\lambda)P^{\bid}_t(\varphi^2,\widetilde{D})$, for any $\widetilde{D}\in\cM$, $\lambda\in L^\infty(\cF_t)$, $0\leq\lambda\leq1$, $t\in\cT$,
which implies that condition (M5) is satisfied.

We are left to show that (M6) holds.
  In view of Theorem~\ref{th:askbid1}.P1, we have that
  \begin{align*}
  \lambda P^{\ask}_t(\varphi^1,\widetilde{D})-(1-\lambda) P^{\bid}_t(-\varphi^2,\widetilde{D})
  &=\lambda\cE_{g_\gamma}\Big[\varphi^1\sum^T_{s=t+1}\widetilde{D}_s\Big|\cF_t\Big]+(1-\lambda)\cE_{g_\gamma}\Big[\varphi^2\sum^T_{s=t+1}\widetilde{D}_s\Big|\cF_t\Big]\\
  &\geq\cE_{g_\gamma}\Big[\vartheta\sum^T_{s=t+1}\widetilde{D}_s\Big|\cF_t\Big]
  =\1_{\vartheta\geq0}P^{\ask}_t(\vartheta,\widetilde{D})-\1_{\vartheta<0}P^{\bid}_t(-\vartheta,\widetilde{D}).
  \end{align*}
  This concludes the proof.
\end{proof}

In Theorem \ref{th:arb2}, we showed that if we take a family of drivers $g=(g_x)_{x>0}$, and if we choose the same acceptability level $\gamma>0$ to define the ask and bid prices, then the market model using such prices is  valid and arbitrage free. In next section, using the results from the current section, we will prove that market model is still arbitrage free even we choose different acceptability level for the two trading sides.

\subsubsection{Ask and Bid Prices Computed at Different Acceptability Levels}

Let $g$ be a family of drivers that satisfies Assumption G.
Also, let $\gamma_1, \gamma_2>0$.
We consider the market model $(\cM,a^{g,\gamma_1},b^{g,\gamma_2})$.
That is $P^{\ask}_t(\varphi,\widetilde{D})=a^{g,\gamma_1}_t(\varphi,\widetilde{D})$ and $P^{\bid}_t(\varphi,\widetilde{D})=b^{g,\gamma_2}_t(\varphi,\widetilde{D})$, for $\widetilde{D}\in\cM$, $\varphi\in L^\infty_+(\cF_t)$.

Similarly as in Section~\ref{se:samelevel}, it is not hard to verify here that the market model of the present section satisfies properties (M2), (M3) and (M4), and we leave the verification of these properties to the reader.  At this time we are unable to verify that property (M6) holds for this market. However, we can verify that property (M5) holds.

\begin{theorem}
  The market model $(\cM,a^{g,\gamma_1},b^{g,\gamma_2})$, $\gamma_1, \gamma_2>0$, is arbitrage free.
\end{theorem}
%\ti{stopped here on 5/12}

\begin{proof}
  First, we consider the case  $\gamma_1\leq\gamma_2$. We will prove the result by contradiction.
  Namely, we will assume that market model $(\cM,a^{g,\gamma_1},b^{g,\gamma_2})$ admits an arbitrage, and we will conclude that market model $(\cM,a^{g,\gamma_1},b^{g,\gamma_1})$ also admits an arbitrage, which is impossible in view of Theorem~\ref{th:arb2}.
  Intuitively, the statement and its proof is clear: since $b^{g,\gamma_1}\geq b^{g,\gamma_2}$,  one will trade at higher bid prices in the new market $(\cM,a^{g,\gamma_1},b^{g,\gamma_1})$, and hence it is enough to trade in this market all assets but banking account according to the arbitrage strategy from  $(\cM,a^{g,\gamma_1},b^{g,\gamma_2})$. Finally, the banking account is set up such that the trading strategy remains self-financing in $(\cM,a^{g,\gamma_1},b^{g,\gamma_1})$.

 Let us assume that the market model $(\cM,a^{g,\gamma_1},b^{g,\gamma_2})$ admits an arbitrage opportunity, which, according to  Proposition~\ref{pr:eqarb}, means that  there is a trading strategy $\phi\in\cS(0,a^{g,\gamma_1},b^{g,\gamma_2})$, such that  $V_T(\phi)\geq0$ and $\bP(V_T(\phi)>0)>0$.  Using $\phi$, we will construct an arbitrage strategy $\psi\in{\cS}(0,a^{g,\gamma_1},b^{g,\gamma_1})$.
 Specifically, we set
  \begin{align*}
 \psi^0_1&=  -\sum^K_{i=1}\Big(a^{g,\gamma_1}_0(\phi^{l,i}_1,D^{\ask,i})-b^{g,\gamma_1}_0(\phi^{s,i}_1,D^{\bid,i})\Big),\\[0.05in]
  \psi^{l/s,i}_1&=\phi^{l/s,i}_1,\ \quad i=1,\ldots,K,
   \end{align*}
  and
  \begin{align*}
  \psi^0_t&=\psi^0_1+\sum^t_{u=2}\zeta_u,\ \quad t=2,\ldots,T,\\
  \psi^i_t&=\phi^i_t,\ \quad i=1,\ldots,K,\ t=2,\ldots,T,
  \end{align*}
   where
  \begin{align*}
    \zeta^0_t=&\sum^K_{i=1}\Big(\phi^{l,i}_{t-1}D^{\text{ask},i}_{t-1}-\phi^{s,i}_{t-1}D^{\text{bid},i}_{t-1}\Big)-\sum^K_{i=1}\Big(\1_{\Delta\phi^{l,i}_t\geq0}a^{g,\gamma_1}_{t-1}(\Delta\phi^{l,i}_t,D^{\ask,i})-\1_{\Delta\phi^{l,i}_t<0}b^{g,\gamma_1}_{t-1}(-\Delta\phi^{l,i}_t,D^{\ask,i})\\
    &-\1_{\Delta\phi^{s,i}_t\geq0}b^{g,\gamma_1}_{t-1}(\Delta\phi^{s,i}_t,D^{\bid,i})+1_{\Delta\phi^{s,i}_t<0}a^{g,\gamma_1}_{t-1}(-\Delta\phi^{s,i}_t,D^{\bid,i})\Big),\ \quad t=2,\ldots,T.
  \end{align*}

  First, we will show that $\psi\in{\cS}(0,a^{g,\gamma_1},b^{g,\gamma_1})$.
  We have that\footnote{Here, we are using the convention that ${V}^{\gamma_1}$ is computed relative to the model $(\cM,a^{g,\gamma_1},b^{g,\gamma_1})$, and that $V^{\gamma_1,\gamma_2}$ is computed relative to the model $(\cM,a^{g,\gamma_1},b^{g,\gamma_2})$.}
  \begin{align*}
   {V}^{\gamma_1}_0(\psi)=&\Delta\psi^0_1+\sum^K_{i=1}\Big(\1_{\Delta\psi^{l,i}_1\geq0}P^{\ask}_0(\Delta\psi^{l,i}_1,D^{\ask,i})-\1_{\Delta\psi^{l,i}_1<0}P^{\bid}_0(-\Delta\psi^{l,i}_1,D^{\ask,i})\\
    &-\1_{\Delta\psi^{s,i}_1\geq0}P^{\bid}_0(\Delta\psi^{s,i}_1,D^{\bid})+\1_{\Delta\psi^{s,i}_1<0}P^{\ask}_0(-\Delta\psi^{s,i}_1,D^{\bid})\Big)\\
    =&-\sum^K_{i=1}\Big(a^{g,\gamma_1}_0(\phi^{l,i}_1,D^{\ask,i})-b^{g,\gamma_1}_0(\phi^{s,i}_1,D^{\bid,i})\Big)+\sum^K_{i=1}\Big(a^{g,\gamma_1}_0(\phi^{l,i}_1,D^{\ask,i})-b^{g,\gamma_1}_0(\phi^{s,i}_1,D^{\bid,i})\Big)\\
    =&0=\sum^K_{i=1}\Big(\psi^{l,i}_1D^{\ask,i}_0-\psi^{s,i}_1D^{\bid,i}_0\Big),
  \end{align*}
  and
  \begin{align*}
    \Delta\psi^0_{t+1}&+\sum^K_{i=1}\Big(\1_{\Delta\psi^{l,i}_{t+1}\geq0}P^{\ask}_t(\Delta\psi^{l,i}_{t+1},D^{\ask,i})-\1_{\Delta\psi^{l,i}_{t+1}<0}P^{\bid}_t(-\Delta\psi^{l,i}_{t+1},D^{\ask,i})\\
    &-\1_{\Delta\psi^{s,i}_{t+1}\geq0}P^{\bid}_t(\Delta\psi^{s,i}_{t+1},D^{\bid})+\1_{\Delta\psi^{s,i}_{t+1}<0}P^{\ask}_t(-\Delta\psi^{s,i}_{t+1},D^{\bid})\Big)\\
    =&\zeta^0_{t+1}+\sum^K_{i=1}\Big(\1_{\Delta\phi^{l,i}_s\geq0}a^{g,\gamma_1}_{s-1}(\Delta\phi^{l,i}_s,D^{\ask,i})-\1_{\Delta\phi^{l,i}_s<0}b^{g,\gamma_1}_{s-1}(-\Delta\phi^{l,i}_s,D^{\ask,i})\\
    &-\1_{\Delta\phi^{s,i}_s\geq0}b^{g,\gamma_1}_{s-1}(\Delta\phi^{s,i}_s,D^{\bid,i})+1_{\Delta\phi^{s,i}_s<0}a^{g,\gamma_1}_{s-1}(-\Delta\phi^{s,i}_s,D^{\bid,i})\Big)\\
    =&\sum^K_{i=1}(\psi^{l,i}_tD^{\text{ask},i}_t-\psi^{s,i}_tD^{\text{bid},i}_t),
  \end{align*}
  for every $t=2,\ldots,T$.
  Hence, $\psi$ is a self-financing trading strategy and $\psi\in \cS(0, a^{g,\gamma_1}, b^{g, \gamma_1})$.

  Next we will show that $\psi^0_1\geq\phi^0_1$ and $\Delta\psi^0_t\geq\Delta\phi^0_t$, $t\in\{2,\ldots, T\}$, which will imply that
  $$
  V^{\gamma_1}_T(\psi)=\psi^0_1+\sum^T_{t=2}\psi^0_t\geq\phi^0_1+\sum^T_{t=2}\phi^0_t=V^{\gamma_1,\gamma_2}_T(\phi).
  $$
Consequently, we have that $V^{\gamma_1}_T(\psi)\geq0$ and $\bP(V^{\gamma_1}_T(\psi)>0)>0$, and thus $\psi$ is an arbitrage opportunity in market model $(\cM,a^{g,\gamma_1},b^{g,\gamma_1})$, which contradicts Theorem~\ref{th:arb2}.

  Since $\gamma_1\leq\gamma_2$, by Proposition~\ref{pr:spread}, we get that $b^{g,\gamma_1}_t(\varphi,D)\geq b^{g,\gamma_2}_t(\varphi,D)$, for any $\varphi\in L^\infty_+(\cF_t)$ and $D\in\cD$.
  Hence,
  \begin{align*}
    \psi^0_1&=-\sum^K_{i=1}\Big(a^{g,\gamma_1}_0(\phi^{l,i}_1,D^{\ask,i})-b^{g,\gamma_1}_0(\phi^{s,i}_1,D^{\bid,i})\Big)\\
    &\geq-\sum^K_{i=1}\Big(a^{g,\gamma_1}_0(\phi^{l,i}_1,D^{\ask,i})-b^{g,\gamma_2}_0(\phi^{s,i}_1,D^{\bid,i})\Big)\\
    &=\phi^0_1,
  \end{align*}
  and  it is clear that $\Delta\psi^0_t=\zeta_t$, $t\in\{2,\ldots,T\}$. Moreover,
  \begin{align*}
    \Delta\psi^0_t=&\sum^K_{i=1}\Big(\phi^{l,i}_{t-1}D^{\text{ask},i}_{t-1}-\phi^{s,i}_{t-1}D^{\text{bid},i}_{t-1}\Big)-\sum^K_{i=1}\Big(\1_{\Delta\phi^{l,i}_t\geq0}a^{g,\gamma_1}_{t-1}(\Delta\phi^{l,i}_t,D^{\ask,i})-\1_{\Delta\phi^{l,i}_t<0}b^{g,\gamma_1}_{t-1}(-\Delta\phi^{l,i}_t,D^{\ask,i})\\
    &-\1_{\Delta\phi^{s,i}_t\geq0}b^{g,\gamma_1}_{t-1}(\Delta\phi^{s,i}_t,D^{\bid,i})+1_{\Delta\phi^{s,i}_t<0}a^{g,\gamma_1}_{t-1}(-\Delta\phi^{s,i}_t,D^{\bid,i})\Big)\\
    \geq&\sum^K_{i=1}\Big(\phi^{l,i}_{t-1}D^{\text{ask},i}_{t-1}-\phi^{s,i}_{t-1}D^{\text{bid},i}_{t-1}\Big)-\sum^K_{i=1}\Big(\1_{\Delta\phi^{l,i}_t\geq0}a^{g,\gamma_1}_{t-1}(\Delta\phi^{l,i}_t,D^{\ask,i})-\1_{\Delta\phi^{l,i}_t<0}b^{g,\gamma_2}_{t-1}(-\Delta\phi^{l,i}_t,D^{\ask,i})\\
    &-\1_{\Delta\phi^{s,i}_t\geq0}b^{g,\gamma_2}_{t-1}(\Delta\phi^{s,i}_t,D^{\bid,i})+1_{\Delta\phi^{s,i}_t<0}a^{g,\gamma_1}_{t-1}(-\Delta\phi^{s,i}_t,D^{\bid,i})\Big)\\
    =&\Delta\phi^0_t,
  \end{align*}
  for every $t=2,\ldots,T$.
  Therefore, we have that $\psi^0_T=\psi^0_1+\sum^T_{t=2}\Delta\psi^0_t\geq\phi^0_1+\sum^T_{t=2}\Delta\phi^0_t=\phi^0_T$.

  The proof for $\gamma_1\geq\gamma_2$ is analogues.

\end{proof}

\section{Derivatives Valuation with Dynamic Conic Finance}\label{se:withhedge}

The aim of this section is to build a pricing methodology for general contingent claims, by using the theory of dynamic acceptability indices.
We assume that there exists an underlying market\footnote{Note that in view of previous section  such markets exist.} $(\cM,P^{\ask},P^{\bid})$ that satisfies conditions (M1)-(M6).
We will use the securities from underlying market as hedging instruments.

%We will use the value process generated by securities in the underlying market and some self-financing trading strategy to hedge any derivative that we want to price.
%Then we propose an acceptability method to analyze the hedging error and get the ask and bid prices of the derivative.

We start by introducing the concept of super-hedging in our context.
% and then an extension of $\cH^0(t)$.

\begin{definition}\label{def:sh}
  Assume that  $(\cM,P^{\ask},P^{\bid})$ is the underlying market, and let $V(\phi)$ be the liquidation value process generated by the trading strategy $\phi\in \cS(t)$ for some fixed $t\in\cT$. A cashflow $H\in\cD$, is said to be \textit{super-hedged at zero-cost by $\phi$}, if $H_s=0$, $s\leq t$, and
  \begin{equation}\label{eq:sh}
    \sum^u_{s=t+1}H_s\leq V_u(\phi),\ \quad u=t+1,\ldots,T.
  \end{equation}
\end{definition}

\begin{remark}
  If \eqref{eq:sh} becomes equality, then we say $H$ can be replicated by $\phi$.
  Clearly, if $H\in\cH^0(t)$, then it is hedgeable.
\end{remark}

Based on Definition~\ref{def:sh}, we introduce the set of cashflows that can be super-hedged by strategies $\phi\in\cS(t)$ at zero cost:
\begin{align}
  &\cH(t)  : = \Big\{ \Big(0,\dots, 0, \Delta (V_{t+1}(\phi)-Z_{t+1}), \ldots, \Delta ( V_{T}(\phi) - Z_{T}) \Big) : \phi\in \cS(t), \ Z\in\cL_+(t)\Big\},
\end{align}
where
\begin{align}
  &  \cL_+(t)  := \Big\{ (Z_s)_{s=0}^T  :    Z_s\in L^2_+(\Omega,\cF_s,\bP), \  Z_s=0,\ s\leq t\Big\},
\end{align}
for $t\in\cT$.

We proceed by defining acceptability ask and bid prices of a derivative cashflow $D$.

\begin{definition}\label{def:askbidhedge}
  Let $g=(g_x)_{x>0}$ be a family of drivers that satisfy Assumption~G.
  The \textit{acceptability ask price} of $\varphi\in L^\infty_+(\cF_t)$ shares of the cashflow $D$, at level $\gamma$, at time $t\in\cT$, is defined as
  \begin{equation}\label{eq:a2}
  \widehat{a}^{g,\gamma}_t(\varphi, D)=\essinf\{a\in\cF_t: \exists H\in\cH(t)\ \text{so that}\ \alpha^g_t(\delta_t(a)+H-\delta^+_t(\varphi D))\geq\gamma\},
  \end{equation}
  and the acceptability bid price of $\varphi\geq0,\ \varphi\in L^\infty_+(\cF_t)$, shares of $D$, at level $\gamma$, at time $t\in\cT$ is defined as
  \begin{equation}\label{eq:b2}
  \widehat{b}^{g,\gamma}_t(\varphi, D)=\esssup\{b\in\cF_t: \exists H\in\cH(t)\ \text{so that}\ \alpha^g_t(\delta^+_t(\varphi D)+H-\delta_t(b))\geq\gamma\}.
  \end{equation}
\end{definition}

\begin{remark}
  If $\cH(t)$ is equal to $\{(0,\ldots,0)\}$, which means hedging is not admitted, then
  \begin{align*}
  \widehat{a}^{g,\gamma}_t(\varphi, D)%&=\essinf\{a\in\cF_t: \exists H\in\cH(t)\ \text{s.t.}\ \alpha^g_t(\delta_t(a)+H-\delta^+_t(\varphi D))\geq\gamma\}\\
  &=\essinf\{a\in\cF_t: \alpha^g_t(\delta_t(a)-\delta^+_t(\varphi D))\geq\gamma\}\\
  &=a^{g,\gamma}_t(\varphi,D),
  \end{align*}
  and
  \begin{align*}
    \widehat{b}^{g,\gamma}_t(\varphi, D)%&=\esssup\{b\in\cF_t: \exists H\in\cH(t)\ \text{s.t.}\ \alpha^g_t(\delta^+_t(\varphi D)+H-\delta_t(b))\geq\gamma\}\\
    &=\esssup\{b\in\cF_t: \alpha^g_t(\delta^+_t(\varphi D)-\delta_t(b))\geq\gamma\}\\
    &=b^{g,\gamma}_t(\varphi,D).
  \end{align*}
\end{remark}

\begin{remark}
  Clearly, $\widehat{a}^{g,\gamma}_t(\varphi,D)\leq a^{g,\gamma}_t(\varphi,D)$ and $\widehat{b}^{g,\gamma}_t(\varphi,D)\leq b^{g,\gamma}_t(\varphi,D)$, for $\varphi\in L^\infty_+(\cF_t)$, $D\in\cD$. The prices in the underlying market generated by \eqref{eq:a2} and \eqref{eq:b2} fall in the arbitrage free price interval with the underlying market is defined as in Section~\ref{se:samelevel}. In addition, the ask-bid spread is narrower.
  Moreover, as we will show later, similar results hold true if we just assume that the underlying market model is arbitrage free.
\end{remark}

Similar as in Section~\ref{sec:PricingOperators},  we note that $\widehat{a}^{g,\gamma}_t(\varphi, D)=\widehat{a}^{g,\gamma}_t(1,\varphi D)$ and $\widehat{b}^{g,\gamma}_t(\varphi, D)=\widehat{b}^{g,\gamma}_t(1,\varphi D)$, and thus, we will prove most of the results for $\widehat{a}^{g,\gamma}_t(1,D)$ and $\widehat{b}^{g,\gamma}_t(1,D)$, and the general case follow immediately.

\begin{proposition}\label{pr:abrep2}
  The acceptability ask and bid prices admit the following representations
\begin{equation}
\begin{aligned}
  \widehat{a}^{g,\gamma}_t(D) & =\essinf_{H\in\cH(t)}\cE_{g_\gamma}\left[\sum^T_{s=t+1}(D_s-H_s)\Big|\cF_t\right], \\
  \widehat{b}^{g,\gamma}_t(D) & =\esssup_{H\in\cH(t)}-\cE_{g_\gamma}\left[\sum^T_{s=t+1}(-H_s-D_s)\Big|\cF_t\right], \label{eq:hatBidAsk}
\end{aligned}
\end{equation}
for $D\in\cD$, at level $\gamma>0$, at time $t\in\cT$.
\end{proposition}

\begin{proof}
  We show the proof of first equality;  the proof for the bid price is similar.

  From the definition of $\widehat{a}^{g,\gamma}_t$, we get that
  \begin{align*}
    \widehat{a}^{g,\gamma}_t(D)=\essinf\left\{a\in\cF_t: \exists H\in\cH(t),  \cE_{g_\gamma}\left[\sum^T_{s=t+1}(D_s-H_s)\Big|\cF_t\right]\leq a\right\},
  \end{align*}
and consequently
  \begin{align}\label{eq:ask2}
    \essinf_{H\in\cH(t)}\cE_{g_\gamma}\left[\sum^T_{s=t+1}(D_s-H_s)\Big|\cF_t\right]\leq\essinf\left\{a\in\cF_t: \exists H\in\cH(t), \cE_{g_\gamma}\left[\sum^T_{s=t+1}(D_s-H_s)\Big|\cF_t\right]\leq a\right\}.
  \end{align}
  To prove the converse inequality, we show that strict inequality in \eqref{eq:ask2} does not hold true. Assume that on some set $A\in\cF_t$, $\bP(A)>0$, we have that
  \begin{align*}
    \essinf_{H\in\cH(t)}\cE_{g_\gamma}\left[\sum^T_{s=t+1}(D_s-H_s)\Big|\cF_t\right]<\essinf\left\{a\in\cF_t: \exists H\in\cH(t),  \ \cE_{g_\gamma}\left[\sum^T_{s=t+1}(D_s-H_s)\Big|\cF_t\right]\leq a\right\}.
  \end{align*}
  Then,  there exists an $H'\in\cH$, such that on $A$
  \begin{align*}
    \cE_{g_\gamma}\left[\sum^T_{s=t+1}(D_s-H'_s)\Big|\cF_t\right]<\essinf\left\{a\in\cF_t: \exists H\in\cH(t), \
     \cE_{g_\gamma}\left[\sum^T_{s=t+1}(D_s-H_s)\Big|\cF_t\right]\leq a\right\}.
  \end{align*}
  Consider $b\in\cF_t$, such that on set $A$
  \begin{align}\label{eq:ask2_1}
    \cE_{g_\gamma}\left[\sum^T_{s=t+1}(D_s-H'_s)\Big|\cF_t\right]<b<\essinf\left\{a\in\cF_t: \exists H\in\cH(t), \
     \cE_{g_\gamma}\left[\sum^T_{s=t+1}(D_s-H_s)\Big|\cF_t\right]\leq a\right\}.
  \end{align}
Then, we have that $\1_Ab\in\left\{a\in\cF_t: \exists H\in\cH(t), \   \1_A\cE_{g_\gamma}\left[\sum^T_{s=t+1}(D_s-H_s)\Big|\cF_t\right]\leq a\right\}$.
  Hence, for almost all $\omega\in A$, we get $\widehat{a}^{g,\gamma}_t(D)(\omega)\leq b(\omega)$.
  However, by \eqref{eq:ask2_1}, we also have that  $\widehat{a}^{g,\gamma}_t(D)(\omega)>b(\omega)$ for such $\omega$'s, that leads to a contradiction.
     This concludes the proof.

\end{proof}

Next, we introduce the concept of good-deals for sets of cashflows, which plays an essential role in derivation of fundamental properties of the ask and bid prices.
Towards this end, let $(\rho^\gamma_t)_{x>0}$ be an increasing  family of dynamic risk measures:

\begin{definition}
  \textit{A good-deal} for $\cH(t)$, at time $t\in\cT$, and at level $\gamma>0$,  is a cashflow $H\in\cH(t)$, such that   $\rho^\gamma_t(H)(\omega)<0$, fop $\omega\in A$, for some $A\in\cF_t$, and $\bP(A)>0$. Respectively, we say that the \textit{no-good-deal condition (NGD)} holds true for $\cH(t)$, at time $t\in\cT$, and at level $\gamma>0$, if $\rho^\gamma_t(H)\geq0$ for any $H\in\cH(t)$.
\end{definition}

Similar to no-arbitrage (NA) condition for $\cH^0(t)$, we define the NA condition for the set  $\cH(t)$, that is,  there does not exist $H\in\cH(t)$ such that $\sum^T_{s=t+1}H_s\geq0$ and $\bP(\sum^T_{s=t+1}H_s>0)>0$.
As next result shows, there exists a direct relationship between no-good-deal condition and no-arbitrage condition.

\begin{proposition}\label{prop:NGD-NA}
  Let $g=(g_x)_{x>0}$ be a family of drivers that satisfy Assumption~G, and let $(\rho^{g_x}_t)_{x>0}$ be the corresponding family of dynamic risk measures.
  If NGD holds for $\cH(t)$ at level $\gamma>0$, and $t\in\cT$, then NA also holds for $\cH(t)$.
\end{proposition}

\begin{proof}
  If NGD holds for $\cH(t)$ at level $\gamma>0$, then we have
  \begin{equation}\label{eq:NGD-Expect}
  \cE_{g_\gamma}\big[-\sum^T_{s=t+1}H_s\big|\cF_t\big]\geq0,
  \end{equation}
  for all  $H\in\cH(t)$.
  Assume that there is an arbitrage opportunity at time $t$, i.e. there exists $H'\in\cH(t)$, $A\in\cF_T$, such that $\sum^T_{s=t+1}H'_s\geq0$, $\bP(A)>0$ and $\sum^T_{s=t+1}H'_s(\omega)>0$, $\omega\in A$.
  By Theorem~\ref{th:comp}, we get that
  $$
  \cE_{g_\gamma}\Big[-\sum^T_{s=t+1}H'_s\Big|\cF_t\Big]\leq0,
  $$
  and
  $$
  \cE_{g_\gamma}\Big[-\sum^T_{s=t+1}H'_s\Big|\cF_t\Big](\omega)<0, \quad \omega\in A.
  $$
  which contradicts \eqref{eq:NGD-Expect}.
  This concludes the proof.
\end{proof}

\begin{remark}
  In \cite{BCIR2012}, the authors prove a similar result to Proposition~\ref{prop:NGD-NA} that corresponds to the case of dynamic coherent risk measure, that shows that NA for $\cH(t)$ is equivalent to NGD  at the some level $\gamma>0$.
  Although we are able to show only one implication in our setup,  the notion of NGD still plays an essential role in the study of properties of the proposed acceptability ask and bid prices. In particular, we will show that under NGD condition, the acceptability bid and ask prices satisfy properties similar to those in Theorem~\ref{th:askbid1}.
\end{remark}

Moreover, if we assume that the acceptability prices are computed by using an acceptability level $\gamma>0$, and if the NGD condition does not hold true at this level, then there exists some $H\in\cH(t)$ and $A\in\cF_t$ such that $\bP(A)>0$, and $\rho^{g_\gamma}_t(H)<0$, on $A$.
   Let us consider the ask price of cashflow $(0,\ldots,0)$.
   According to Proposition~\ref{pr:abrep2}, we immediately get that $\widehat{a}^{g,\gamma}_t(0)(\omega)<0$, on $A$, i.e. the price of zero cashflow is not equal to zero, which leads to an inconsistent pricing theory.

\begin{proposition}\label{pr:askbid2}
  Assume that NGD holds for $\cH(t)$ at level $\gamma>0$ and some fixed $t\in\cT$.
  Then, for any $D\in\cD$, we have $\widehat{a}^{g,\gamma}_t(D)\geq \widehat{b}^{g,\gamma}_t(D)$.
\end{proposition}

\begin{proof}
  We prove the statement by contradiction.
  Assume that there exists some $D\in\cD$, $A\in\cF_t$, $\bP(A)>0$, such that $\widehat{b}^{g,\gamma}_t(D)(\omega)>\widehat{a}^{g,\gamma}_t(D)(\omega)$, on $A$.
  Then, by Proposition~\ref{pr:abrep2}, we have that
  \begin{align*}
    \esssup_{H\in\cH(t)}-\cE_{g_\gamma}\Big[\sum^T_{s=t+1}(-H_s-D_s)\Big|\cF_t\Big](\omega)>\essinf_{H\in\cH(t)}\cE_{g_\gamma}\Big[\sum^T_{s=t+1}(D_s-H_s)\Big|\cF_t\Big](\omega),
  \end{align*}
  where $\omega\in A$.
  Let $M = (\widehat{b}^{g,\gamma}_t(D) + \widehat{a}^{g,\gamma}_t(D))/2$.
  Then, there exists $H^1, H^2\in\cH(t)$ such that
  $$
  -\cE_{g_\gamma}\Big[\sum^T_{s=t+1}(-H^1_s-D_s)\Big|\cF_t\Big](\omega)>M(\omega)>\cE_{g_\gamma}\Big[\sum^T_{s=t+1}(D_s-H^2_s)\Big|\cF_t\Big](\omega),
  $$
  for $\omega\in A$.
  Hence, we get that
  \begin{equation}\label{eq:inter2}
  \cE_{g_\gamma}\Big[\sum^T_{s=t+1}(-H^1_s-D_s)\Big|\cF_t\Big](\omega)+\cE_{g_\gamma}\Big[\sum^T_{s=t+1}(D_s-H^2_s)\Big|\cF_t\Big](\omega)<0,\ \omega\in A.
  \end{equation}
  On the other hand, in view of Proposition~\ref{pr:gexp4}.(vi), we have
  \begin{align*}
  &\cE_{g_\gamma}\Big[\frac{1}{2}\Big(\sum^T_{s=t+1}(-H^1_s-D_s)+\sum^T_{s=t+1}(D_s-H^2_s)\Big)\Big|\cF_t\Big]\\
  \leq&\frac{1}{2}\Big(\cE_{g_\gamma}\Big[\sum^T_{s=t+1}(-H^1_s-D_s)\Big|\cF_t\Big]+\cE_{g_\gamma}\Big[\sum^T_{s=t+1}(D_s-H^2_s)\Big|\cF_t\Big]\Big),
  \end{align*}
  which combined with \eqref{eq:inter2}, and using  \eqref{eq:defRhoG}, we obtain
  $$
  \rho^{g,\gamma}_t(\frac{1}{2}(H^1+H^2))(\omega) = \cE_{g_\gamma}\Big[\frac{1}{2}\sum^T_{s=t+1}(-H^1_s-H^2_s)\Big|\cF_t\Big](\omega) <0,
  $$
  for $\omega\in A$.
  However, by Proposition~\ref{pr:convexset}, we have that $\frac{1}{2}(H^1+H^2)\in\cH(t)$, and since NGD is satisfied for all $H\in\cH(t)$, we have that $\rho^{g,\gamma}_t(\frac{1}{2}(H^1+H^2))\geq0$, which leads to a contradiction.
The proof is complete.
\end{proof}

Using the above proposition, we will show that under NGD condition the prices of cashflows from $\cH(t)$ are equal to zero.

\begin{proposition}\label{pr:0price}
 Assume that NGD holds for $\cH(t)$ at level $\gamma>0$, at time $t\in\cT$.
  Then, $\widehat{a}^{g,\gamma}_t(D)=\widehat{b}^{g,\gamma}_t(D)=0$, for any $D\in\cH(t)$.
\end{proposition}

\begin{proof}
  Since $D\in\cH(t)$, then $D-\delta^+_t(D)=0$, and thus
  \begin{align*}
    \widehat{a}^{g,\gamma}_t(D)=\essinf \{a\in\cF_t:\exists H\in\cH(t), \text{ s.t. } \alpha^g_t(\delta_t(a)+H-\delta^+_t(D))\geq\gamma\}\leq0.
  \end{align*}
  Similarly, $\widehat{b}^{g,\gamma}_t(D)\geq 0$, and therefore  $\widehat{a}^{g,\gamma}_t(D)\leq0\leq \widehat{b}^{g,\gamma}_t(D)$.
 On the other hand, by Proposition~\ref{pr:askbid2}, $\widehat{a}^{g,\gamma}_t(D)\geq \widehat{b}^{g,\gamma}_t(D)$.
  Hence, $\widehat{a}^{g,\gamma}_t(D)=\widehat{b}^{g,\gamma}_t(D)=0$.
\end{proof}

Next we will show that pricing the underlying securities $D\in\cM$ by the bid and ask acceptability prices defined by \eqref{eq:b2} and \eqref{eq:a2} yields the market prices $P^{\ask/\bid}$ of these securities.

\begin{proposition}
  Fix $t\in\cT$.
  Assume that NGD holds for $\cH(t)$ at level $\gamma>0$, at time $t\in\cT$. Then,
  \begin{align*}
  \widehat{a}^{g,\gamma}_t(\varphi,\widetilde{D}) & = P^{\ask}_t(\varphi,\widetilde{D}) \\
  \widehat{b}^{g,\gamma}_t(\varphi,\widetilde{D}) & = P^{\bid}_t(\varphi,\widetilde{D}),
   \end{align*} for any $\varphi\in L^\infty_+(\cF_t)$, and $\widetilde{D}\in\cM$.
\end{proposition}

\begin{proof}
  Since $\widetilde{D}\in\cM$, then $\overline{H}:=(0,\ldots,0,\varphi \widetilde{D}_{t+1}-P^{\text{ask}}_t(\varphi,\widetilde{D}), \varphi \widetilde{D}_{t+2},\ldots,\varphi \widetilde{D}_T)\in\cH(t)$. Therefore, by Proposition \ref{pr:0price}, we get $\widehat{a}^{g,\gamma}_t(\overline{H})=0$, which is equivalent to
  $$
  \essinf_{H\in\cH(t)}\cE_{g_\gamma}\Big[\sum^T_{s=t+1}(\varphi \widetilde{D}_s-H_s)-P^{\text{ask}}_t(\varphi,\widetilde{D})\Big|\cF_t\Big]=0.
  $$
  Note that $P^{\text{ask}}_t(\varphi,\widetilde{D})$ is $\cF_t$-measurable, and it does not depend on the argument $H\in\cH(t)$ over which the $\essinf$ is taken.
  Hence, we immediately get that
   $$
  P^{\text{ask}}_t(\varphi,\widetilde{D})=\essinf_{H\in\cH(t)}\cE_{g_\gamma}\Big[\sum^T_{s=t+1}(\varphi \widetilde{D}_s-H_s)\Big|\cF_t\Big]=\widehat{a}^{g,\gamma}_t(\varphi,\widetilde{D}).
  $$
  The proof for the bid price is analogous.
\end{proof}

In \cite{MadanCherny2010} and \cite{BCIR2012}, the authors consider ask-bid prices produced by coherent acceptability indices.
However, in some literature there are arguments that coherent acceptability indices fail to take liquidity risk into account.
In our set-up, acceptability indices are assumed to be quasi-concave and we will show that corresponding ask-bid prices reflect liquidity risk as in the following proposition.

\begin{proposition}\label{pr:convexity}
  The acceptability ask and bid prices satisfy
  \begin{align}
    \widehat{a}^{g,\gamma}_t(\lambda D^1+(1-\lambda)D^2)&\leq\lambda \widehat{a}^{g,\gamma}_t(D^1)+(1-\lambda)\widehat{a}^{g,\gamma}_t(D^2),\\
    \widehat{b}^{g,\gamma}_t(\lambda D^1+(1-\lambda)D^2)&\geq\lambda \widehat{b}^{g,\gamma}_t(D^1)+(1-\lambda)\widehat{b}^{g,\gamma}_t(D^2), \label{eq:bidConcave2}
  \end{align}
  for $D^1, D^2\in\cD$, $\lambda\in\cF_t$, $0\leq\lambda\leq1$, at level $\gamma>0$, at time $t\in\cT$.
\end{proposition}

\begin{proof}
  Due to Proposition~\ref{pr:abrep2}, we have that
  By convexity of $g$-expectation, for any $H^1, H^2\in\cH(t)$, $\lambda\in \cF_t$,  we have that
  \begin{align*}
    \lambda\cE_{g_\gamma}\Big[\sum^T_{s=t+1}(D^1_s-H^1_s)\Big|\cF_t\Big]+ & (1-\lambda)\cE_{g_\gamma}\Big[\sum^T_{s=t+1}(D^2_s-H^2_s)\Big|\cF_t\Big] \\
    %\geq&\cE_{g_\gamma}\Big[\sum^T_{s=t+1}(\lambda D^1_s+(1-\lambda)D^2_s-(\lambda H^1_s+(1-\lambda)H^2_s))\Big|\cF_t\Big]\\
    \geq &\cE_{g_\gamma}\Big[\sum^T_{s=t+1}(\lambda D^1_s+(1-\lambda)D^2_s-H^3_s)\Big|\cF_t\Big],
  \end{align*}
  where $H^3=\lambda H^1+(1-\lambda)H^2$. Due to  convexity of $\cH(t)$ (see Proposition~\ref{pr:convexset}), we have that $H^3\in\cH(t)$.
  Consequently, using  Proposition~\ref{pr:abrep2}, we continue
   \begin{align*}
    \lambda \widehat{a}^{g,\gamma}_t(D^1) & +(1-\lambda)\widehat{a}^{g,\gamma}_t(D^2) \\
    =&\essinf_{H^1, H^2\in\cH(t)}\Big(\lambda\cE_{g_\gamma}\Big[\sum^T_{s=t+1}(D^1_s-H^1_s)\Big|\cF_t\Big]+(1-\lambda)\cE_{g_\gamma}\Big[\sum^T_{s=t+1}(D^2_s-H^2_s)\Big|\cF_t\Big]\Big)\\
    \geq&\essinf_{H^1, H^2\in\cH(t)}\cE_{g_\gamma}\Big[\sum^T_{s=t+1}(\lambda D^1_s+(1-\lambda)D^2_s-H^3_s)\Big|\cF_t\Big]\\
    \geq&\essinf_{H\in\cH(t)}\cE_{g_\gamma}\Big[\sum^T_{s=t+1}(\lambda D^1_s+(1-\lambda)D^2_s-H_s)\Big|\cF_t\Big]\\
    =&\widehat{a}^{g,\gamma}_t(\lambda D^1+(1-\lambda)D^2).
  \end{align*}
  The proof of \eqref{eq:bidConcave2} is similar. This conclude the proof.

   \end{proof}

As an immediate consequence of Proposition~\ref{pr:0price} and Proposition~\ref{pr:convexity} we deduce  the following result about market impact on acceptability ask and bid prices. Namely we show that the acceptability bid and ask prices are not homogenous in number of shares traded - larger number of shares one trades, a higher price per share it will cost.

\begin{corollary}
  Assume that NGD holds for $\cH(t)$ at level $\gamma>0$, $t\in\cT$.
  Then, the acceptability ask and bid prices satisfy the following inequalities
  \begin{align*}
    &\widehat{a}^{g,\gamma}_t(\lambda\varphi,D)\leq\lambda \widehat{a}^{g,\gamma}_t(\varphi,D),\ \widehat{b}^{g,\gamma}_t(\lambda\varphi,D)\geq\lambda \widehat{b}^{g,\gamma}_t(\varphi,D),\ \lambda,\varphi\in L^\infty_+(\cF_t),\ 0\leq\lambda\leq1;\\
    &\widehat{a}^{g,\gamma}_t(\lambda\varphi,D)\geq\lambda \widehat{a}^{g,\gamma}_t(\varphi,D),\ \widehat{b}^{g,\gamma}_t(\lambda\varphi,D)\leq\lambda \widehat{b}^{g,\gamma}_t(\varphi,D),\ \lambda,\varphi\in L^\infty_+(\cF_t),\ \lambda\geq1,
  \end{align*}
  for $\gamma>0$, $t\in\cT$.
\end{corollary}

Finally we move to the central question of this section - the absence of arbitrage in the market model driven by acceptability bid and ask prices as defined above.
Recall that the starting point was the notion of hedgeable cashflows $\cH(t)$ initiated at time $t$ at zero cost and generated by self-financing trading strategies in the underlying market. Similarly, we consider the set of extended cashflows defined next.

\begin{definition}
   The set of \textit{extended cashflows} associated with an $\cF_t$-measurable random variable $S_t$, $t\in\cT$, and a process $D\in\cD$, is defined as
  \begin{equation*}
    \widetilde{\cH}(t,S_t):=\Big\{\Big(0,\ldots,0,S_t,H_{t+1}-\varphi D_{t+1},\ldots,H_T-\varphi D_T\Big): H\in\cH(t), \varphi\in L^2(\cF_t)\Big\}.
  \end{equation*}
\end{definition}

The quantity $S_t$ should be viewed as the set-up cost of $\varphi$ shares of the cashflow $D$ at time $t\in\cT$.
 Consistently with Proposition~\ref{pr:eqarb}, we say that that the acceptability ask price \eqref{eq:a2} is \textit{arbitrage free price}, if there does not exist a cashflow $\widetilde{H}\in\widetilde{\cH}(t,\widehat{a}^{g,\gamma}_t(\varphi,D))$ such that \begin{equation}\label{eq:arbi}
  \sum^T_{s=t}\widetilde{H}_s\geq0,\ \bP(\sum^T_{s=t}\widetilde{H}_s>0)>0.
\end{equation}
Similarly, the bid price \eqref{eq:b2} is arbitrage free price, if there is no $\widetilde{H}\in\widetilde{\cH}(t,-\widehat{b}^{g,\gamma}_t(\varphi,D))$ that  satisfies \eqref{eq:arbi}.

\begin{theorem}
  Assume that NGD holds for $\cH(t)$, $t\in\cT$.
  Then, then acceptability ask and bid prices are arbitrage-free prices.
\end{theorem}
\begin{proof}
  We give the proof for ask prices; the proof for bid prices is similar and we omit it here.

  If there exists $\widetilde{H}\in\widetilde{\cH}(t,\widehat{a}^{g,\gamma}_t(\varphi,D))$ such that \eqref{eq:arbi} is satisfied, then there exists $H'\in\cH(t)$, $A\in\cF_t$ such that
  $$
  \widehat{a}^{g,\gamma}_t(\varphi,D)\geq - \sum^T_{s=t+1}(H'_s-D_s) ,
  $$
  and the inequality is strict on set $A$.

  In view of Proposition~\ref{pr:abrep2}, we immediately get that
  \begin{equation*}
    \cE_{g_\gamma}\Big[\sum^T_{s=t+1}(D_s-H'_s)\Big|\cF_t\Big]\geq\sum^T_{s=t+1}(D_s-H'_s)
  \end{equation*}
  with strict inequality on $A$.
  Consequently, by Theorem~\ref{pr:gexp4}, we conclude that
  \begin{align*}
    \cE_{g_\gamma}\Big[\sum^T_{s=t+1}(D_s-H'_s)\Big](\omega)=&\cE_{g_\gamma}\Big[\cE_{g_\gamma}\Big[\sum^T_{s=t+1}(D_s-H'_s)\Big|\cF_t\Big]\Big](\omega)\\
    >&\cE_{g_\gamma}\Big[\sum^T_{s=t+1}(D_s-H'_s)\Big](\omega),
  \end{align*}
  for $\omega\in A$, with $\bP(A)>0$, that leads to a contradiction.
  Hence, the acceptability ask price $\widehat{a}^{g,\gamma}_t(\varphi,D)$ is an  arbitrage free price.
 The proof is complete.
\end{proof}

Similar to the discussion in Section~\ref{se:marketmodel}, we study the relationship between the acceptability ask and bid prices in the case when these prices are generated by  different families of drivers and and different acceptability levels.
We start with a result similar to Proposition~\ref{pr:spread}.

\begin{proposition}\label{pr:diffgamma}
The acceptability ask and bid prices at time $t\in\cT$ of a cashflow $D\in\cD$, satisfy the following inequalities
$\widehat{a}^{g,\gamma_1}_t(D)\leq \widehat{a}^{g,\gamma_2}_t(D)$ and $\widehat{b}^{g,\gamma_1}_t(D)\geq \widehat{b}^{g,\gamma_2}_t(D)$, for $\gamma_2\geq\gamma_1>0$.
\end{proposition}

\begin{proof}
    Since $g_{\gamma_1}\leq g_{\gamma_2}$, according to Theorem~\ref{th:comp}, we get that
  $$
  \cE_{g_{\gamma_1}}\Big[\sum^T_{s=t+1}(D_s-H_s)\Big|\cF_t\Big]\leq\cE_{g_{\gamma_2}}\Big[\sum^T_{s=t+1}(D_s-H_s)\Big|\cF_t\Big],
  $$
  for any $H\in\cH(t)$.
  Hence, using the representation \eqref{eq:hatBidAsk}, we obtain
  $$
  \widehat{a}^{g,\gamma_1}_t(D)=\essinf_{H\in\cH(t)}\cE_{g_{\gamma_1}}\Big[\sum^T_{s=t+1}(D_s-H_s)\Big|\cF_t\Big]\leq\essinf_{H\in\cH(t)}\cE_{g_{\gamma_2}}\Big[\sum^T_{s=t+1}(D_s-H_s)\Big|\cF_t\Big]\leq \widehat{a}^{g,\gamma_2}_t(D).
  $$
  Analogously, one proofs the corresponding inequality for the bid prices.
\end{proof}

\begin{corollary}
  Assume that NGD holds for $\cH(t)$ at time $t\in\cT$, and  at levels $\gamma_1, \gamma_2>0$.
  Then, $\widehat{a}^{g,\gamma_2}_t(D)\geq \widehat{b}^{g,\gamma_1}_t(D)$, for any $d\in\cD$.
\end{corollary}

\begin{proof}
  If $\gamma_1\geq\gamma_2$, then by Proposition \ref{pr:askbid2} and \ref{pr:diffgamma}, we have that $\widehat{a}^{g,\gamma_2}_t(D)\geq \widehat{b}^{g,\gamma_2}_t(D)\geq \widehat{b}^{g,\gamma_1}_t(D)$.
  Similarly, if $\gamma_1\leq\gamma_2$, then $\widehat{a}^{g,\gamma_2}_t(D)\geq \widehat{a}^{g,\gamma_1}_t(D)\geq \widehat{b}^{g,\gamma_1}_t(D)$.
  This completes the proof.
\end{proof}

Finally we want to mention that we were not able to establish a general result on comparison of acceptability bid and ask prices, similar to Proposition~\ref{pr:ag1bg2}.
Generally speaking we do not know if $\widehat{a}^{g^1,\gamma_1} \geq \widehat{b}^{g^2,\gamma_2}$, for some arbitrary family of drivers $g^1, \ g^2$, and levels $\gamma_1,\gamma_2>0$. In the nutshell, this is due to the lack of an appropriate form of time consistency property for $\widehat{a}^{g,\gamma}$, and  $\widehat{b}^{g,\gamma}$, similar to Property~P5, in Theorem~\ref{th:askbid1}. We leave the answer to this question for further investigations.
Nevertheless, we do have a result that shows that once two couterparties, who may use different acceptability levels and different drivers, find that their prices are such that
 $\widehat{a}^{g^1,\gamma_1}_t(D)\leq \widehat{b}^{g^2,\gamma_2}_t(D)$ for all $D\in\cD$, then the bid and ask prices coincide  and hence the trade will go through.

\begin{proposition}
  Fix $t\in\cT$.
  Let $g^1$ and $g^2$ be two families of drivers.
  Assume that  $\widehat{a}^{g^1,\gamma_1}_t(D)\leq \widehat{b}^{g^2,\gamma_2}_t(D)$, for a fixed $t\in\cT$, and for any $D\in\cD$.
  Then,
  $$
  \widehat{a}^{g^1,\gamma_1}_t(D)=\widehat{b}^{g^1,\gamma_1}_t(D)=\widehat{a}^{g^2,\gamma_2}_t(D)=\widehat{b}^{g^2,\gamma_2}_t(D),
  $$
  for any $D\in\cD$.
\end{proposition}

\begin{proof}
  Since $\widehat{a}^{g^1,\gamma_1}_t(D)\leq \widehat{b}^{g^2,\gamma_2}_t(D)$ for any $D\in\cD$, then we have that
  $$
  \widehat{a}^{g^1,\gamma_1}_t(-D)\leq \widehat{b}^{g^2,\gamma_2}_t(-D),\ \quad D\in\cD.
  $$
  According to Proposition~\ref{pr:abrep2}, it is clear that
  \begin{align*}
    \widehat{a}^{g^1,\gamma_1}_t(D)=-\widehat{b}^{g^1,\gamma_1}_t(-D),\\
    \widehat{a}^{g^2,\gamma_2}_t(D)=-\widehat{b}^{g^2,\gamma_2}_t(-D).
  \end{align*}
  Therefore,  $-\widehat{b}^{g^1,\gamma_1}_t(D)=\widehat{a}^{g^1,\gamma_1}_t(-D)\leq \widehat{b}^{g^2,\gamma_2}_t(-D)=-\widehat{a}^{g^2,\gamma_2}_t(D)$, for any $D\in\cD$.
  Hence, $\widehat{a}^{g^2,\gamma_2}_t(D)\leq\widehat{b}^{g^1,\gamma_1}_t(D)$, and due to Proposition~\ref{pr:askbid2}, we get that
  $$
  \widehat{b}^{g^2,\gamma_2}_t(D)\leq\widehat{a}^{g^2,\gamma_2}_t(D)\leq\widehat{b}^{g^1,\gamma_1}_t(D)\leq\widehat{a}^{g^1,\gamma_1}_t(D),\ \quad D\in\cD.
  $$
  Thus, using our initial assumptions, we have that
  $$
  \widehat{b}^{g^2,\gamma_2}_t(D)=\widehat{a}^{g^2,\gamma_2}_t(D)=\widehat{b}^{g^1,\gamma_1}_t(D)=\widehat{a}^{g^1,\gamma_1}_t(D),\ \quad D\in\cD.
  $$
  This concludes the proof.
\end{proof}

\begin{appendix}
  \section{Appendix}
This section is devoted to some auxiliary technical results, used throughout  the paper.

\begin{proposition}\label{pr:randomwalk}
  Assume that $W$ is a symmetric random walk, such that $P(\Delta W_t=\pm 1)=\frac{1}{2}$, $t=1,\ldots,T$, and let $\{\cF^W_t\}_{t=0}^T$ be the filtration generated by $W$.
  Then, for any square integrable martingale process $X$ on $(\Omega, \cF^W, \{\cF^W_t\}_{t=0}^T, \bP\}$, there exists a predictable process $\xi$, such that
  \begin{align*}
    \Delta X_t=\xi_t\Delta W_t, \quad t=1,\ldots,T.
  \end{align*}
\end{proposition}

\begin{proof}
According to the Galtchouk-Kunita-Watanabe \cite[Theorem 10.18]{FollmerSchiedBook2004} decomposition, $X$ admits the following representation
\begin{align}\label{eq:rw0}
\Delta X_t=\xi_t\Delta W_t+\Delta M_t,
\end{align}
where $\xi$ is a predictable process such that $\xi_t\Delta W_t\in L^2(\cF_t)$, for  $t=1,\ldots,T$, and $M$ is a square-integrable martingale which is orthogonal to $W$ and such that $M_0=0$.
Since $M$ is orthogonal to $W$, and using the fact that $\Delta\langle W\rangle_t=1$, we get
\begin{equation}
\xi_t=\frac{\bE[\Delta X_t\Delta W_t|\cF_{t-1}]}{\Delta\langle W\rangle_t} =  \bE[\Delta X_t\Delta W_t|\cF_{t-1}],\quad t=1,\ldots,T.  \label{eq:rw1}
\end{equation}

Next, we fix $t\in\set{1,\ldots,T}$. Clearly $\Delta X_t = f_t(\Delta W_1, \ldots, \Delta W_t)$, for some function $f$.
Consequently, we have that
\begin{align}
\bE[\Delta X_t & \Delta W_t | \cF_{t-1}]  = \bE [\Delta X_t \1_{\Delta W_t=1} | \cF_{t-1} ]   - \bE[ \Delta X_t \1_{\Delta W_t=-1} |  \cF_{t-1} ] \nonumber \\
 & = f_t( \Delta W_1, \ldots, \Delta W_{t-1}, 1 ) \bE [ \1_{\Delta W_t=1} | \cF_{t-1} ]   - f_t( \Delta W_1, \ldots, \Delta W_{t-1}, -1 ) \bE[ \1_{\Delta W_t=-1} |  \cF_{t-1} ] \nonumber \\
 & = \frac12  f_t( \Delta W_1, \ldots, \Delta W_{t-1}, 1 ) - \frac12 f_t( \Delta W_1, \ldots, \Delta W_{t-1}, -1 ). \label{eq:rw2}
\end{align}
By similar reasoning, since $\bE[\Delta X_t|\cF_{t-1}]=0$, we deduce that
\begin{align}\label{eq:rw3}
 \frac12  f_t( \Delta W_1, \ldots, \Delta W_{t-1}, 1 ) + \frac12 f_t( \Delta W_1, \ldots, \Delta W_{t-1}, -1 ) = 0.
  %\frac{1}{2}\Delta X_t|_{\Delta W_t=1}+\frac{1}{2}\Delta X_t|_{\Delta W_t=-1}=0
\end{align}
Consequently, by \eqref{eq:rw2} and \eqref{eq:rw3}, we have that
\begin{align*}
 1_{\Delta W_t=1}  \xi_t & = 1_{\Delta W_t=1} f_t(\Delta W_1,\ldots, \Delta W_t) = 1_{\Delta W_t=1} \frac{f_t(\Delta W_1,\ldots, \Delta W_t)}{\Delta W_t} = 1_{\Delta W_t=1} \frac{\Delta X_t}{\Delta W_t},   \\
 1_{\Delta W_t= - 1}  \xi_t     & = - 1_{\Delta W_t=-1} f_t(\Delta W_1,\ldots, \Delta W_t)  = 1_{\Delta W_t=-1} \frac{f_t(\Delta W_1,\ldots, \Delta W_t)}{\Delta W_t} = 1_{\Delta W_t=-1} \frac{\Delta X_t}{\Delta W_t}.
\end{align*}
Thus, we conclude that
\begin{align*}
  \xi_t= \1_{\Delta W_t=1}\frac{\Delta X_t}{\Delta W_t}+\1_{\Delta W_t=-1}\frac{\Delta X_t}{\Delta W_t} = \frac{\Delta X_t}{\Delta W_t}.
\end{align*}
From here, and in view of \eqref{eq:rw0},  we get that
$$
\Delta X_t=\Delta X_t+\Delta M_t,
$$
and thus $\Delta M_t=0$, which completes the proof.

\end{proof}

\begin{proposition}\label{pr:convexfunctions}
Assume that the functions $f_1, f_2:\bR\rightarrow\bR$ are convex, and Lipschitz continuous on $\bR$, with Lipschitz constants $c_1,c_2$ respectively. Also suppose that
  $f_1(x)\leq f_2(x), \ x\in\bR$. Then, $c_1\leq c_2$.
\end{proposition}

\begin{proof}
  Denote by $\partial_+f(x)$ and $\partial_-f(x)$ the right and left derivative, respectively, of the function $f:\bR\to\bR$.
Note that, since $f_i(x)$ is a convex function, then $\partial_+f_i(x),\ \partial_-f_i(x), \ x\in\bR$, exist.
Moreover,  $\partial_+f_i(x)$ is increasing,  $\partial_-f_i(x)$ is decreasing, and $\partial_+f_i(x)\geq\partial_-f_i(x)$, $i=1,\ 2$.

We claim that $\partial_+f_i\leq c_i$, $i=1, 2$.
Otherwise, there exists $x_i^0$, such that $\partial_+f_i(x)>c_i$, for all $x\geq x_i^0$.
Therefore, for any $x'_i>x''_i>x_i^0$, we get that
$$
\frac{|f_i(x'_i)-f_i(x''_i)|}{|x'_i-x''_i|}=\frac{f_i(x'_i)-f_i(x''_i)}{x'_i-x''_i}>c_i,
$$
which contradicts to the assumption that $c_i$ is the Lipschitz constant.
Hence, $\partial_+f_i(x)\leq c_i$, which implies that $\lim_{x\to\infty}\partial_+f_i(x)$ exists, and
\begin{align}\label{eq:app4}
  \lim_{x\to\infty}\partial_+f_i(x)\leq c_i, \quad i=1,2.
\end{align}
Similarly, one can prove that
\begin{align}\label{eq:app5}
  \lim_{x\to-\infty}\partial_-f_i(x)\geq-c_i, \quad i=1,2.
\end{align}
By \eqref{eq:app4}, \eqref{eq:app5} and the fact that $\partial_+f_i(x)\geq\partial_-f_i(x)$, $x\in\bR$, we have that
\begin{align}\label{eq:app6}
\max\Big\{\Big|\lim_{x\to-\infty}\partial_-f_i(x)\Big|,\Big|\lim_{x\to\infty}\partial_+f_i(x)\Big|\Big\}\leq c_i, \quad i=1,2.
\end{align}
By convexity of $f_i$, we have that
\begin{align*}
  \partial_-f_i(x^1)\leq\partial_+f_i(x^1)&\leq\frac{f_i(x^1)-f_i(x^2)}{x^1-x^2}\leq\partial_-f_i(x^2)\leq\partial_+f_i(x^2),
\end{align*}
which implies that
\begin{align}\label{eq:app7}
  \max\Big\{\Big|\lim_{x\to-\infty}\partial_-f_i(x)\Big|,\Big|\lim_{x\to\infty}\partial_+f_i(x)\Big|\Big\}\geq\frac{|f_i(x^1)-f_i(x^2)|}{|x^1-x^2|}, \quad -\infty<x^1<x^2<\infty.
\end{align}
Then, in view of \eqref{eq:app6}, and the fact that $c_i$ is the (smallest) Lipschitz constant, we conclude that
$$
\max\Big\{\Big|\lim_{x\to-\infty}\partial_-f_i(x)\Big|,\Big|\lim_{x\to\infty}\partial_+f_i(x)\Big|\Big\}=c_i, \quad i=1,2.
$$
To complete the proof, in view of  \eqref{eq:app8} and \eqref{eq:app9}, it is enough to show that
\begin{align}\label{eq:app8}
  \lim_{x\to\infty}\partial_+f_1(x)\leq\lim_{x\to\infty}\partial_+f_2(x), \\
  \lim_{x\to-\infty}\partial_-f_1(x)\geq\lim_{x\to-\infty}\partial_-f_2(x). \label{eq:app9}
\end{align}
If $\lim_{x\to\infty}\partial_+f_1(x)>\lim_{x\to\infty}\partial_+f_2(x)$. Then, there exists $x^*\in\bR$ and $\delta>0$, such that $\partial_+f_1(x)>\partial_+f_2(x)+\delta$ for any $x\geq x^*$. Let us consider $F(x)=f_1(x)-f_2(x)$. Then,
\begin{align*}
  \partial_+F(x)=\partial_+f_1(x)-\partial_+f_2(x)>\delta>0, \quad x>x^*.
\end{align*}
Hence, for any $x\in\bR$, such that  $x-x^*>-\frac{F(x^*)}{\delta}\geq0$, we get that
$$
F(x)>F(x^*)+\delta(x-x^*)>F(x^*)-F(x^*)=0,
$$
which contradicts the original assumption that $f_1\leq f_2$.
  Therefore,  $c_1\leq c_2$, and the proof is complete.

\end{proof}

  \begin{proposition}\label{pr:predictable} Let $(\Omega, \cF, \bP)$ be a probability space, and let $Z$ be an $\cF$-measurable random variable.
    Assume that the function $g:\Omega\times\bR\to\bR$ is such that, for every $z\in\bR$, the mapping $\omega\mapsto g(\omega,z)$ is $\cF$-measurable, and the mapping $z\mapsto g(\omega,z)$ is continuous  $\bP$-a.s.
    Then, $g(\, \cdot \, ,Z(\, \cdot\,))$ is $\cF$-measurable.
  \end{proposition}

  \begin{proof}
Since $Z$ is $\cF$-measurable, then there exists a sequence of simple $\cF$-measurable random variables $(\mu^n)_{n=1}^\infty$, such that $\lim_{n\rightarrow\infty}\mu^n=Z$ with probability one, and $\mu^n=\sum_{i=1}^{m_n}\1_{A^n_i}\mu^n_i$, where $\mu^n_i\in\bR$, $A^n_i\in\cF$, $\cup^{m_n}_{i=1}A^n_i=\Omega$, $A^n_i\cap A^n_j=\emptyset$, $i\neq j$.
Then, we have that
$$
g(\omega,\mu^n(\omega))=\sum_{i=1}^{m_n}\1_{A^n_i}(\omega)g(\omega,\mu^n_i), \quad n\in\bZ^+.
$$
Then, the measurability of $\omega\mapsto g(\omega,z)$ implies that $g(\omega,\mu^n_i(\omega))$ is $\cF$-measurable, and hence, $g(\omega,\mu^n(\omega))$ is $\cF$-measurable for any $n\in\bZ^+$.
Moreover, since $g(\omega,z)$ is continuous in $z$, we have that $\lim_{n\rightarrow\infty}g(\omega,\mu^n(\omega))=g(\omega,Z(\omega))$ almost surely, and therefore $g(\, \cdot \, ,Z(\, \cdot\,))$ is $\cF$-measurable.
\end{proof}

\begin{proposition}\label{pr:convcombi}
  Assume that the underlying market satisfies conditions (M1)-(M6).
  Then, for any trading strategy $\phi$, $\psi\in\cS(t)$, there exists $\theta\in\cS(t)$ such that
  \begin{equation} \label{eq:app1}
  V_u(\theta)\geq\lambda V_u(\phi)+(1-\lambda)V_u(\psi), \quad \lambda\in L^\infty(\cF_t), \ 0\leq\lambda\leq1, \ u=t+1,\ldots,T.
\end{equation}
\end{proposition}

\begin{proof}
  We define the trading strategy $\theta$ as follows:
  \begin{equation*}
    \theta^{l/s,i}_u=\left\{
    \begin{array}{l l}
    0 & \quad u=0,\ldots,t,\\
    \lambda\phi^{l/s,i}_u+(1-\lambda)\psi^{l/s,i}_u & \quad u=t+1,\ldots,T,
  \end{array} \right.
  \end{equation*}
 for $i=1,\ldots,K$, and
  \begin{equation*}
    \theta^0_u=\left\{
    \begin{array}{ll}
    0 & \quad u=0,\ldots,t,\\
    \theta^0_{t+1}+\sum^u_{r=t+2}\zeta_r, & \quad u=t+1,\ldots,T,
    \end{array}\right.
  \end{equation*}
  where
  \begin{align}
    \theta^0_{t+1}=&-\sum^K_{i=1}\Big(P^{\ask}_t(\theta^{l,i}_{t+1},D^{\ask,i})-P^{\bid}_t(\theta^{s,i}_{t+1},D^{\bid,i})\Big),\label{eq:theta0tp1}\\
    \zeta_u=&\sum^K_{i=1}(\theta^{l,i}_{u-1}D^{\ask,i}_{u-1}-\theta^{s,i}_{u-1}D^{\bid,i}_{u-1})-\sum^K_{i=1}\Big(1_{\Delta\theta^{l,i}_{u}\geq0}P^{\ask}_{u-1}(\Delta\theta^{l,i}_u,D^{\ask,i})-\1_{\Delta\theta^{l,i}_u<0}P^{\bid}_{u-1}(-\Delta\theta^{l,i}_u,D^{\ask,i})\nonumber\\
    &-\1_{\Delta\theta^{s,i}_u\geq0}P^{\bid}_{u-1}(\Delta\theta^{s,i}_u,D^{\bid,i})+\1_{\Delta\theta^{s,i}_u<0}P^{\ask}_{u-1}(-\Delta\theta^{s,i}_u,D^{\bid,i})\Big), \quad u=t+2,\ldots,T.\nonumber
  \end{align}

By straightforward bookkeeping one can show that $\theta$ is a self-financing trading strategy, namely it satisfies \eqref{eq:self-fin}.
Also, due to the construction of $\theta$, clearly, it belongs to $\cS(t)$.

Clearly \eqref{eq:self-fin} is fulfilled for $u=0,\ldots, t$.  For $u=t+1$, we get that
\begin{align*}
  \Delta\theta^0_{t+1}&=\theta^0_{t+1},\\
  \Delta\theta^{l/s,i}_{t+1}&=\theta^{l/s,i}_{t+1}, \quad i=1,\ldots,K.
\end{align*}
Then, it implies that
  \begin{align*}
    \Delta\theta^0_{t+1}+\sum^K_{i=1}\Big(\1_{\Delta\theta^{l,i}_{t+1}\geq0} & P^{\ask}_t(\Delta\theta^{l,i}_{t+1},D^{\ask,i})-\1_{\Delta\theta^{l,i}_{t+1}<0}P^{\bid}_t(-\Delta\theta^{l,i}_{t+1},D^{\ask,i})\\
    -\1_{\Delta\theta^{s,i}_{t+1}\geq0} & P^{\bid}_t(\Delta\theta^{s,i}_{t+1},D^{\bid,i})+\1_{\Delta\theta^{s,i}_{t+1}<0}P^{\ask}_t(-\Delta\theta^{s,i}_{t+1},D^{\bid,i})\Big)\\
    =&-\sum^K_{i=1}\Big(P^{\ask}_t(\theta^{l,i}_{t+1},D^{\ask,i})-P^{\bid}_t(\theta^{s,i}_{t+1},D^{\bid,i})\Big)\\
    &+\sum^K_{i=1}\Big(P^{\ask}_t(\theta^{l,i}_{t+1},D^{\ask,i})-P^{\bid}_t(\theta^{s,i}_{t+1},D^{\bid,i})\Big)\\
    =&0=\sum^K_{i=1}(\theta^{l,i}_tD^{\text{ask},i}_t-\theta^{s,i}_tD^{\text{bid},i}_t).
  \end{align*}
For $u=t+2,\ldots,T$, we have that $\Delta\theta_u=\zeta_u$.
Therefore,
  \begin{align*}
    \Delta\theta^0_u+\sum^K_{i=1}\Big(\1_{\Delta\theta^{l,i}_u\geq0} & P^{\ask}_{u-1}(,\Delta\theta^{l,i}_u,D^{\ask,i})-\1_{\Delta\theta^{l,i}_u<0}P^{\bid}_{u-1}(-\Delta\theta^{l,i}_u,D^{\ask,i})\\
    -\1_{\Delta\theta^{s,i}_u\geq0} & P^{\bid}_{u-1}(\Delta\theta^{s,i}_u,D^{\bid,i})+\1_{\Delta\theta^{s,i}_u<0}P^{\ask}_{u-1}(-\Delta\theta^{s,i}_u,D^{\bid,i})\Big)\\
    %=&\zeta_u+\sum^K_{i=1}\Big(\1_{\Delta\theta^{l,i}_u\geq0}P^{\ask}_{u-1}(\Delta\theta^{l,i}_u,D^{\ask,i})-\1_{\Delta\theta^{l,i}_u<0}P^{\bid}_{u-1}(-\Delta\theta^{l,i}_u,D^{\ask,i})\\
%    &-\1_{\Delta\theta^{s,i}_u\geq0}P^{\bid}_{u-1}(\Delta\theta^{s,i}_u,D^{\bid,i})+\1_{\Delta\theta^{s,i}_u<0}P^{\ask}_{u-1}(-\Delta\theta^{s,i}_u,D^{\bid,i})\Big)\\
    =&\sum^K_{i=1}(\theta^{l,i}_{u-1}D^{\ask,i}_{u-1}-\theta^{s,i}_{u-1}D^{\bid,i}_{u-1}), \quad u=t+2,\ldots,T
  \end{align*}
  Thus, $\theta$ is a self-financing trading strategy.

  %We want to show that $V_u(\theta)\geq\lambda V_u(\phi)+(1-\lambda)V_u(\psi)$ for any $u=t+1,\ldots,T$.
  We are left to prove \eqref{eq:app1}.
  Since $\theta^{l/s,i}_u=\lambda\phi^{l/s,i}_u+(1-\lambda)\psi^{l/s,i}_u$, $u=t+1,\ldots,T$, and in view of (M5), it is sufficient to verify that
  \begin{equation} \label{eq:app2}
  \theta^0_u\geq\lambda\phi^0_u+(1-\lambda)\psi^0_u, \quad u=t+1,\ldots,T,
  \end{equation}
  or equivalently,
  \begin{align}
    \theta^0_{t+1}&\geq\lambda\phi^0_{t+1}+(1-\lambda)\psi^0_{t+1},\label{eq:thetatp1}\\
    \Delta\theta^0_u&\geq\lambda\Delta\phi^0_u+(1-\lambda)\Delta\psi^0_u, \quad u=t+2,\ldots,T \label{eq:dthetau}.
  \end{align}
  Since $\phi, \psi\in\cS(t)$, then
  \begin{align*}
    \phi^0_{t+1}=&-\sum^K_{i=1}\Big(P^{\ask}_t(\phi^{l,i}_{t+1},D^{\ask,i})-P^{\bid}_t(\phi^{s,i}_{t+1},D^{\bid,i})\Big),\\
    \psi^0_{t+1}=&-\sum^K_{i=1}\Big(P^{\ask}_t(\psi^{l,i}_{t+1},D^{\ask,i})-P^{\bid}_t(\psi^{s,i}_{t+1},D^{\bid,i})\Big).
  \end{align*}
  From here, combine the above with \eqref{eq:theta0tp1}, also take into account of (M5).
  Then, \eqref{eq:thetatp1} is an immediate result.

Next, we will prove \eqref{eq:dthetau}.
  Fix $u=t+2,\ldots,T$.
  Due to the above construction, $\Delta\theta^0_u$ satisfies the following identity
  \begin{align*}
  \Delta\theta^0_u=\zeta_u=&\sum^K_{i=1}(\theta^{l,i}_{u-1}D^{\ask,i}_{u-1}-\theta^{s,i}_{u-1}D^{\bid,i}_{u-1})\\
  &-\sum^K_{i=1}\Big(\underbrace{1_{\Delta\theta^{l,i}_{u}\geq0}P^{\ask}_{u-1}(\Delta\theta^{l,i}_u,D^{\ask,i})-\1_{\Delta\theta^{l,i}_u<0}P^{\bid}_{u-1}(-\Delta\theta^{l,i}_u,D^{\ask,i})}_{I^l(\theta,i)}\\
  &\underbrace{-\1_{\Delta\theta^{s,i}_u\geq0}P^{\bid}_{u-1}(\Delta\theta^{s,i}_u,D^{\bid,i})+\1_{\Delta\theta^{s,i}_u<0}P^{\ask}_{u-1}(-\Delta\theta^{s,i}_u,D^{\bid,i})}_{I^s(\theta,i)}\Big).
  \end{align*}
  According to the self-financing condition \eqref{eq:self-fin}, we get that
  \begin{align*}
    \Delta\phi^0_u=&\sum^K_{i=1}(\phi^{l,i}_{u-1}D^{\ask,i}_{u-1}-\phi^{s,i}_{u-1}D^{\bid,i}_{u-1})-\sum^K_{i=1}(I^l(\phi,i)+I^s(\phi,i)),\\
  \Delta\psi^0_u=&\sum^K_{i=1}(\psi^{l,i}_{u-1}D^{\ask,i}_{u-1}-\psi^{s,i}_{u-1}D^{\bid,i}_{u-1})-\sum^K_{i=1}(I^l(\psi,i)+I^s(\psi,i)).
  \end{align*}
  Thus \eqref{eq:dthetau} is equivalent to
  $$
  \sum^K_{i=1}(I(\theta,l,i)+I(\theta,s,i))\leq\lambda\sum^K_{i=1}(I(\phi,l,i)+I(\phi,s,i))+(1-\lambda)\sum^K_{i=1}(I(\psi,l,i)+I(\psi,s,i)),
  $$
  which will follow once we show that
  \begin{align}\label{eq:dtheta}
  I^l(\theta,i)\leq\lambda I^l(\phi,i)+(1-\lambda)I^l(\psi,i),
  \end{align}
  and
  \begin{align}\label{eq:dthetas}
  I^s(\theta,i)\leq\lambda I^s(\phi,i)+(1-\lambda)I^s(\psi,i).
  \end{align}
  We will prove \eqref{eq:dtheta}, and \eqref{eq:dthetas} follows similarly.
  Due to the symmetry in $\phi$ and $\psi$, then without loss of generality, we assume that $\frac{1}{2}\leq\lambda\leq1$.

  For the purpose of showing \eqref{eq:dtheta}, we define the following sets
  \begin{align*}
  A^{i,0}_t&=\{\Delta\phi^{l,i}_t\geq0,\Delta\psi^{l,i}_t\geq0\},\\
  A^{i,1}_t&=\Big\{\Delta\phi^{l,i}_t\geq0,\Delta\psi^{l,i}_t<0,\Delta\phi^{l,i}_t\geq\frac{\lambda-1}{\lambda}\Delta\psi^{l,i}_t\Big\},\\
  A^{i,2}_t&=\Big\{\Delta\phi^{l,i}_t\geq0,\Delta\psi^{l,i}_t<0,\Delta\phi^{l,i}_t<\frac{\lambda-1}{\lambda}\Delta\psi^{l,i}_t\Big\},\\
  A^{i,3}_t&=\Big\{\Delta\phi^{l,i}_t<0,\Delta\psi^{l,i}_t\geq0,\Delta\phi^{l,i}_t\geq\frac{\lambda-1}{\lambda}\Delta\psi^{l,i}_t\Big\},\\
  A^{i,4}_t&=\Big\{\Delta\phi^{l,i}_t<0,\Delta\psi^{l,i}_t\geq0,\Delta\phi^{l,i}_t<\frac{\lambda-1}{\lambda}\Delta\psi^{l,i}_t\Big\},\\
  A^{i,5}_t&=\{\Delta\phi^{l,i}_t<0,\Delta\psi^{l,i}_t<0\},
  \end{align*}
  for any $t\in\cT$.
  Clearly, we have that $A^{i,0}_t\cup A^{i,1}_t\cup A^{i,2}_t=\{\Delta\phi^{l,i}_t\geq0\}$, $A^{i,0}_t\cup A^{i,3}_t\cup A^{i,4}_t=\{\Delta\psi^{l,i}_t\geq0\}$, $A^{i,1}_t\cup A^{i,2}_t\cup A^{i,5}_t=\{\Delta\psi^{l,i}_t<0\}$, and $A^{i,3}_t\cup A^{i,4}_t\cup A^{i,5}_t=\{\Delta\phi^{l,i}_t<0\}$.
  Hence, \eqref{eq:dtheta} can be represented as
  \begin{align*}
    &\lambda\Big((\1_{A^{i,0}_u}+\1_{A^{i,1}_u}+\1_{A^{i,2}_u})P^{\ask}_{u-1}(\Delta\phi^{l,i}_u,D^{\ask,i})-(\1_{A^{i,3}_u}+\1_{A^{i,4}_u}+\1_{A^{i,5}_u})P^{\bid}_{u-1}(-\Delta\phi^{l,i}_u,D^{\bid,i})\Big)\\
    &+(1-\lambda)\Big((\1_{A^{i,0}_u}+\1_{A^{i,3}_u}+\1_{A^{i,4}_u})P^{\ask}_{u-1}(\Delta\psi^{l,i}_u,D^{\ask,i})-(\1_{A^{i,1}_u}+\1_{A^{i,2}_u}+\1_{A^{i,5}_u})P^{\bid}_{u-1}(-\Delta\psi^{l,i}_u,D^{\bid,i})\Big).
  \end{align*}
  In view of (M5), we get that
  \begin{align}\label{eq:ai0}
    \1_{A^{i,0}_u}\Big(\lambda P^{\ask}_{u-1}(\Delta\phi^{l,i}_u,D^{\ask,i})+(1-\lambda)P^{\ask}_{u-1}(\psi^{l,i}_u,D^{\ask,i})\Big)\geq\1_{A^{i,0}_u}P^{\ask}_{u-1}(\Delta\theta^{l,i}_u,D^{\ask,i}),
  \end{align}
  and
  \begin{align}\label{eq:ai5}
    \1_{A^{i,5}_u}\Big(-\lambda P^{\bid,i}_{u-1}(\Delta\phi^{l,i}_u,D^{\ask,i})-(1-\lambda)P^{\bid,i}_{u-1}(-\Delta\psi^{l,i}_u,D^{\ask,i})\Big)\geq-\1_{A^{i,5}_u}P^{\bid}_{u-1}(-\Delta\theta^{l,i}_u,D^{\ask,i}).
  \end{align}
  According to (M6), we have that
  \begin{align}
    &\1_{A^{i,1}_u}\Big(\lambda P^{\ask}_{u-1}(\Delta\phi^{l,i}_u,D^{\ask,i})-(1-\lambda)P^{\bid}_{u-1}(-\Delta\psi^{l,i}_u,D^{\ask,i})\Big)\geq\1_{A^{i,1}_u}P^{\ask}_{u-1}(\Delta\theta^{l,i}_u,D^{\ask,i}),\label{eq:ai1}\\
    &\1_{A^{i,2}_u}\Big(\lambda P^{\ask,i}_{u-1}(\Delta\phi^{l,i}_u,D^{\ask,i})-(1-\lambda)P^{\bid,i}_{u-1}(-\Delta\psi^{l,i}_u,D^{\ask,i})\Big)\geq-\1_{A^{i,2}_u}P^{\bid}_{u-1}(\Delta\theta^{l,i}_u,D^{\ask,i}),\label{eq:ai2}\\
    &\1_{A^{i,3}_u}\Big(-\lambda P^{\bid,i}_{u-1}(-\Delta\phi^{l,i}_u,D^{\ask,i})+(1-\lambda)P^{\ask,i}_{u-1}(\Delta\psi^{l,i}_u,D^{\ask,i})\Big)\geq\1_{A^{i,3}_u}P^{\ask}_{u-1}(\Delta\theta^{l,i}_u,D^{\ask,i}),\label{eq:ai3}\\
    &\1_{A^{i,4}_u}\Big(-\lambda P^{\bid,i}_{u-1}(-\phi^{l,i}_u,D^{\ask,i})+(1-\lambda)P^{\ask,i}_{u-1}(\Delta\psi^{l,i}_u,D^{\ask,i})\Big)\geq-\1_{A^{i,4}_u}P^{\bid}_{u-1}(\Delta\theta^{l,i}_u,D^{\ask,i})\label{eq:ai4}.
  \end{align}
  Summing up inequalities \eqref{eq:ai0}-\eqref{eq:ai4} part by part, we deduce that \eqref{eq:dtheta} is true, which concludes the proof.
\end{proof}

\begin{proposition}\label{pr:convexset}
  Let $\phi\in\cS(t)$, and let $\cL_+(t)$ be defined as
  $$
  \cL_+(t):=\Big\{(Z_s)^T_{s=0}:Z_s\in L^2_+(\Omega,\cF_s,\bP), Z_s=0,\ s\leq t\Big\},
  $$
  for $t=0,\ldots,T-1$.
  Then, for any $t\in\set{0, \ldots, T-1}$, the set
  $$
  \cH(t):=\Big\{\Big(0,\ldots,0,\Delta(V_{t+1}(\phi)-Z_{t+1}),\ldots,\Delta(V_T(\phi)-Z_T)\Big):\phi\in\cS(t),Z\in\cL_+(t)\Big\}
  $$
  is a convex set.
\end{proposition}

\begin{proof}
  Let $H^1$, $H^2\in\cH(t)$, and let $\lambda\in L^\infty(\cF_t)$ such that $0\leq\lambda\leq1$.
  Then, there exists $\phi$, $\psi\in\cS(t)$, $Z^1$, $Z^2\in\cL_+(t)$, such that
  $$
  H^1=\Big(0,\ldots,0,\Delta(V_{t+1}(\phi)-Z^1_{t+1}),\ldots,\Delta(V_T(\phi)-Z^1_T)\Big),
  $$
  and
  $$
  H^2=\Big(0,\ldots,0,\Delta(V_{t+1}(\psi)-Z^2_{t+1}),\ldots,\Delta(V_T(\psi)-Z^2_T)\Big).
  $$
  According to Proposition~\ref{pr:convcombi}, there exists $\theta\in\cS(t)$, such that $\lambda V_s(\phi)+(1-\lambda)V_s(\psi)\leq V_s(\theta)$ for any $s=t+1,\ldots,T$.
  Therefore,
  \begin{align*}
  \lambda\sum^s_{u=t+1}H^1_u+(1-\lambda)\sum^s_{u=t+1}H^2_u&=\lambda V_s(\phi)+(1-\lambda)V_s(\psi)-\lambda Z^1_s-(1-\lambda)Z^2_s\\
  &\leq V_s(\theta)-(\lambda Z^1_s+(1-\lambda)Z^2_s)\leq V_s(\theta),
  \end{align*}
  for any $s=t+1,\ldots,T$.
  By Definition~\ref{def:sh}, we have that $\lambda H^1+(1-\lambda)H^2\in\cH(t)$.
  The proof is complete.
\end{proof}

\end{appendix}

\section*{Acknowledgments}
Tomasz R. Bielecki and Igor Cialenco acknowledge support from the National Science Foundation  grant DMS-0908099, and DMS-1211256.

\bibliographystyle{alpha}
%\bibliography{MathFinanceMaster-10-20-2014}
\newcommand{\etalchar}[1]{$^{#1}$}

\end{document}